\documentclass[sigconf]{acmart}

\usepackage{xspace}
\usepackage[margin=4pt,caption=false]{subfig}
\usepackage{diagbox}
\usepackage{multirow}
\usepackage{bm}
\usepackage{paralist}
\usepackage{enumitem}
\usepackage{capt-of}
\usepackage{algpseudocode}
\usepackage{color, colortbl}
\usepackage[titletoc,title]{appendix}
\usepackage{algorithm2e}
\usepackage{soul}
\SetKwInOut{Parameter}{Parameters}

\DeclareMathOperator*{\argmin}{arg\,min}

\usepackage{tikz}

\definecolor{Gray}{gray}{0.9}

\let\oldReturn\Return
\renewcommand{\Return}{\State\oldReturn}

\newcommand{\good}{\text{\tt GOOD}}
\newcommand{\bad}{\text{\tt BAD}}
\newcommand{\stai}{\text{\tt stnd}_i}
\newcommand{\staij}{\text{\tt stnd}_{ij}}
\newcommand{\stami}{\text{\tt stnd}_{-i}}
\newcommand{\ait}{a_i^t}
\newcommand{\aitars}{a_{i,ars}^t}
\newcommand{\aitmi}{a_i^{t-1}}
\newcommand{\aitmiars}{a_{i,ars}^{t-1}}
\newcommand{\aimitmi}{a_{-i}^{t-1}}
\newcommand{\aimitmiars}{a_{-i,ars}^{t-1}}

\def\algbackskip{\hskip-\ALG@thistlm}

\algblockdefx[Foreach]{Foreach}{EndForeach}[1]{\textbf{foreach} #1 \textbf{do}}{\textbf{end foreach}}
\algnewcommand\algorithmiccase{\textbf{case}}
\algdef{SE}[CASE]{Case}{EndCase}[1]{\algorithmiccase\ #1}{\algorithmicend\ \algorithmiccase}%
\algtext*{EndCase}%


\makeatletter
\newcommand{\removelatexerror}{\let\@latex@error\@gobble}
\makeatother

\allowdisplaybreaks

\setcopyright{acmcopyright}
\begin{document}
\title{An Eye for an Eye: Economics of Retaliation in Mining Pools}

\copyrightyear{2019} 
\acmYear{2019} 
\acmConference[AFT '19]{1st ACM Conference on Advances in Financial Technologies}{October 21--23, 2019}{Zurich, Switzerland}
\acmBooktitle{1st ACM Conference on Advances in Financial Technologies (AFT '19), October 21--23, 2019, Zurich, Switzerland}
\acmPrice{15.00}
\acmDOI{10.1145/3318041.3355472}
\acmISBN{978-1-4503-6732-5/19/10}

\author{Yujin Kwon*, Hyoungshick Kim$^\dagger$, Yung Yi*, Yongdae Kim*}
\affiliation{\vspace{2mm}
\institution{*KAIST}}
\email{{dbwls8724, yiyung, yongdaek}@kaist.ac.kr}
\affiliation{\vspace{2mm} \institution{$^\dagger$Sungkyunkwan University}}
\email{hyoung@skku.edu}

\renewcommand{\shortauthors}{Kwon et al.}

\begin{abstract}
Currently, miners typically join \textit{mining pools} to solve cryptographic puzzles together, and mining pools are in high competition. 
This has led to the development of several attack strategies such as \textit{block withholding} (BWH) and \textit{fork after withholding} (FAW) attacks that can weaken the health of PoW systems and but maximize mining pools' profits.
In this paper, we present strategies called Adaptive Retaliation Strategies (ARS) to mitigate not only BWH attacks but also FAW attacks. 
In ARS, each pool cooperates with other pools in the normal situation, and adaptively executes either FAW or BWH attacks for the purpose of retaliation only when attacked. 
In addition, in order for rational pools to adopt ARS, ARS should strike to an adaptive balance between retaliation and selfishness because the pools consider their payoff even when they retaliate.
We theoretically and numerically show that ARS would not only lead to the induction of a no-attack state among mining pools, but also achieve the adaptive balance between retaliation and selfishness.
\end{abstract}




\maketitle

\section{Introduction}
\label{Introduction}

Numerous cryptocurrencies based on a peer-to-peer network now exist, utilizing an open ledger called a \textit{blockchain}. In such blockchain systems, network nodes called \textit{miners} verify the transactions collected through the network, generate a block consisting of valid transactions, and propagate the block to the entire network. 
In general, (public) blockchain systems offer financial incentives to encourage miners to participate in this process. For example, in Bitcoin, a miner is currently rewarded with 12.5 BTC for a new block creation.

The security and reliability of blockchain systems depend on their consensus algorithm. Most popular cryptocurrency systems such as Bitcoin and Ethereum adopt a \textit{proof-of-work} (PoW) mechanism in order to agree on the same blockchain~\cite{blockchain}. In such PoW mechanisms, miners must solve cryptographic puzzles (i.e., \textit{proofs-of-work}) showing that a certain amount of computational resources (e.g., time and memory) was spent in order to generate a new block. The mining difficulty is adjusted automatically to maintain an average mining rate of one block at a fixed time interval with changes in the total computational power of the blockchain system.

The significant increase in mining difficulty led the formation of \textit{mining pools} in which miners gather to mine together; these pools perform mining as a single node in a network. 
Most pools consist of a manager and miners, and the manager assigns a puzzle to miners at the beginning of every round. 
The miners then find partial proofs-of-work (PPoWs) and full proofs-of-work (FPoWs) for a given puzzle and submit them to the manager.
PPoWs are needed for assessing each miner's contribution to share the reward in the pool.
If a miner fully solves the puzzle, the manager generates and propagates a block based on the submitted PoW. 
The manager then earns the reward for one block and splits it among miners in proportion to their number of submitted PPoWs.

However, previous studies~\cite{rosenfeld2011analysis,kwon2017selfish} demonstrated how existing mining pools' protocols can be vulnerable to \textit{block withholding} (BWH)~\cite{rosenfeld2011analysis} and \textit{fork after withholding} (FAW) attacks~\cite{kwon2017selfish}.
Moreover, pools in high competition can launch these attacks against each other by infiltrating a part of their mining power (i.e., computational power) into other pools.
A BWH attacking pool first infiltrates its mining power into other pools (i.e., victim pools) and submits only PPoWs but not FPoWs to the victim pools. The FAW attack is an extended attack of the BWH attack. 
Similar to the BWH attack, a FAW attacking pool also infiltrates its mining power into other pools, and submits only PPoWs but not FPoWs to the victim pools except for the case when an external honest miner (i.e., neither the attacker nor a miner in the victim pool) propagates a block. 
Unlike the BWH attack, in the case, the attacker \textit{intentionally} generates forks through the victim pools. 

To analyze the strategic interaction between miners for those two attacks, we can use game-theoretic models: the BWH game and the FAW game. 
Eyal~\cite{eyal2015miner} showed that the BWH game between two pools where they can execute BWH attacks is similar to the prisoner's dilemma. 
For the FAW game between two pools where they can execute FAW attacks, 
Kwon et al.~\cite{kwon2017selfish} showed that mining pools' strategies form a Nash equilibrium in which a larger pool among two pools earns extra profit. 
Thus, FAW attacks between pools can cause more centralization of the system because miners would join in large pools that earn the extra reward through the FAW attack. 
\textit{The objective of this study is to find strategies inducing a no-attack state where FAW and BWH attacks do not occur.}

\smallskip
\noindent\textbf{System model:} To achieve this goal, we take a more macroscopic view of the system by modeling long-term interactions between two pools as a notion of repeated game, which we call a \textit{repeated FAW-BWH game}, and by considering both FAW and BWH attacks together in one framework. 
In the repeated game, the \textit{FAW-BWH game} is repeated as a one-stage game over time, where in every stage each rational pool makes a decision to cooperate or execute FAW or BWH attacks (see Section~\ref{Model} for a full description of our model). 
Unlike previous studies~\cite{eyal2015miner, kwon2017selfish} focusing on a single stage game, we model a game with multiple stages to analyze the effects of long-term interactions among pools. 
Also, we considered both FAW and BWH attack strategies while they considered a single attack strategy (i.e., FAW or BWH attack) only.

\smallskip
\noindent\textbf{A novel strategy ARS:}
In game theory, the \textit{iterated prisoner's dilemma} (i.e., a repeated version of the prisoner's dilemma) has been extensively studied in terms of how rational players can cooperate through retaliation~\cite{axelrod1987evolution}.
Using this retaliation concept, we find how rational pools can cooperate in the repeated FAW-BWH game. 
Unlike the iterated prisoner's dilemma, the repeated FAW-BWH game leads to a situation where a larger pool always wins the game (i.e., the pool size game). This occurs when two pools execute FAW attacks against each other as in the \textit{FAW game}. 
Therefore, it is relatively under-explored to find a cooperation-inducing strategy in such a situation. 
In this paper, we propose strategies called Adaptive Retaliation Strategies (ARS), which can lead to the state of no-attack between two pools. 
In ARS, a pool cooperates at first and continues to cooperate, but executes attacks for retaliation only when attacked. 
Here, ARS must achieve a good balance between retaliation and selfishness because rational pools would consider their payoff even when they retaliate. 
In other words, if retaliation is costly, they would not follow ARS, which implies that we should find a \textit{credible retaliation}.
This is done by ARS' adaptive retaliation, i.e., adaptively deciding the amount of infiltration power for retaliation against FAW or BWH attacks.

We formally describe ARS and prove that 1) ARS leads rational pools to cooperate and 2) pools are likely to adopt ARS, by using a popular concept of equilibrium in repeated game theory, called subgame perfect Nash equilibrium (see Section~\ref{sec:cooperation}). 
Furthermore, our numerical analysis demonstrates that ARS makes the FAW and BWH attacks unprofitable (see Section~\ref{sec:numerical}).

\smallskip
\noindent\textbf{Practical requirements for ARS: } To apply ARS in real-world settings, pools should be able to monitor other pools' information about whether an attacker executes FAW or BWH attacks, and if so, how much the attacker's infiltration power is used. 
To monitor these information, victim pool's some statistical properties can be used.
For example, a pool can detect FAW and BWH attacks by observing its fork rate and the ratio between the number of submitted PPoWs and FPoWs, respectively. 
It can also determine the attacker's infiltration power in the victim pool through this detection method.
Another prerequisite for ARS is to identify which pools have attacked a victim pool, and this may take a long time. 
Therefore, we propose investigating the variance of the reward densities of pools through moles to reduce the identifying time (see Section~\ref{sec:detection}).
The benign pool's reward density is not at all related with block rewards of other pools. 
Meanwhile, when a pool executes either FAW or BWH attacks, its reward density would be correlated with block rewards of other pools because the attacker earns part of its reward from the victim. 
This implies that we can identity the pools executing attacks based on this correlation information (see Section~\ref{sec:detection}). 

\smallskip
\noindent\textbf{Multiple pools:} 
In addition, for generalization, we extend ARS to that for multiple pools beyond two pools. Through a case study, we show that ARS still makes FAW and BWH attacks unprofitable in the game among multiple pools (see Section~\ref{sec:npools}).

\section{Background}
\label{Background}

In current cryptocurrencies, peers verify transactions issued by clients.
Peers record the verified transactions in a ``blockchain.''
To maintain the blockchain, many cryptocurrencies including Bitcoin adopt the PoW mechanism.
In this section, we describe the mining process, focusing on Bitcoin. 
Further, we demonstrate FAW and BWH attacks against mining pools.

\subsection{Bitcoin Basics}

\noindent\textbf{Mining Process:}
For mining, miners first gather issued transactions, which are not yet recorded in the blockchain, into their local storage. 
Then, miners place the transactions into a block and find a PoW, spending their computational power to generate a valid block.
The header of a block includes a Merkle root~\cite{merkle1980protocols} of transactions in the corresponding block and the hash value of the header of the previous block.
The block header also includes a nonce, which is a key component necessary to become a valid block. 

To be a valid block, the hash value of block header must be less than a given target number $T_1$.
In particular, Bitcoin uses the double SHA256 hash function.
The hash value of a block header is obtained as an output of double SHA256 for an input containing a concatenation of block contents including a nonce.
Miners increment a nonce to find a valid nonce, which makes the hash value less than the target number $T_1$. 
If a miner finds a valid nonce as a PoW, the miner generates the valid block including the nonce and propagates the block to a peer-to-peer network. 
Finally, the block is appended to the blockchain, and the above mining process is repeated.

In Bitcoin, the target number $T_1$ is adjusted every 2016 blocks to keep the average period of one block generation (i.e, the average period of one round) at 10 mins. 
The smaller the value of $T_1$, the more difficult the mining process will be.

\smallskip 
\noindent\textbf{Forks:}
When a miner $A$ propagates a block, another block can also be generated and propagated by a miner $B$ who has not yet received $A$'s block.
Therefore, miners receive two blocks and regard the first received block as the blockchain head. This situation is called a \emph{fork}. When a fork occurs, only one block becomes valid. 
Moreover, an attacker can generate intentionally forks to execute double-spending~\cite{doublespend,karame2012double} and selfish mining~\cite{eyal2014majority, nayak2016stubborn, sapirshtein2015optimal,gervais2016security}.

\smallskip 
\noindent\textbf{Mining Pools:}
As the mining difficulty has been increasing, 
mining pools were introduced, in which miners gather to mine together. 
Major pools consist of a manager and many miners. 
These pools run as one node in the Bitcoin network, and pool miners only need to connect to the manager and create IDs.
The manager forms and distributes a potential block to miners, and then miners spend their computational power to generate a valid nonce based on the block form provided by the manager.
Moreover, there are \textit{open} and \textit{closed} pools, depending on the policy whether any miner can join or not. 

Pools' reward systems are different from the block reward system in Bitcoin.
Miners in pools receive rewards for nonces, which make the hash value of the block header less than a new target number $T_2$.
The number $T_2$ is greater than the original target number $T_1$. 
We refer to a nonce for the target $T_2$ and $T_1$ as a partial proof-of-work (PPoW) and a full proof-of-work (FPoW), respectively.
When a miner finds a PPoW or an FPoW, the miner submits it (called \textit{shares}) to the manager, where PPoWs are needed to assess each miner's contribution in order to share the reward in the pool.
When the submitted share is FPoW, the manager generates a valid block and earns the block reward.
The manager then distributes the block reward to miners in proportion to the number of submitted shares.

\subsection{Existing Attacks on Pools}

\noindent\textbf{Block Withholding:}
To execute the BWH attack~\cite{rosenfeld2011analysis,courtois2014subversive}, 
an attacker splits her computational power into two parts, which are used for solo mining and malicious mining in the victim pool, respectively.
In the malicious mining, she submits only a PPoW to the victim without submitting an FPoW.
Although she undermines the victim by withholding blocks in the victim pool, she still earns a portion of the reward through PPoWs submitted to the victim.
In addition, her solo mining has a higher efficiency compared to the case where she does not attack because the block reward is gained in proportion to each pool's \textit{relative} computational power in Bitcoin. 
In other words, by undermining the victim pool, her relative computational power would increase. 
As a result, this point allows the attacker to earn an extra reward. 
Naturally, the BWH attack can be executed in many proof-of-work cryptocurrencies including Bitcoin, Ethereum, and Litecoin.

In 2015, Eyal~\cite{eyal2015miner} developed a BWH game between two pools. 
In the game, each pool can launch the BWH attack against an opponent by infiltrating a part of the computational power into the opponent.
Eyal found that the BWH game results in \textit{the miner's dilemma}. 
In other words, there is only one Nash equilibrium in which two pools execute BWH attacks against each other and both pools suffer losses. 

\smallskip 
\noindent\textbf{Fork After Withholding:}
For the FAW attack proposed in 2017~\cite{kwon2017selfish}, similar to the BWH attack, an attacker splits her computational power into two parts, which are used for her solo mining and malicious mining in a victim pool.
However, when the attacker finds an FPoW in the victim pool, she submits the FPoW to the manager unlike in the BWH attack. 
This occurs only if an external honest miner (i.e., neither the attacker nor a miner in the victim pool) propagates a block. 
Therefore, the attacker intentionally generates a fork through pools. 
In the FAW game, pools can launch FAW attacks against each other by infiltrating a portion of their computational power into the other pools.
There is unique Nash equilibrium, where two pools execute the FAW attack against each other.
In the equilibrium, a larger pool earns an extra reward (unlike in the BWH game) while a smaller pool suffers a loss. In other words, the game leads to a \textit{pool size game.} 
This fact can make the decentralization level of Bitcoin decrease when occurring FAW attacks among pools.

\smallskip 
\noindent\textbf{Countermeasure:}
Because no viable countermeasure against FAW and BWH attacks~\cite{eyal2015miner, kwon2017selfish} exists, the attacks can be launched in practice.
Indeed, the BWH attack was executed against the ``Eligius'' pool in 2014.
To detect the attack, a pool's manager can investigate the ratio of FPoWs to PPoWs.
If the ratio is low enough, the manager can speculate that the BWH attack has occurred in his pool.
However, identification of the attacker is known to be difficult if she executes BWH attacks using many Sybil nodes (IDs) in the pool.

For the FAW attack, miners can detect it because the attack will cause a high fork rate.
However, it is still difficult to identify the attacker because she indirectly generate intentional forks through pools instead of generating them by herself.
Of course, in the victim pool, the manager can expel any miners suspected of causing forks.
Nonetheless, the attacker's reward can be unaffected by this manager's behavior by planting many Sybil IDs in the victim pool.
In other words, even though an ID that generates a fork is expelled by the manager, the attacker still earns the extra reward though other IDs. Eventually, the manager would be unable to prevent FAW attacks with a simple blacklist of suspects.

Even though many papers~\cite{twophase, rosenfeld2011analysis, daian2017short} proposed a new PoW mechanism that can prevent BWH and FAW attacks, they might be impractical. 
These protocols make pool miners not know whether their found nonce can generate a valid block, and thus the miners cannot execute BWH and FAW attacks. 
However, this point causes another withholding attack of a manager, where she can stealthily withhold blocks found by a pool miner and earn extra profit through her solo mining.  
In addition, the protocols increase the pool operation cost.
The above facts make the protocols impractical~\cite{kwon2017selfish}. 
\section{Model and Formulation}
\label{Model}

\subsection{System Model}
\label{subsec:systemmodel}

\noindent{\bf Block generation in PoW:} 
In PoW mechanisms, miners attempt to generate valid blocks by finding an inverse image of a hash function satisfying a certain condition in each round, where one round is defined as the time during which a new task is generated and a valid block is found by a miner. 
Due to the pseudorandomness of hash functions, we assume that the number of blocks found by a miner for one round follows a Poisson distribution with the miner's relative computational power. 
Then, the number of blocks found by a pool also follows a Poisson distribution because the sum of Poisson random variables is a Poisson random variable. 
For simplicity, we assume that natural forks do not occur in the block generating process, as the probability ($\approx$0.004~\cite{gervais2016security}) of natural forks occurring is significantly low in practice~\cite{eyal2015miner,gervais2016security,kwon2017selfish}. 

\smallskip
\noindent{\bf Computation power and reward:} 
We let $\mathcal{I}$ be the set of all pools and solo miners\footnote{A solo miner directly conducts mining, not joining pools.}, and denote by $\alpha_i$ the computational power of $i\in\mathcal{I}$. 
For analytical convenience, we normalize the total computation power with 1, and thus $\sum_{i \in \mathcal{I}} \alpha_i = 1$. 
We assume that a node cannot possess more than 50\% computational power (i.e., $\alpha_i \leq 0.5$) as in the previous works. 
A reward for one block is assumed to be 1, implying that the total reward of a node $i$ in a round is simply the total number of blocks found by $i$ in that round. 
When a pool finds a block and earns the reward for the block, the pool manager distributes the reward to miners in proportion to the number of shares submitted for the time duration over which the pool finds the block.  
We also assume that the manager honestly distributes the reward among the pool's miners.
We define node $i$'s {\em reward density} at round $r$ as $\frac{R^r_i}{\alpha_i}$ where $R^r_i$ is $i$'s reward earned for round $r.$

\smallskip
\noindent{\bf Attacks:} 
We consider a case in which only two types of attacks, FAW and BWH, can be executed by pools. 
In addition, because most pools are open pools, we consider only open pools, not closed pools.
Closed pools will be discussed in Section~\ref{Discussion}.
For the worst case analysis, we assume that the FAW attack is executed under the best network capability\footnote{A variable $0\leq c \leq 1$ represents the network capability~\cite{kwon2017selfish}, where we assume $c=1$ in this study. 
This assumption, which is made for the sake of simplicity, can be relaxed readily.}, meaning that the attacker's blocks always become valid after forks caused by the FAW attack.
During an attack, we assume that an attacking pool executes either FAW or BWH attacks, but not both at the same time.
Indeed, an attack combining FAW and BWH is equivalent to the FAW attack under some network capability.
An attacking pool infiltrates a part of its computation power into ``victim'' pools.  
Note that such a choice on the attack strategy can be time-varying (see \textbf{stage} in the next paragraph). 
We consider a regime in which there are a sufficient number of miners in each pool, so as to assume that each pool's infiltration power used for attacks is a real number. 
Moreover, an attacker infiltrates its partial computational power into a victim pool using Sybil IDs in order to make it more difficult to identify who the attacker is.  
Nonetheless, we assume that the victim can trace the attacker and know how much infiltration power the attacker has used for attacks because the attacking pool is an open pool. 
In Section~\ref{sec:detection}, we describe how this becomes possible in practice.

\smallskip
\noindent{\bf Stage:} 
Our interest lies in investigating how pools interact over a long-term period. To this end, the entire time is divided into 
a sequence of {\em stages,} where over each stage a pool can know
  whether an attack occurs against itself, how much infiltration power
  is being used, and who the attacker is. This
  notion of stage is the one that is popularly used in repeated games (see Section~\ref{sec:model_fbgame}).
  Therefore, at the end of each stage, a victim identifies an
  attacker when an attack was executed against the victim during the stage.
Note that a stage consists of multiple rounds. At the start of each stage, pools can change their actions based on other pools' actions.  
For analytical tractability, we assume that stages are synchronized among pools.

\vspace{-2mm}
\subsection{Repeated FAW-BWH Game}
\label{sec:model_fbgame}

\noindent{\bf Primer on repeated game:} 
We aim at modeling how multiple pools interact to achieve their own objectives over a {\em long-term} period.  
To this end, we use the theory of {\em repeated games,} 
popularly used to understand the long-term interactions among players.  
In repeated games, the interactions among players repeat for multiple stages,
and the players become aware of other players' past behaviors and
their future benefits, accordingly adapting their strategies over time.
The main idea of the theory of repeated games is that a player may be
deterred from exploiting her short-term advantage by the threat of
punishment that reduces her long-term payoff.

The basic component of a repeated game is a (simultaneous-move) stage game $G$ played for each stage. 
The stage game $G$ is represented by $<\mathcal{N}, (A_i)_{i \in \mathcal{N}}, (U_i)_{i \in \mathcal{N}}>~,$ where $\mathcal{N}$ is the set of players, $A_i$ is the set of actions of the player $i,$ and $U_i(a)$ is a player $i$'s payoff function when the players' action profile is $a \in (A_i)_{i\in\mathcal{N}}.$ 
We denote by $G^T$ the $T$-period repeated game of $G$ with perfect information, where the players play the same stage game $G$ for $T$ stages, possibly $T = \infty$. 
We use the superscript $t$ to express the stage $t$  
in all notations, and let $a^t := (a_i^t)_{i \in \mathcal{N}}$ denote the action profile at stage $t,$ i.e., the actions by all players at stage $t.$ 
Also, we denote by $a_i := (a_i^t)_{0\leq t\leq T}$ the actions of player $i$ for $T$ stages. 
For $t \geq 1,$ let $h^t = (a^0, a^1, \ldots, a^{t-1})$ denote the history up to stage $t-1,$ where $H^t$ is the space of all possible stage-$t$ histories.
Depending on whether $T < \infty$ or $T=\infty$, we call each case as a finitely or infinitely repeated game. 
At each stage $t,$ each player $i$ knows all past actions $h^{t-1}$ and chooses the next action $a^t_i$ according to $i$'s strategy, thus $a^t=(a_i^t)_{i \in \mathcal{N}}$ is determined. 
Here, a {\em strategy} $s_i$ for player $i$ in the repeated game is a sequence of maps $s_i^t$---one for each stage $t$---that map a $(t-1)$-history $h^{t-1}$ to an action $a_i^t$ in $A_i.$
By perfect information, we mean that at the end of each stage game, players are able to know other players' actions and their payoffs.
In this paper, we focus on the infinitely repeated game, in which given the whole action profiles $\bm{a} = (a_1, a_2, \ldots )$, the payoff $\mathcal{U}_i(\bm{a})$ of player $i$ in the corresponding repeated game is the discounted average payoff, i.e., 

\vspace*{-5mm}

\begin{equation}
\label{eq:payoff}
  \mathcal{U}_i(\bm{a}) = \sum_{t=1}^\infty \delta^{t-1} U_i(a^t), 
\end{equation}

\vspace{-1mm}
\noindent where $ 0 < \delta < 1$ is the discount factor. The discount factor indicates how much a player discounts in the future.
Next, we introduce the concept of the subgame perfect Nash equilibrium.

\begin{definition}[Subgame Perfect Nash Equilibrium (SPNE)]
For a given history $h^t$ and player $i$'s strategy $s_i$, we denote the discounted average payoff of player $i$ in the subgame given the history $h^t$ as ${U}_i(s_i,\bm{s_{-i}}\,|\,h^t)$.
Then the strategy profile $\bm{s^\star}=(s^\star_i)_{i\in\mathcal{N}}$ is a subgame perfect Nash equilibrium if
\vspace{-1mm}
\begin{equation}
\mathcal{U}_i(s^\star_i,\bm{s^\star_{-i}}\,|\,h^t)=\max_{s'_i\in S_i} {\mathcal{U}_i(s'_i, \bm{s^\star_{-i}}\,|\,h^t)}
\,\text{ for } \forall i\in\mathcal{N}, \forall h^t\in H^t, \forall t>0, \notag
\vspace{-1mm}
\end{equation}
where $S_i$ is a space of player $i$'s strategies.
\end{definition}

\noindent In the repeated game, SPNE is a stronger version of Nash equilibrium, roughly meaning a strategy profile, which is a Nash equilibrium in \textit{every} subgame.
Thus, a SPNE is regarded as a mathematically-proved strategy vector that rational players are likely to follow when players interact over a long-term time period.  
There is also Folk Theorem~\cite{fudenberg2009folk} that states the existence of SPNE outcomes under a certain condition. 
More specifically, it represents that if there is a \textit{credible retaliation,} it is likely to achieve cooperation among rational players, where a credible retaliation indicates that is not costly for a retaliator. 

\smallskip\noindent{\bf Repeated FAW-BWH game:}
As in the previous studies~\cite{eyal2015miner, kwon2017selfish}, for simplicity we consider a game only between two pools (Pool$_1$ and Pool$_2$). 
As mentioned in Section~\ref{subsec:systemmodel}, we define a stage as the duration of time when each pool is able to trace other pools' attacking behaviors, which enables us to obtain the condition of perfect information in our analytical model.

In defining the repeated FAW-BWH game, it is crucial to define a stage game $G,$ for which we now describe how we model the set of actions $(A_i)_{i = 1, 2}$ and the payoff function $U_i(\cdot).$ 
First, each pool $i$'s action is defined in terms of which attack is performed and the amount of infiltration power used. 
We assume that a pool's attack is homogeneous, i.e., it executes either FAW attack or BWH attack. 
However, a pool is able to change its attack across stages. 
Then, in this stage game, each pool $i$'s action space is no-attack, FAW, or BWH, with a choice of some infiltration power for each attack. 
More formally, at stage $t$, each pool $i$'s action can be expressed as a vector $a_i^t = (f_i^t, b_i^t)$, where $f_i^t$ and $b_i^t$ are the infiltration powers for each FAW and BWH attack, respectively $(0\leq f_i^t, b_i^t \leq \alpha_i)$.  
It is clear that the case of $f^t_i = b^t_i =0$ corresponds to no-attack, in which case we simply denote it by $a_i^t=\bm{\overline{0}}$. 
Moreover, because we assume homogeneity in the attack, always only one of $f_i^t$ or $b_i^t$ is positive, i.e., $f_i^t \times b_i^t =0.$
Fig.~\ref{fig:model} depicts a model of our repeated FAW-BWH game between two pools.
When a stage $t$ ends, each pool $i$ is aware of its opponent pool's strategy ($a_{-i}^t$) for this stage game.
Then, pools can change their action, depending on the opponent's action at the previous stage, when a new stage game starts. 
To complete the definition of the repeated FAW-BWH game, it remains to define the payoff function $U_i(a^t)$ in \eqref{eq:payoff} at each stage game, which we will delay until Section~\ref{sec:cooperation}. 

\begin{figure}[t]
\centering
\includegraphics[width=0.9\columnwidth]{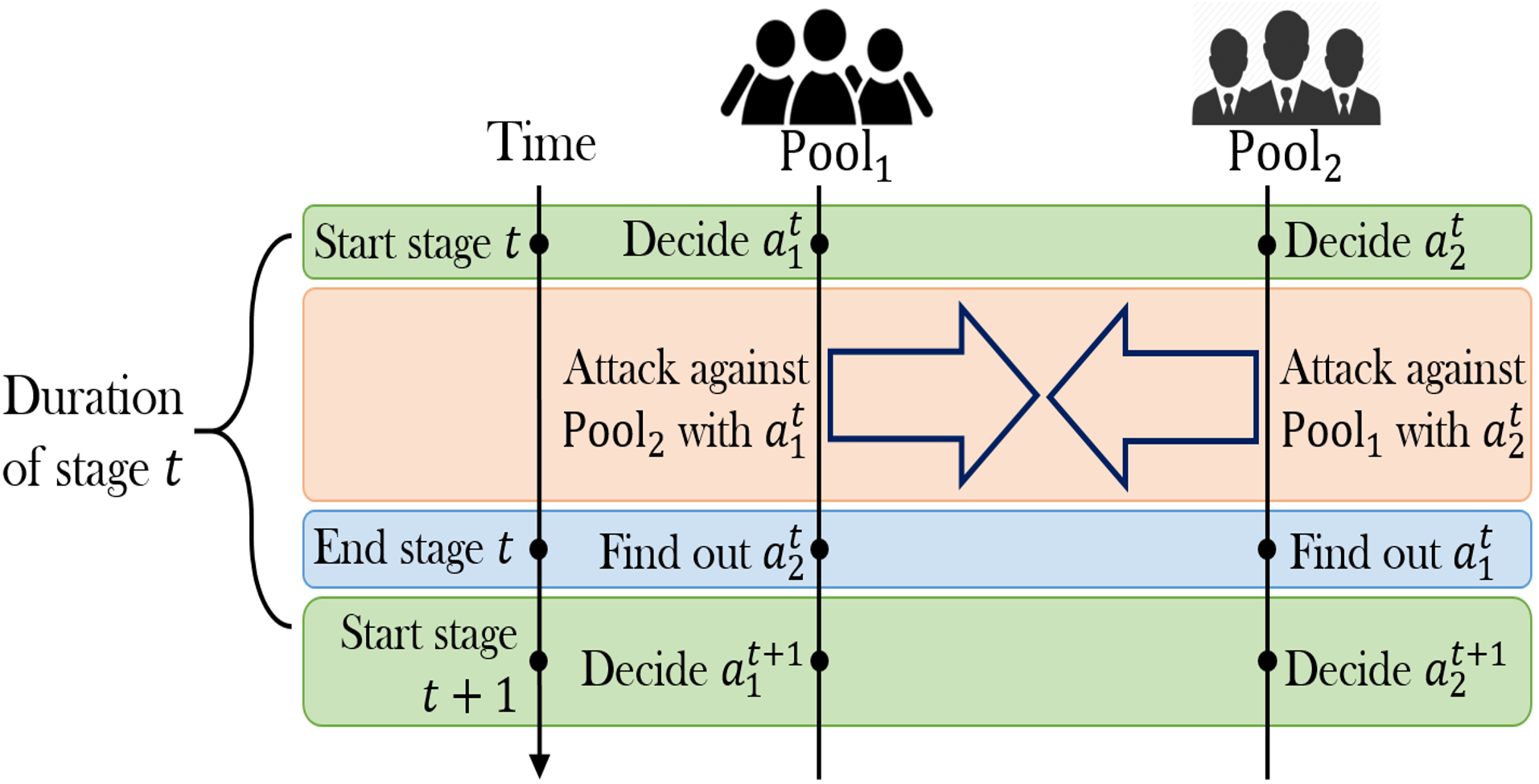}
\caption{\small At the start of stage $t$, two pools can decide their strategy $a_i^t$.
They then know the opponent's strategy $a_{-i}^t$ at the end of stage $t$.
This process is repeated at a new stage $t+1$.}
\label{fig:model}
\end{figure}
\section{Analysis of Repeated FAW-BWH Game} 
\label{sec:cooperation}

In this section, we study how two pools choose their strategies in the equilibrium of the repeated FAW-BWH game, to understand the pools' behaviors in terms of their interaction over a long time. 

\subsection{Payoff Function in the Stage Game}

Because pools play the same game at each stage, we remove the stage index $t$ in this section for simplicity. 
For an action profile $a=(a_1, a_2)$ of two pools, it seems natural to define the payoff $U_i$ as its extra reward density with respect to 1, i.e., $U_i(a_1, a_2) = \frac{R(a_1, a_2)}{\alpha_i}-1,$ where $R(a_1, a_2)$ is the average earned reward in each round. 
As described in Section~\ref{sec:model_fbgame}, each pool's action is expressed as its infiltration power for either FAW or BWH attack, i.e., $a_1 = (f_1, b_1)$ and  $a_2 = (f_2, b_2)$ where either $f_i$ or $b_i$ is 0.
Then, for convenience, we separately present four possible cases as follows: for a given profile $a=(a_1, a_2),$
\vspace{-1mm}
\begin{equation*}
\resizebox{0.85\hsize}{!}{$\begin{aligned}
  U_i\Big((f_i,b_i), (f_{-i},b_{-i})\Big) & = 
\begin{cases}
U^{FF}(f_i, f_{-i})& b_i=0\text{ and } b_{-i}=0,\cr
U^{FB}(f_i, b_{-i})& b_i=0\text{ and } f_{-i}=0,\cr
U^{BF}(b_i, f_{-i})& f_i=0\text{ and } b_{-i}=0,\cr
U^{BB}(b_i, f_{-i})& b_i=0\text{ and } f_{-i}=0,
\end{cases}
\end{aligned}$}
\vspace{-1mm}
\end{equation*}
where we henceforth provide the forms of four functions: $U^{FF},$ $U^{BB},$ $U^{BF},$ and $U^{FB}.$
According to the definition of payoff $U_i$, if two pools do not attack, their payoffs are 0.
If $U_i$ is positive, Pool$_i$ earns an extra reward.
Otherwise, Pool$_i$ suffers a loss.

\smallskip
\noindent{\bf Homogeneous attacks:}
The case in which two pools execute the same attack, FAW or BWH, has been studied in two related studies \cite{kwon2017selfish,eyal2015miner}.
For the FAW-FAW attack, the following payoff function can be obtained from Kwon's work~\cite{kwon2017selfish}.
\vspace{-1mm}
\begin{equation}
\resizebox{0.9\hsize}{!}{$
\begin{aligned}
U^{FF}(f_i,
f_{-i}) &=\frac{\alpha_i-f_i}{(1-f_i-f_{-i})(\alpha_i+f_{-i})}+\frac{f_{-i} (1-\alpha_i-\alpha_{-i})}{(1-f_{-i})(\alpha_i+f_{-i})} 
\cr 
&+\frac{f_if_{-i}}{2}\left(\frac{1}{1-f_i}+\frac{1}{1-f_{-i}}\right)\frac{1-\alpha_i-\alpha_{-i}}{(1-f_i-f_{-i})(\alpha_i+f_{-i})}
\cr
&+\frac{(U^{FF}(f_{-i}, f_i)+1)\cdot f_i}{\alpha_i+f_{-i}}-1.
\end{aligned}$}
\label{eq:faw_payoff}
\vspace{-1mm}
\end{equation}
In \eqref{eq:faw_payoff}, the first term is obtained from the honest mining of each pool, achieved with the computational power remaining after deducting the infiltration power.
Note that Pool$_i$ gets a reward of $\frac{\alpha_i-f_i}{1-f_i-f_{-i}}$ from the honest mining because each node earns a mining reward based on how many blocks it generated relative to others. 
The second term represents the extra reward density that is earned in the case where the opponent generates an intentional fork and Pool$_i$ does not generate any block.
In this case, both an external honest miner and the infiltration power of Pool$_{-i}$ generate a block, and the probabilities of these events are $\frac{1-\alpha_i-\alpha_{-i}}{1-f_{-i}}$ and $f_{-i}$, respectively.
This derives the second term. 
The third term is from intentional forks caused by both Pool$_1$ and Pool$_2.$ 
In this case, an external honest miner and infiltration powers of Pool$_1$ and Pool$_2$ find blocks, and then a fork with three branches occurs.
If the infiltration power of Pool$_i$ finds a block faster than that for Pool$_{-i}$, its probability would be $f_i\cdot\frac{f_{-i}}{1-f_i}\cdot\frac{1-\alpha_i-\alpha_{-i}}{1-f_i-f_{-i}}.$
On the other hand, if the infiltration power of Pool$_{-i}$ generates a block faster than that for Pool$_i$, its probability would be $f_{-i}\cdot\frac{f_i}{1-f_{-i}}\cdot\frac{1-\alpha_i-\alpha_{-i}}{1-f_i-f_{-i}}.$ 
Considering these facts, the third term is derived. 
Lastly, the fourth term is from its infiltration mining into the opponent and is derived from that the opponent distributes the reward of $(U^{FF}(f_{-i}, f_i)+1)\cdot f_i$ to Pool$_i.$ 
Note that Pool$_i$ infiltrates the computational power of $f_i$ into the opponent.

Next, we consider when two pools execute BWH attacks against each other. In this case, we have the following form of the payoff function from Eyal's work~\cite{eyal2015miner}, where forks are not intentionally generated so the second and third terms in \eqref{eq:faw_payoff} disappear in \eqref{eq:bwh_payoff}.
\begin{equation}
\resizebox{0.92\hsize}{!}{$
\label{eq:bwh_payoff}
U^{BB}(b_i, b_{-i})
 =\frac{\alpha_i-b_i}{(1-b_i-b_{-i})(\alpha_i+b_{-i})} %\cr
+\frac{(U^{BB}(f_{-i}, f_i)+1)\cdot b_i}{\alpha_i+b_{-i}}-1.$}
\end{equation}

\smallskip
\noindent{\bf Heterogeneous attacks:} As opposed to the payoff functions in homogeneous attacks borrowed from previous studies~\cite{kwon2017selfish,eyal2015miner}, it still remains to establish the payoff functions for when each of two pools execute FAW and BWH attacks. 
We first consider the case when Pool$_i$ and Pool$_{-i}$ execute FAW and BWH attacks, respectively. Then, the payoff $U^{FB}(f_i, b_{-i}),$ which quantifies the extra reward density, turns out to be given by:
\begin{equation}
\label{eq:faw_bwh_payoff}
\resizebox{0.92\hsize}{!}{$
U^{FB}(f_i,
  b_{-i})  =\frac{\alpha_i-f_i}{(1-f_i-b_{-i})(\alpha_i+b_{-i})}
 +\frac{(U^{BF}(b_{-i}, f_i)+1)\cdot f_i}{\alpha_i+b_{-i}}-1.$}
\end{equation}
This payoff can be easily derived. 
First, the first term represents the earned reward density of Pool$_i$ through its honest mining with the computational power remaining after deducting the infiltration power.
The second term is obtained from Pool$_i$'s infiltration mining into the opponent, Pool$_{-i}$.
Note that \eqref{eq:faw_bwh_payoff} does not have any reward density term earned from generated forks in Pool$_i$ because Pool$_{-i}$ does not generate forks in Pool$_i$.

Now, when Pool$_i$ and Pool$_{-i}$ execute BWH and FAW attacks with infiltration power $b_i$ and $f_{-i}$, respectively, we have: 
\begin{equation}
\resizebox{0.9\hsize}{!}{
$\begin{aligned}
& U^{BF}(b_i,  f_{-i})=\frac{\alpha_i-b_i}{(1-b_{i}-f_{-i})(\alpha_i+f_{-i})} 
+\frac{f_{-i}}{1-b_i} \cr
& \times\frac{1-\alpha_i-\alpha_{-i}}{(1-b_i-f_{-i})(\alpha_i+f_{-i})}
+\frac{(U^{FB}(f_{-i}, b_i)+1)\cdot b_{i}}{\alpha_i+f_{-i}}-1.
\end{aligned}$}
\label{eq:bwh_faw_payoff}
\end{equation}
Because only Pool$_{-i}$ executes the FAW attack, forks are generated by Pool$_{-i}$ in only Pool$_{i}.$ 
Thus, \eqref{eq:bwh_faw_payoff} is the addition of a similar form of \eqref{eq:faw_bwh_payoff} to the reward density (the second term) earned when forks occur in Pool$_i$.

\subsection{Equilibrium at the Stage Game}\label{sec:faw_retaliate}

We now discuss how two pools would behave at the equilibrium for each stage game, before we study how rational pools behave through long-term interactions in the repeated game. 
This step is of significant importance because (i) it clearly shows how much a near-sighted view of pools' interaction in each stage game (as in prior work~\cite{eyal2015miner, kwon2017selfish}) differs from a far-sighted one in the repeated games, and (ii) understanding the per-stage equilibrium behaviors is a key to understanding what happens if such stage games are repeated among pools. 
This per-stage equilibrium is stated in Theorem~\ref{thm:nash_stage}. 

\begin{theorem}[Nash equilibrium for stage game]
\label{thm:nash_stage}
There exists a unique Nash equilibrium (NE) $a^\star = (a_i^\star, a_{-i}^\star)$ in the stage game; it is characterized as:

\vspace*{-4mm}
\begin{align}
  \label{eq:stage_result}
  (a_i^\star, a_{-i}^\star)   = \Big( (f_i^\star,0), (f_{-i}^\star,0)
  \Big), \quad \text{where}  \quad  f_i^\star >0\text{ and } f_{-i}^\star >0.
\end{align}

\noindent
Further, the following payoff values are obtained for different cases of two pools' computational powers $\alpha_i$ and $\alpha_{-i}$:

\vspace*{-1mm} \noindent
\begin{align}
  U_i(a_i^\star, a_{-i}^\star)>0, \ U_{-i}(a_i^\star, a_{-i}^\star)<0&
\quad  \text{if} \quad \alpha_i > \alpha_{-i}, \label{eq:stage1} \\
U_i(a_i^\star, a_{-i}^\star) =  U_{-i}(a_i^\star, a_{-i}^\star) =0 &
\quad \text{if} \quad \alpha_i = \alpha_{-i}. \label{eq:stage2}
\end{align}
\end{theorem}

\noindent
(see Appendix~\ref{sec:proof5.1} for our proof of the theorem.) 
Theorem~\ref{thm:nash_stage} states the existence and the uniqueness of the Nash equilibrium, which is technically meaningful in the sense that per-stage equilibrium is predictably interpretable from a mathematical perspective. 
The major messages of Theorem~\ref{thm:nash_stage} are: (i) when both pools are allowed to execute FAW and BWH attacks, at the Nash equilibrium, the two pools execute only FAW attacks (see
\eqref{eq:stage_result}), and (ii) the larger pool always earns an extra reward, whereas the smaller pool always suffers a loss (see \eqref{eq:stage1}). 
However, when they possess the same computational power, no additional reward is provided to both pools (see \eqref{eq:stage2}). 
This is in stark contrast to the previous game where only BWH is allowed~\cite{eyal2015miner}. 
Note that Theorem~\ref{thm:nash_stage} provides the \textit{first} analysis of a scenario in which both BWH and FAW attacks are possible.
Also, the actions at the Nash equilibrium and their resulting payoffs differ markedly from those in classical games such as the prisoner's dilemma.

\subsection{ARS (Adaptive Retaliation Strategies)}
\label{sec:algo}
We now propose strategies that induce cooperation among two pools, i.e., no-attack, which is provably verified in the framework of repeated games. 
In the classical repeated game theory, it is well-known that ``threat of future punishment'' induces cooperation. 
We inherit such a rationale in our study; however, the following key differences are noted: (i) As mentioned in Section~\ref{sec:faw_retaliate}, 
the prisoner's dilemma is played repeatedly in many studies, whereas our stage game significantly differs from the prisoner's dilemma, and (ii) our stage game is also defined for a continuous action space, and thus, in punishing other pools deviating from cooperation, it is critically important to adaptively determine the amount of infiltration power for retaliation. 
As a result, considering the above facts, we should find a \textit{credible retaliation}, which is necessary for inducing cooperation according to Folk Theorem.

\subsubsection{Strategy description}
In this paper, we denote by $(a_i^t, a_{-i}^t)$ the resulting actions of two pools at stage $t$. A given strategy of both pools would produce the sequence of actions $(a_i^t, a_{-i}^t)_{t=0}^\infty.$
We now describe special strategies, named {\em ARS (Adaptive Retaliation Strategies),} which call a subroutine \textbf{Retaliate} of Algorithm~\ref{al:rt}. 
Here \textbf{Retaliate} has infinitely many versions, depending on a parameter $K$ that we will describe in \textbf{Retaliate subroutine} paragraph.
Therefore, we denote by ARS$_K$ a strategy belonging to ARS, and ARS$_K$ is represented in Algorithm~\ref{al:ars}.
When playing ARS$_K$, an internal variable $\stai$ representing the {\em standing} of Pool$_i$ is maintained by each pool; $\stai$ represents whether Pool$_i$ has followed ARS$_K$ well or not at the previous stage. 
We use the notation $\aitars$ to refer to the action when ARS$_K$ is played, in order to differentiate an action $\ait$ from a different strategy. 
Thus, $\ait = \aitars$ when Pool$_i$ plays ARS$_K$.

\begin{figure}[t!]
\centering
\includegraphics[width=\columnwidth]{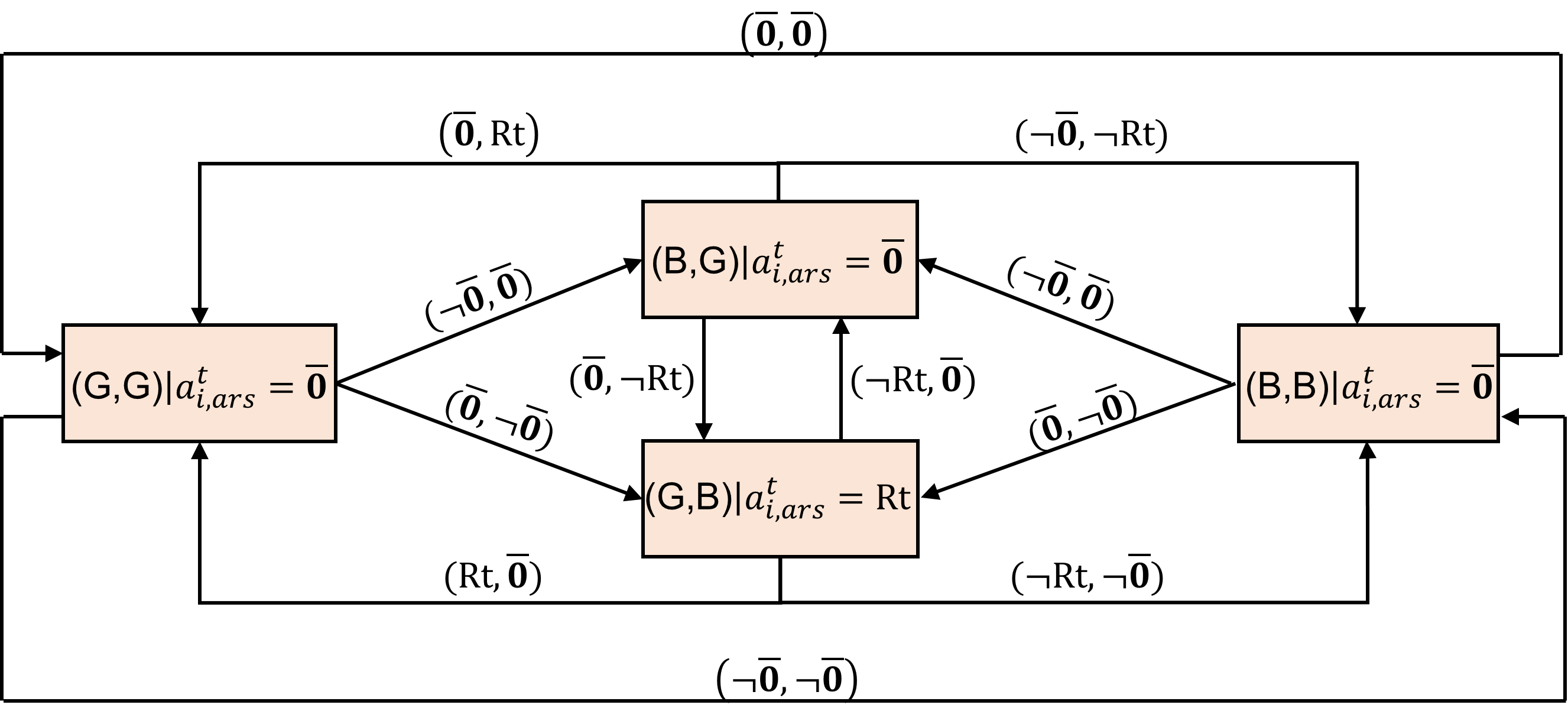}
\caption{\small ARS$_K$ has four states, depending on the two pools' standings. 
In each state box, ({\tt G} or {\tt B},{\tt G} or {\tt B}) represents $(\stai, \stami)$, and Pool$_i$'s action $a_{i,ars}^t$ is presented for the given pools' standings according to ARS$_K$.
The tuple at each edge represents $(a_i^t,a_{-i}^{t})$. 
The action tuple at each edge results in a stage change.}
\label{fig:ars}
\end{figure}

\smallskip\noindent{\bf ARS$_K$: }
In ARS$_K$, Pool$_i$ starts to cooperate with no-attack when $t=0$, and initializes its standing variable $\stai$ to $\good.$ 
At each stage $t$, Pool$_i$ first sets its standing $\stai$, depending on whether its stage $t-1$ action $\aitmi$ matches $\aitmiars$ from ARS$_K$ ({\bf \em S1}). Thus, if Pool$_i$ deviates from what ARS$_K$ does at the stage $t-1$, its standing at the stage $t$ is set to $\bad.$
Then, different standing values of both pools lead to the following combinations: $(\stai, \stami) = ({\tt G,G}),$ $({\tt G,B}), ({\tt B, G}), ({\tt B,B}),$ where ${\tt G} =$ `\good' and ${\tt B} =$ `\bad'.
To help readers better understand ARS$_K$, we present a state diagram of ARS$_K$ in Fig.~\ref{fig:ars}, where four states exist, depending on two pools' standings; 
in each state, we also present Pool$_{i}$'s action $\aitars$ at stage $t$ according to ARS$_K$. 
In the figure, $\text{Rt}$ is the output of {\bf Retaliate} with $K$ (denoted by {\bf Retaliate$_K$} in Algorithm~\ref{al:ars}), and $\neg (*)$ denotes an action value that differs from $(*)$. The action tuple at each edge (which may deviate from ARS$_K$) is what results in a state change, and we did not present the action tuples that do not change a state (e.g., in {\tt (G,G)}, the action $(\bm{\overline{0}}, \bm{\overline{0}})$ does not incur the state change). 

To summarize ARS$_K$, Pool$_{i}$ starts with cooperation, and then retaliates when the opponent deviates from cooperation. However, if, as a response to Pool$_i$'s retaliation, the opponent goes back to cooperation, Pool$_i$ stops retaliating and resumes cooperation with its opponent. 
If the opponent is not back to cooperation (thus not following ARS$_K$) and keeps executing attacks, Pool$_i$, which follows ARS$_K$, also keeps retaliating.
When two pools deviate from ARS$_K$ simultaneously, both of them turn out to cooperate at the next stage.
Retaliation phase is presented in Fig.~\ref{fig:ars}, where Pool$_i$ retaliates against its opponent ({\bf \em S2}) only when the standing is $({\tt G,B})$ (i.e., when Pool$_i$ follows ARS$_K$ but the opponent deviates from ARS$_K$). 
Considering this fact, at least one of $\aitmi$ and $\aimitmiars$, which are two inputs of {\bf Retaliate}, should be $\bm{\overline{0}}$ (see the actions at edges toward $({\tt G,B})$). 
This is because (i) if $a_i^{t-1} \neq \bm{\overline{0}}$, it indicates that the opponent's standing at stage $t-1$ was $\bad,$ and thus Pool$_{-i}$ should not attack according to ARS$_K$, i.e., $a_{-i,ars}^{t-1}$ is $\bm{\overline{0}}$, or (ii) if $a_{-i,ars}^{t-1} \neq \bm{\overline{0}}$, $a_i^{t-1}$ should be $\bm{\overline{0}}$.
After Pool$_i$'s retaliation, the opponent goes back to cooperation as a contrite behavior; the contrite phase is represented as a change from {\tt (G,B)} to {\tt (G,G)}, where two pools cooperate (Fig.~\ref{fig:ars}).
In the {\tt (B,B)} state where two pools deviate from ARS$_K$, both of them cooperate at the next stage, making the transition to {\tt (G,G)}. 

Note that ARS assumes that Pool$_i$ has values of $a^{t-1}_{-i}$, $\bar{a}^{t-1}_{-i}$, $\stami$, and $\alpha_{-i}$ of its opponent; we will discuss how that information is available to Pool$_i$ in Section~\ref{sec:detection}.
Indeed, there exists a strategy, \textit{contrite tit-for-tat} (CTFT)~\cite{boyd1989mistakes}, which uses standings similar to (not the same as) that for ARS and induces cooperation in the iterated prisoner's dilemma.
However, CTFT is studied as a strategy for the iterated prisoner's dilemma with a discrete action space including only two actions, where ARS significantly differs from CTFT.

\begin{figure}[!t]
 \removelatexerror
\scalebox{0.9}{
\begin{minipage}{1.15\linewidth}
\begin{algorithm}[H] 
  \caption{ARS (Adaptive Retaliation Strategies) for two pools.}
 \par\noindent\hrulefill \\
 {\bf ARS$_K$ for each pool $i$}
  \vspace{-0.15cm}
 \par\noindent\hrulefill
  \begin{compactenum}[$~$]
  \item {\bf Start when $t=0$:} $~$\\
  \begin{compactenum}[$~$]
  \item Start the stage game with no-attack, (i.e., $a_i^0 = \bm{\overline{0}}$), and set a variable $\stai=\good.$
  \end{compactenum}
 \item {\bf At each stage $t \geq 1$}: $~$\\ 
   \begin{compactenum}[\hspace{4mm}\bf \em S1.]
\item {\em Set the standing of this stage.} \\
     If ($\aitmi == \aitmiars$),
       $\stai=\good,$ else $\stai=\bad.$
    \item {\em Estimate the infiltration power.} \\    
    If ($\stai == \good$) and ($\stami == \bad$) \\
      \hspace{0.6cm}$\ait = \text{\bf Retaliate$_K$}\left(\alpha_i, \aitmi,\alpha_{-i},\aimitmi,\aimitmiars \right)$ \\
      else \hspace{0.08cm}$\ait = \bm{\overline{0}}.$
    \item Output $\ait.$
\end{compactenum}
\end{compactenum}
  \par\noindent\hrulefill	
  \vspace{0.1in}
  \label{al:ars}
  \end{algorithm}
\end{minipage}}
\end{figure}

\begin{figure}[!t]
 \removelatexerror
\scalebox{0.85}{
\begin{minipage}{1.15\linewidth}
\begin{algorithm}[H] 
  \caption{Retaliate Subroutine where  $M_i^F(\alpha_i,\alpha_{-i})$ and $M_i^B(\alpha_i,\alpha_{-i})$ are given in \eqref{eq:max_faw} and \eqref{eq:max_bwh}, respectively.}
  \par\noindent\hrulefill \\
  \KwIn{{\em Local.} computation power $\alpha_i$, previous action $\aitmi$}
  \KwIn{{\em Opponent.}  computation power $\alpha_{-i},$ previous action $\aimitmi$, 
  previous ARS$_K$ action $\aimitmiars$}
\smallskip
  \KwOut{Pool$_i$'s action $a_i^t$}
  \vspace{-0.15cm}
  \par\noindent\hrulefill
  \begin{compactenum}[\bf \em S.1]
  \item {\bf \em Retaliation with FAW} 
  \begin{compactenum}[\bf \em {S.1}.1]
  \item Construct the infiltration set for FAW $IP_{\text{faw}}  = IP_{\text{faw}} (a_i^{t-1},
a_{-i}^{t-1}, 
\aimitmiars)$.  If $(IP_{\text{faw}} ==\emptyset)$, goto
{\bf \em S.2}.
\item Find the following two sets $F_1$ and $F_2$ as follows:
  \begin{align}
    \label{eq:f1f2}
    F_1 &= \min\Big\{f_i \in IP_{\text{faw}}  \mid U_i(a^{t-1}_i,a^{t-1}_{-
          i})-\cr 
      &\hspace{1.2cm} U_i(a^{t-1}_i,\aimitmiars) \geq U_{- i}((f_i,0), \bm{\overline{0}})\Big\}, \cr
          F_2 & = \argmin_{f_i\in IP_{\text{faw}}} |f_i-M_{i}^F(\alpha_i,\alpha_{-i})|,
  \end{align}
        \item Compute the infiltration power $f_i$ for retaliating with FAW as
$f_i  = \min\{F_1\cup F_2\},$  and  set $\ait = (f_i,0).$ Goto
{\bf \em S.3}. 
  \end{compactenum}

\item {\bf \em Retaliation with BWH} 
  \begin{compactenum}[\bf \em {S.2}.1]
  \item Construct the infiltration set for BWH $IP_{\text{bwh}}  = IP_{\text{bwh}} (a_i^{t-1},
a_{-i}^{t-1}, \aimitmiars).$
\item Find the following two sets $F_1$ and $F_2$ as follows:
  \begin{align}
    \label{eq:b1b2}
    B_1 &= \min\Big\{b_i \in IP_{\text{bwh}}  \mid U_i(a^{t-1}_i,a^{t-1}_{-
          i})-\cr 
      &\hspace{1.2cm} U_i(a^{t-1}_i,\aimitmiars) \geq U_{- i}((0,b_i), \bm{\overline{0}})\Big\}, \cr
          B_2 & = \argmin_{b_i\in IP_{\text{bwh}}} |b_i-M_{i}^B(\alpha_{i},\alpha_{-i})|,
  \end{align}
        \item Compute the infiltration power $b_i$ for retaliating with BWH as
$b_i  = \min\{B_1\cup B_2\},$  and  set $\ait = (0,b_i).$ Goto
{\bf \em S.3}. 
  \end{compactenum}

\item {\bf \em Terminate.} Output $\ait.$
   \end{compactenum}    
 
  \par\noindent\hrulefill	
  \vspace{0.1in}
  \label{al:rt}
\end{algorithm}
\end{minipage}}
\end{figure}

\smallskip
\noindent{\bf Retaliate subroutine:} 
Prior to explaining {\bf Retaliate}, we first introduce the notion of an {\em infiltration power candidate set} (or simply, {\em infiltration set}) as follows: for given pools' local and opponent actions $\aitmi$, $\aimitmi,$ and $\aimitmiars,$ we define the infiltration set with respect to either FAW or BWH as the set of Pool$_i$'s infiltration powers that makes Pool$_{-i}$'s attack unprofitable as a retaliating response to Pool$_{-i}$'s deviation from cooperation. 
Formally, Pool$_i$'s infiltration set $IP_{\text{faw}}$ for the FAW attack is given as: 
\begin{equation}
\resizebox{0.9\hsize}{!}{$
\begin{aligned}
  \label{eq:FAW-can}
        IP_{\text{faw}}(\aitmi, \aimitmi, \aimitmiars) := \Big\{f_i|\,U_{- i} (a^{t-1}_i,a^{t-1}_{- i}) + K\cdot U_{- i} ((f_i,0),\bm{\overline{0}}) \cr
<  U_{- i} (a^{t-1}_i,\aimitmiars), 0 \leq f_i \leq \alpha_i \Big\},
\end{aligned}$}
\end{equation}
where $K$ is an arbitrary number in $[0,1)$. As $K$ gets close to 1, the retaliator tries to use FAW attacks as much as possible rather than BWH attacks.
Similarly, we define $IP_{\text{bwh}}$ for the BWH attack as:
\begin{equation}
\resizebox{0.9\hsize}{!}{$
\begin{aligned}
  \label{eq:BWH-can}
        IP_{\text{bwh}}(\aitmi, \aimitmi, \aimitmiars) := \Big\{b_i|\,U_{- i} (a^{t-1}_i,a^{t-1}_{- i}) +U_{- i} ((0,b_i),\bm{\overline{0}}) \cr
< U_{- i} (a^{t-1}_i,\aimitmiars), 0 \leq b_i \leq \alpha_i \Big\}.
\end{aligned}$}
\end{equation}

\noindent 
The main goal of {\bf Retaliate} is to determine which attack to perform and how much infiltration power is needed to retaliate against the deviating opponent while maximizing the retaliator's (long-term) payoff. 
Thus, the crux of {\bf Retaliate} is to strike a good balance between retaliation and selfishness.
To this end, we first prioritize FAW over BWH, simply because the FAW attack is known to be more profitable than the BWH attack~\cite{kwon2017selfish} (see {\bf \em S.1} and {\bf \em S.2}, where {\bf \em S.1} is first attempt). 
We henceforth focus on the steps for retaliation with FAW ({\bf \em S.1}), which is quite similar to that with BWH, where we first construct the FAW-infiltration set $IP_{\text{faw}}$ as in \eqref{eq:FAW-can}. 
In fact, it is possible for $IP_{\text{faw}}$ to be empty,
and this occurs when the FAW attack has no effect of retaliation, in which case the retaliation with BWH is then tried ({\bf \em S.2}). 
Note that the BWH-infiltration set ({\bf \em  S.2.1}) is provably non-empty. 
Intuitively, this is because BWH is known to have more strength in damaging the opponent more severely~\cite{kwon2017selfish}. 
In Appendix~\ref{sec:nonempty}, we prove the non-emptiness of $IP_{\text{bwh}}$.

Next, in balancing between retaliation and selfishness, we construct 
two filtered sets of infiltration powers, $F_1$ and $F_2$ ({\bf \em  S.1.2}), which consider retaliation and selfishness, respectively. 
In $F_1,$ Pool$_i$ following ARS$_K$ computes the set of infiltration powers in proportion to the degree of Pool$_{-i}$'s attack, i.e., generating the same amount of loss to Pool$_{-i}$ as that to Pool$_{i}$ from Pool$_{-i}$'s attack, which we call "equal retaliation". 
In $F_2,$ the set of infiltration powers is constructed so as to maximize Pool$_i$'s payoff for the FAW attack, 
where $f_i \in F_2$ is chosen to be closest to the infiltration power $M_{i}^F(\alpha_1,\alpha_2)$ maximizing Pool$_i$'s payoff $U_i$, expressed as:
\begin{equation}
\resizebox{0.89\hsize}{!}{$
M_{i}^F(\alpha_1,\alpha_2) = 
\frac{-1+\sqrt{1-\alpha_i(1+\alpha_{- i})-(1-\alpha_i-\alpha_{- i})
(1+\alpha_{- i}-\alpha_i\alpha_{- i})}}{\alpha_i(1-\alpha_i-\alpha_{- i})},\label{eq:max_faw}$}
\end{equation}
which is obtained from \cite{kwon2017selfish}.
Finally, in $\textbf{\em S.1.3}$, Pool$_i$ decides to retaliate by deciding between equal retaliation ($F_1$) and selfishness ($F_2$).
{\bf Retaliate} chooses the minimum infiltration power for FAW in $F_1$ and $F_2$, which is the minimum amount of power to achieve retaliation while considering its own payoff. 
Therefore, if Pool$_i$ must use a significant portion of its computational power to infiltrate for equal retaliation, it instead maximizes its payoff rather than pursuing equal retaliation.
Sets $B_1$ and $B_2$, which are similar to $F_1$ and $F_2$ for the FAW attack, are constructed for the BWH attack, where $M_{i}^B(\alpha_1,\alpha_2)$ is derived from \cite{eyal2015miner}, given as:
\begin{equation}
\vspace{-2mm}
\resizebox{0.7\hsize}{!}{$
M_{i}^B(\alpha_1,\alpha_2) =
\frac{-\alpha_{- i}(1-\alpha_i)+\sqrt{-\alpha_{- i}^2(-1+\alpha_i+\alpha_i\alpha_{- i})}}{1-\alpha_i-\alpha_{- i}}$}
\label{eq:max_bwh}
\end{equation}
\vspace{-1mm}

\noindent
Note that {\bf Retaliate} outputs $\bm{\overline{0}}$ if $0 \in F_1$, in which case Pool$_i$ does not need to retaliate.
This occurs when the opponent did not attack at stage $t-1,$ even though the opponent would retaliate against Pool$_i$ at stage $t-1$ according to ARS$_K$.
In this case, because the opponent did not follow ARS$_K$, the opponent's standing would be $\bad$, where Pool$_i$ would call {\bf Retaliate}.
However, {\bf Retaliate} would usually output $\bm{\overline{0}}$ in this case, and Pool$_i$ would not attack for retaliation because $IP_{\text{faw}}$ and $F_1$ would include 0 in most cases.

\subsubsection{Equilibrium Analysis}
Next, we prove that ARS is a subgame perfect Nash equilibrium for a sufficiently large $\delta$. 

\begin{theorem}
\label{thm:main}
There exists a function $F_K(\alpha_1,\alpha_2)$ such that, for all discount factor $\delta\geq F_K(\alpha_1,\alpha_2),$ the two-pool strategy vector (ARS$_K$, ARS$_K$) is a subgame perfect Nash equilibrium. 
Function $F_K(\alpha_1,\alpha_2)$ is always less than 1, and $F_K(\alpha_1,\alpha_2)$ is an increasing function of $K$ and $|\alpha_1-\alpha_2|$ for given $\alpha_1.$ 
Moreover, (ARS$_K$, ARS$_{K^\prime}$) is a Nash equilibrium for all $\delta\geq F_{\max(K,K^\prime)}(\alpha_1,\alpha_2).$
\end{theorem}
\noindent A proof of Theorem~\ref{thm:main} appears in Appendix~\ref{sec:proofmain}.
In the proof of Theorem~\ref{thm:main}, we show that it is not more profitable for each player to deviate ARS at the start of any subgame when compared to the case where it follows ARS. This implies that ARS is a subgame perfect Nash equilibrium according to \textit{one-time deviation property}. 

If two pools use one of ARS (their strategies 
need to be not necessarily the same), the strategy vector is a Nash equilibrium.
Especially, if two pools use the same strategy, the strategy vector is a subgame perfect Nash equilibrium.
As described in Section~\ref{sec:model_fbgame}, a subgame perfect Nash equilibrium refines a Nash equilibrium by eliminating \textit{non-credible threats}, which is a strategy vector that rational pools are actually unlikely to follow.
In addition, a large value of $\delta$ implies a condition in which pools consider future payoffs significantly, or the probability that pools are patient enough to stay inside the system for a long time. 
Indeed, most pools, including Slush, Eligius and F2Pool, are operated for a long time in practice.
A large value of $\delta$ is also better satisfied when the duration of one stage is short compared to the pools' entire operation period.
In Section~\ref{sec:detection}, we explain that the duration of a stage period can be short, which supports the practical value of our analytical result. 

Indeed, there are many other subgame perfect Nash equilibria (from Folk Theorem in repeated games~\cite{fudenberg2009folk}) in the repeated FAW-BWH game, which trivially include the one that two pools always execute FAW attacks against each other.
From a manager's perspective, the manager would want to increase its pool size while earning extra rewards, until the pool increases to a size that does not seriously threaten the system. 
This is a good reason for the manager to execute the FAW attack.  
Meanwhile, it is unknown whether the subgame perfect Nash equilibria have cooperation between pools because the existence of \textit{credible retaliation} in the repeated FAW-BWH game has not been studied to the best of our knowledge.
Our results imply that cooperation can be stable when ARS is used even though the FAW-BWH game has a certain winner; i.e., a larger pool.
Moreover, ARS includes infinitely many strategies with $K$ in \eqref{eq:FAW-can}.
As such, there are infinitely many ways to achieve cooperation between pools.
In addition, ARS restores cooperation even if FAW and BWH attacks impulsively occur, which is another advantage of ARS. 
For example, if Pool$_i$ impulsively attacks, the opponent would retaliate. 
After that, Pool$_i$ does not attack, being contrite, and two pools achieve the no-attack status. 

\section{Numerical Analysis}
\label{sec:numerical}
In this section, we use a numerical analysis to demonstrate how much infiltration power each pool would use for retaliation in ARS$_K$ in response to the opponent's action. We simulate the repeated FAW-BWH game with varying Pool$_1$ and Pool$_2$'s sizes.
We consider a case in which Pool$_1$ deviates from ARS$_K$ to attack Pool$_2$ while maximizing its payoff $U_1$ during one stage.
As a result, Pool$_2$ would retaliate against Pool$_1$ according to ARS$_K$. 
In this section, we consider a strategy ARS$_{1^-}$, where $K$ is close to 1. 

\begin{figure*}[ht]
\centering{
\subfloat[Pool$_2$'s infiltration ratio for retaliation using the FAW attack.]{
\includegraphics[width=0.25\textwidth]{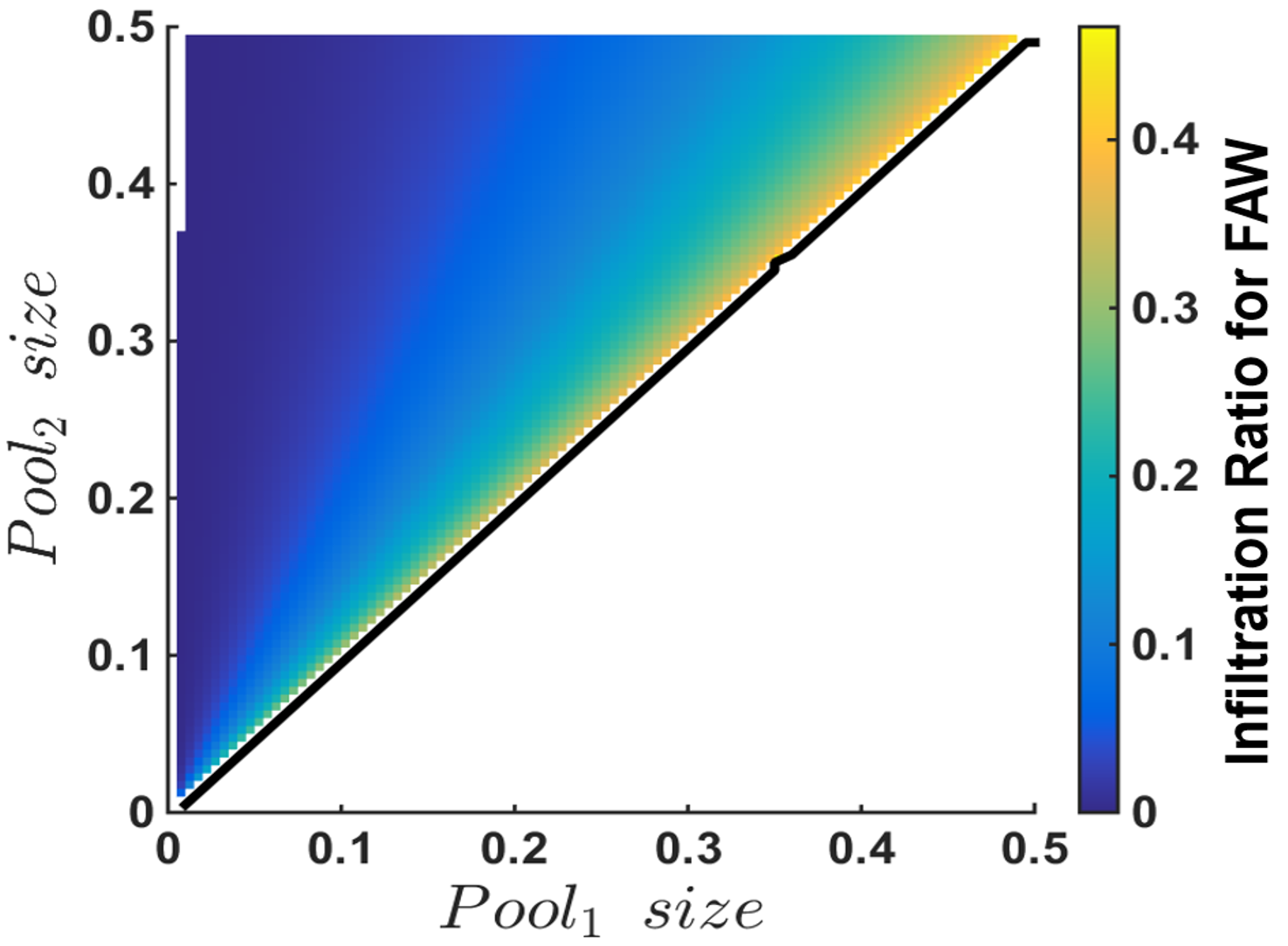}
\label{fig:faw_faw_rt} 
}
\subfloat[Pool$_2$'s infiltration ratio for retaliation using the BWH attack.]{
\includegraphics[width=0.25\textwidth]{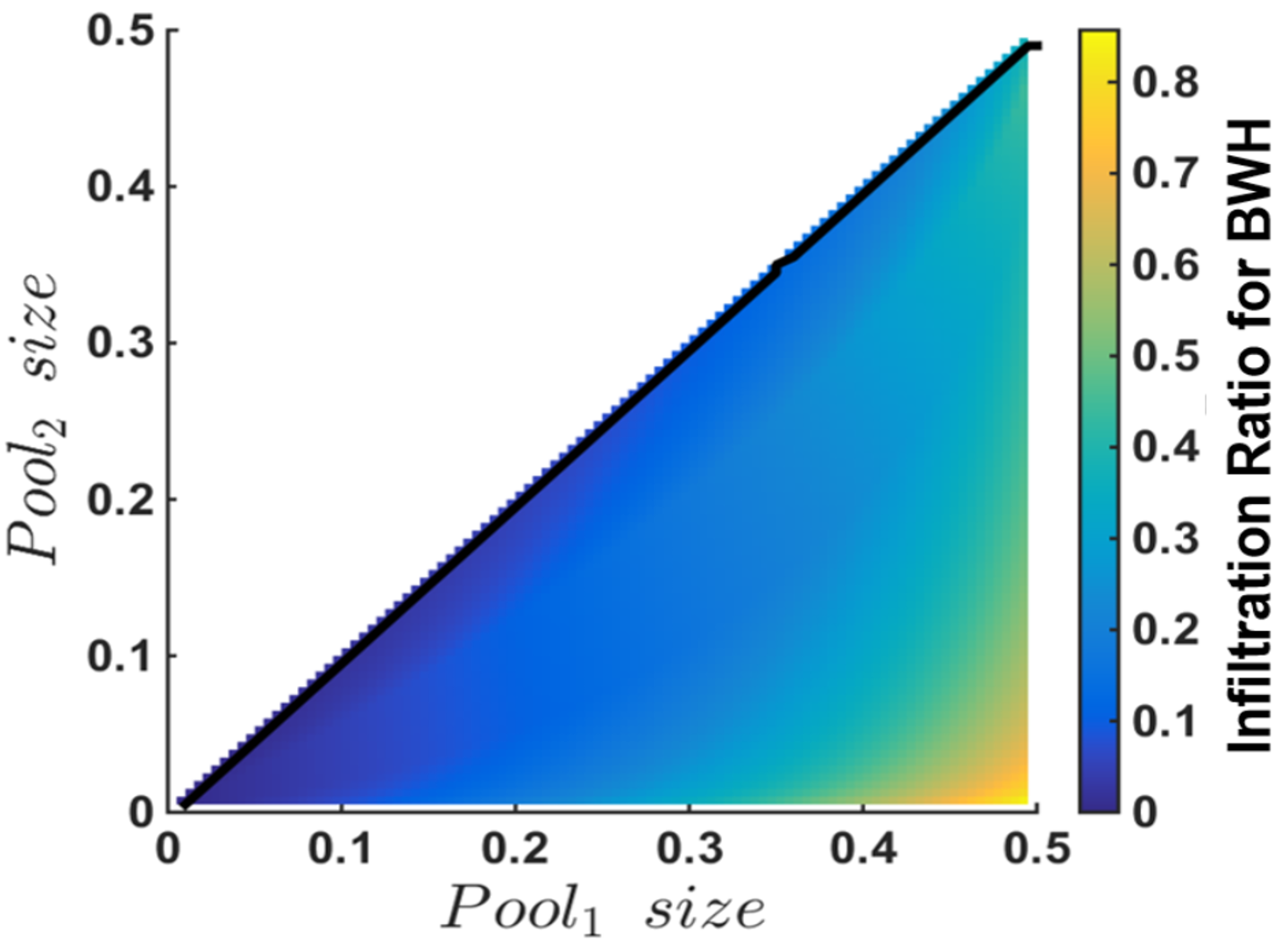}
\label{fig:faw_bwh_rt}
}
\subfloat[Average relative extra reward (\%) of Pool$_1$ for two stages.]{
\includegraphics[width=0.25\textwidth]{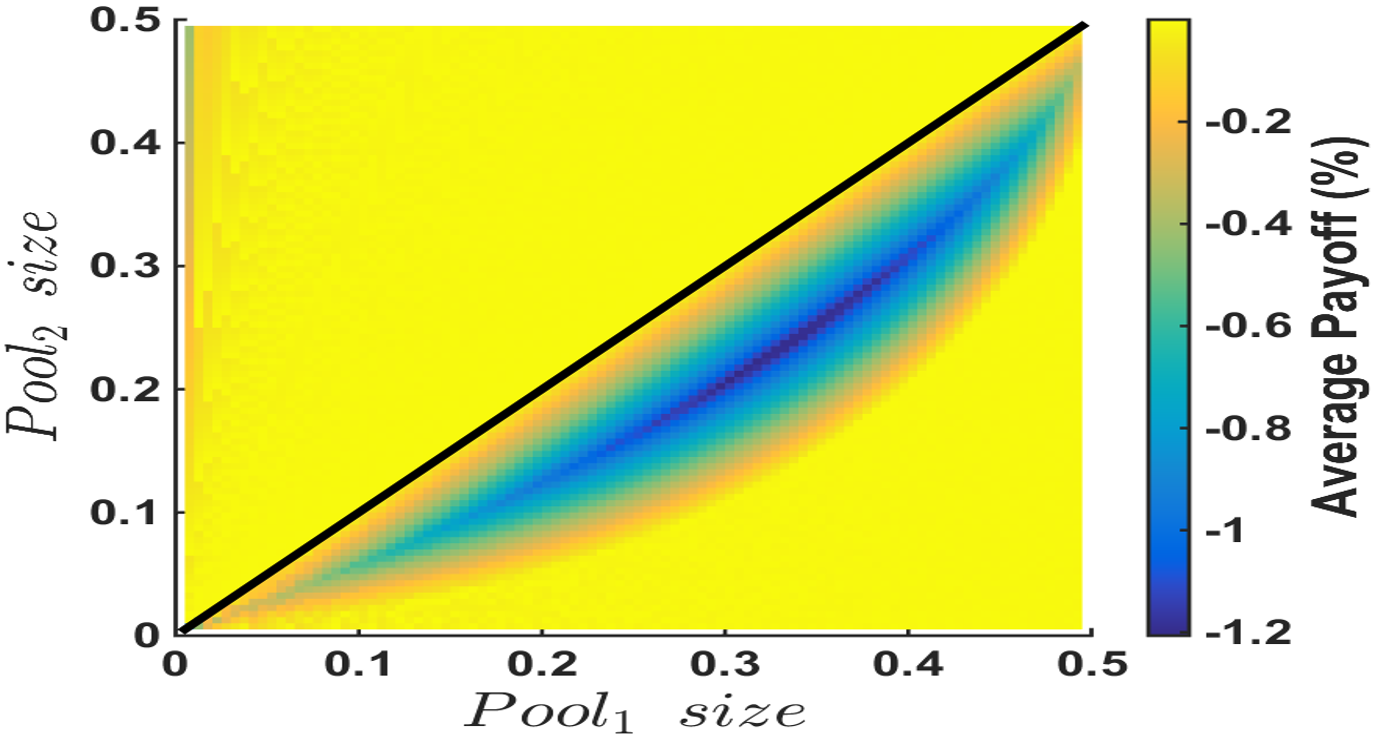}
\label{fig:faw_reward_a}
}
\subfloat[Average relative extra reward (\%) of Pool$_2$ for two stages.]{
\includegraphics[width=0.25\textwidth]{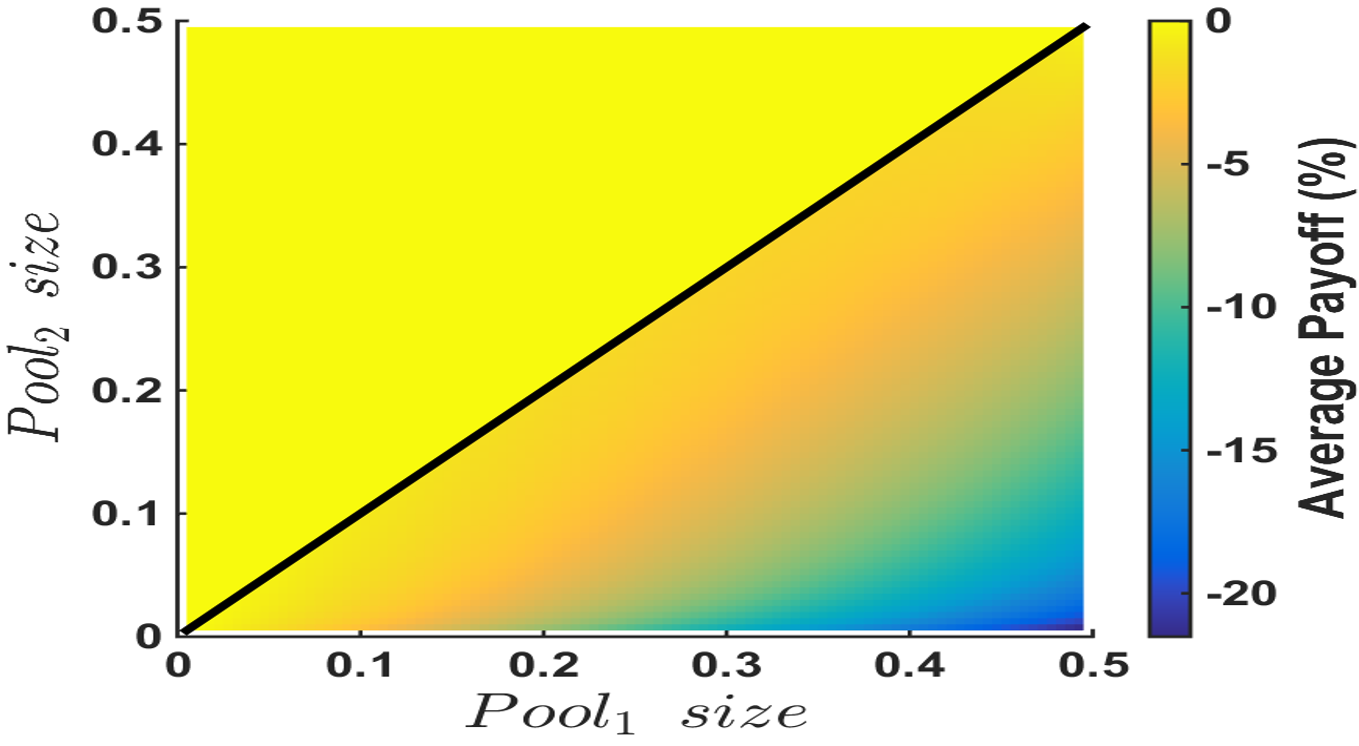}
\label{fig:faw_reward_p}
}} \vspace{-2mm}
\caption{Pool$_1$ optimally executes the FAW attack.}
\label{fig:faw} \vspace{-5mm}
\end{figure*}

\begin{figure*}[ht]
\centering{
\subfloat[Pool$_2$'s infiltration ratio for retaliation using the FAW attack.]{
\includegraphics[width=0.25\textwidth]{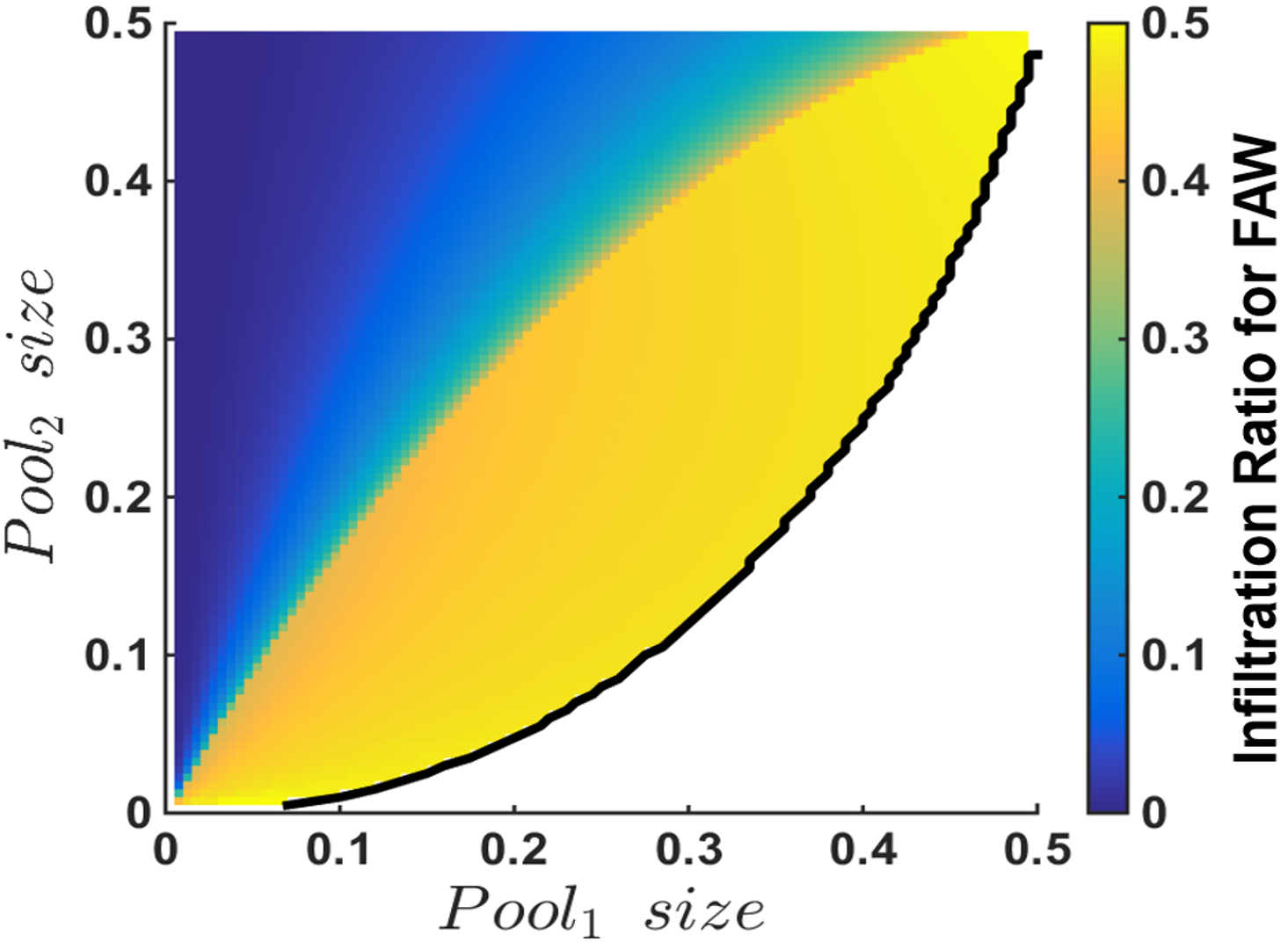}
\label{fig:bwh_faw_rt}
}
\subfloat[Pool$_2$'s infiltration ratio for retaliation using the BWH attack.]{
\includegraphics[width=0.25\textwidth]{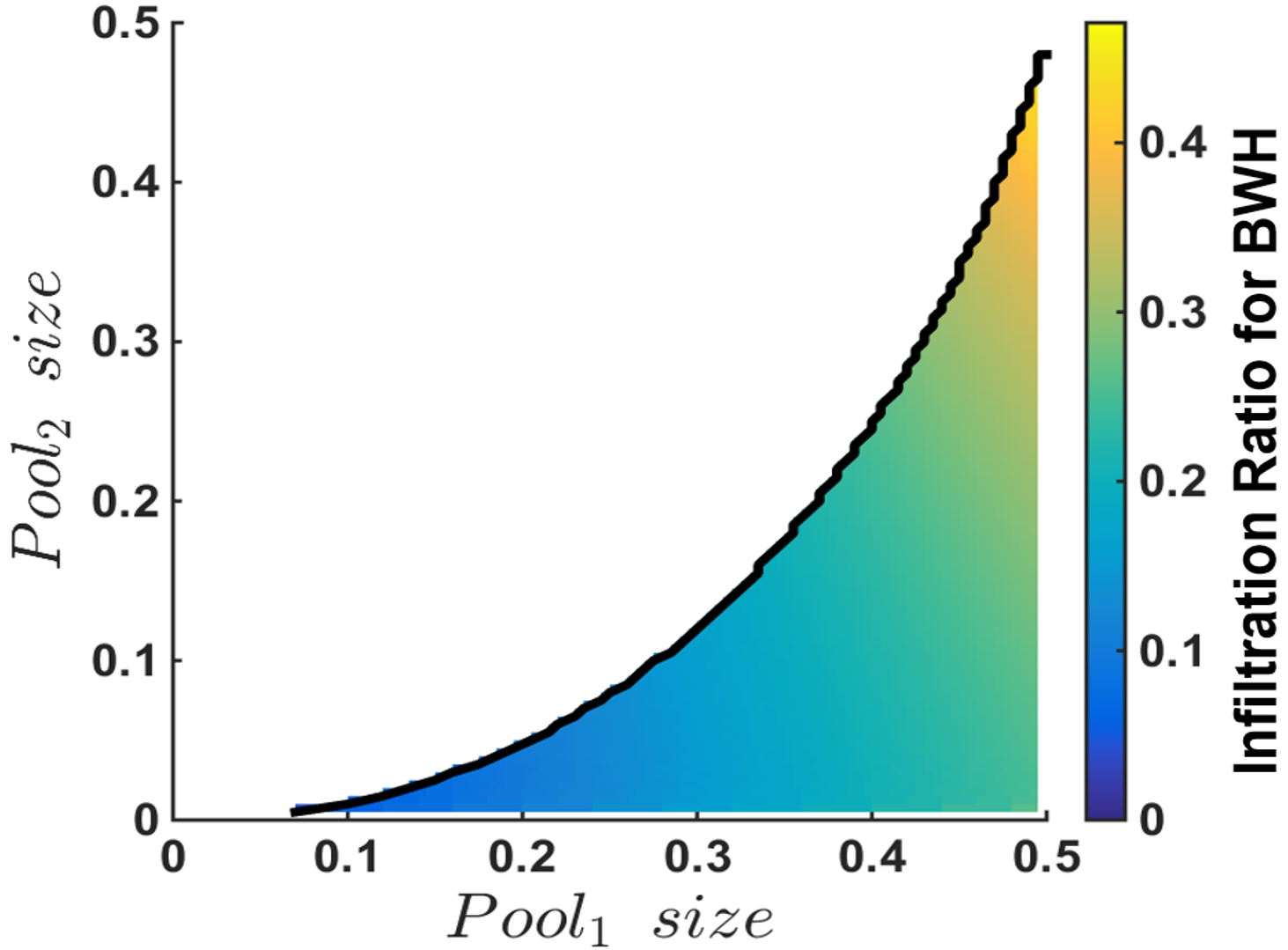}
\label{fig:bwh_bwh_rt}
}
\subfloat[Average relative extra reward (\%) of Pool$_1$ for two stages.]{
\includegraphics[width=0.25\textwidth]{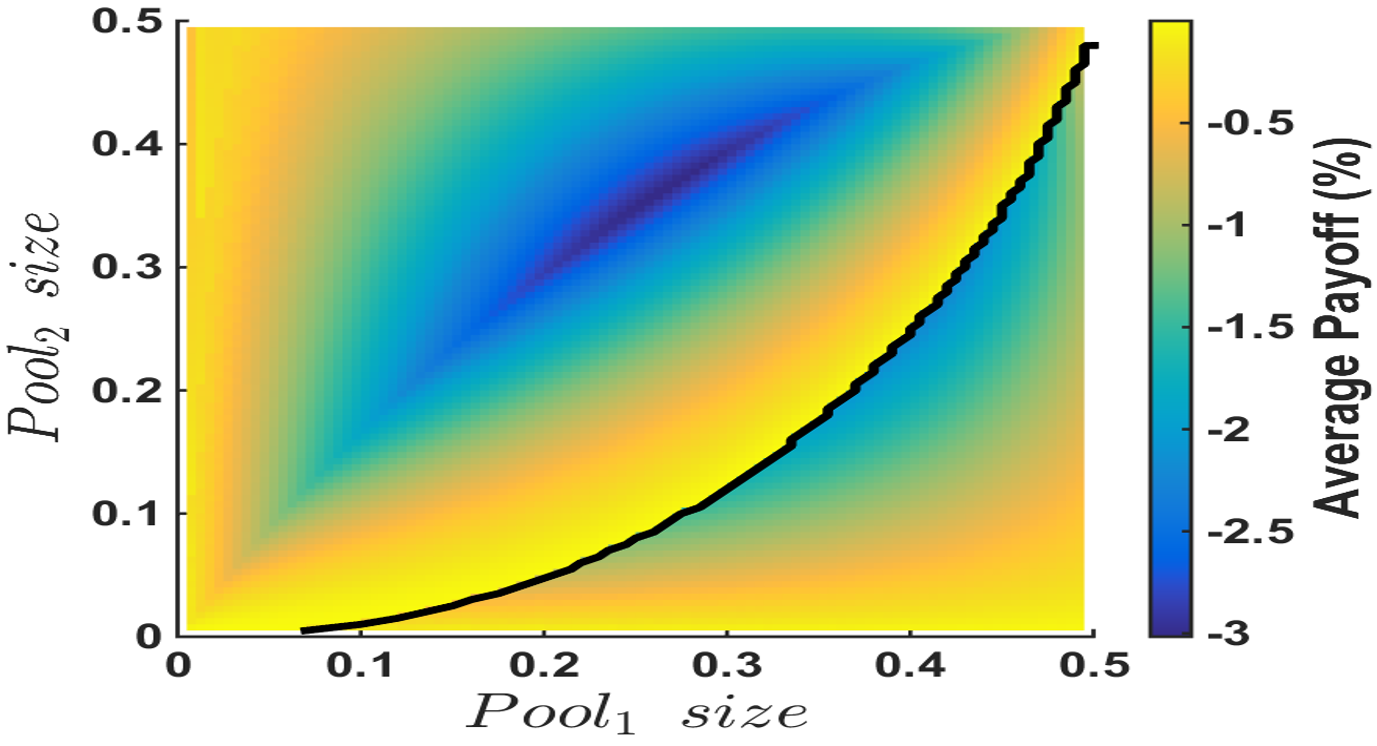}
\label{fig:bwh_reward_a}
}
\subfloat[Average relative extra reward (\%) of Pool$_2$ for two stages.]{
\includegraphics[width=0.25\textwidth]{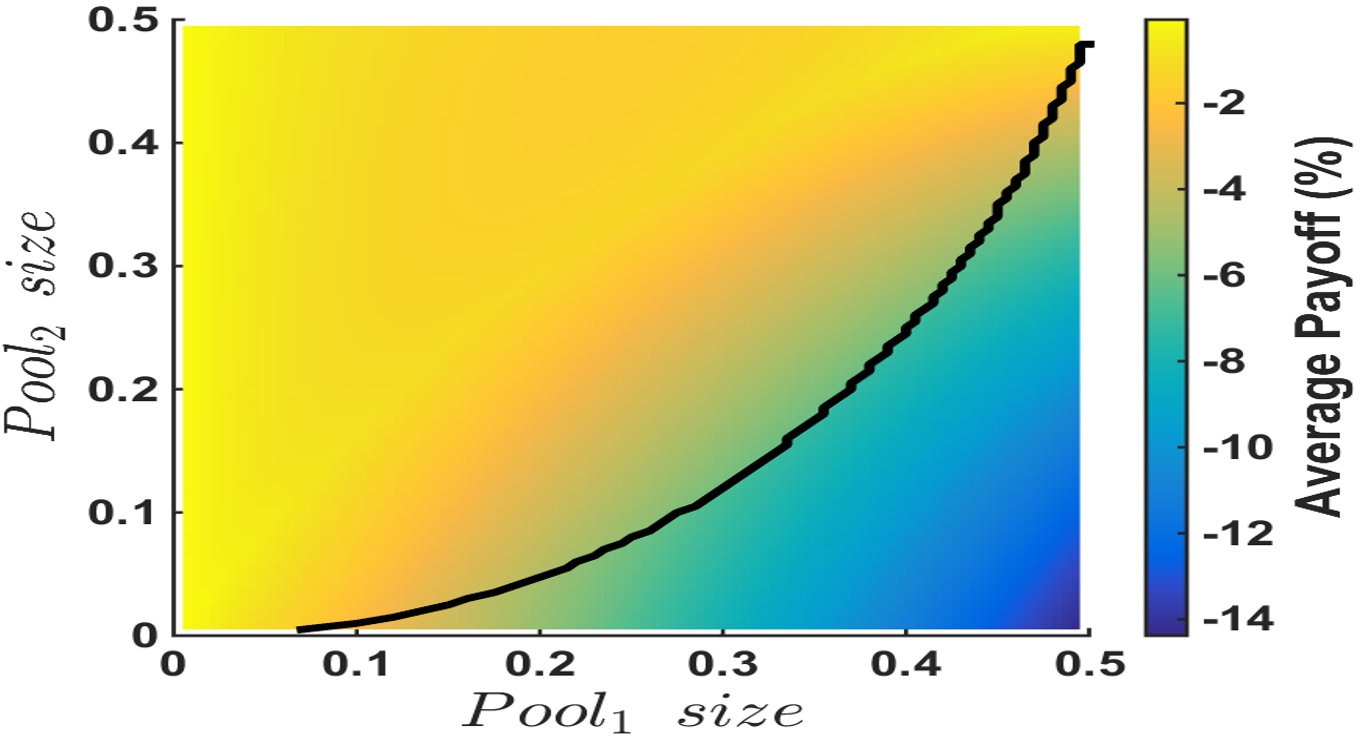}
\label{fig:bwh_reward_p}}}\vspace{-2mm}
\caption{Pool$_1$ optimally executes the BWH attack.} 
\label{fig:bwh} \vspace{-5mm}
\end{figure*}

\begin{figure*}[ht]
\centering{
\subfloat[Pool$_2$'s infiltration ratio for retaliation with the FAW attack against Pool$_1$'s FAW attack.]{
\includegraphics[width=0.25\textwidth]{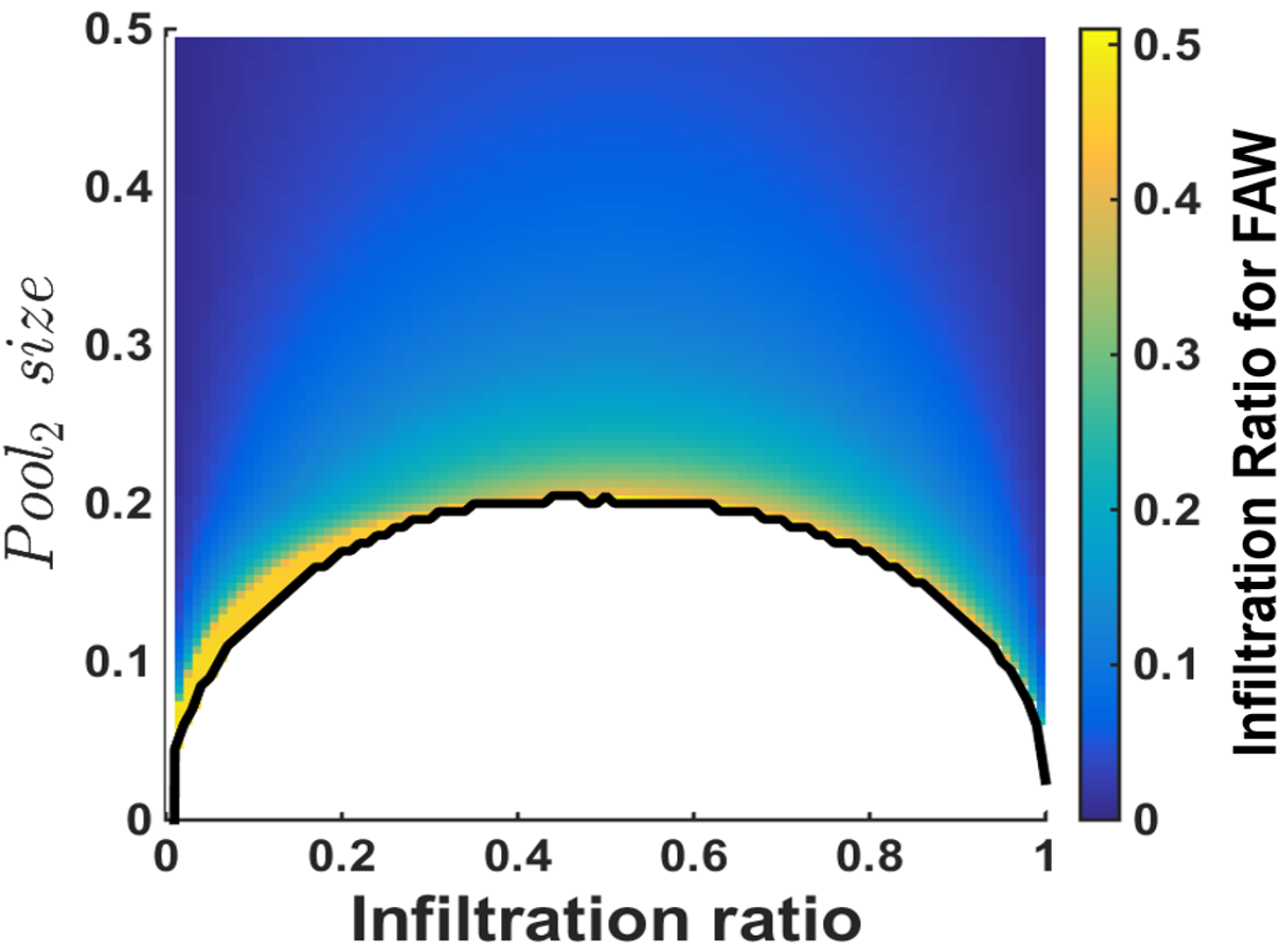}
\label{fig:faw_faw_rt2}
}
\subfloat[Pool$_2$'s infiltration ratio for retaliation with the BWH attack against Pool$_1$'s FAW attack.]{
\includegraphics[width=0.25\textwidth]{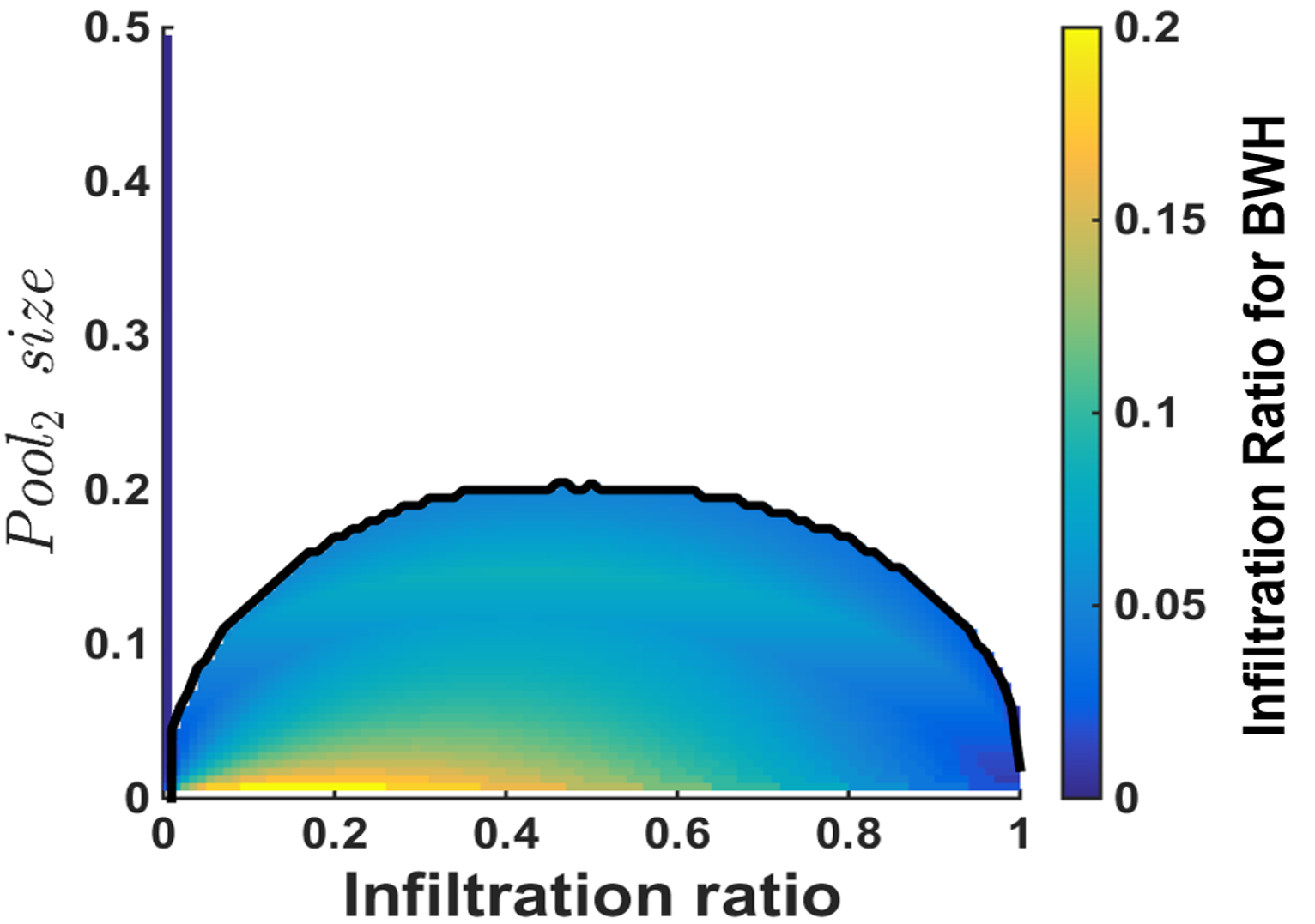}
\label{fig:faw_bwh_rt2}
}
\subfloat[Average relative extra reward (\%) of Pool$_1$ for the two stages 
when Pool$_1$ executes the FAW attack.]{
\includegraphics[width=0.25\textwidth]{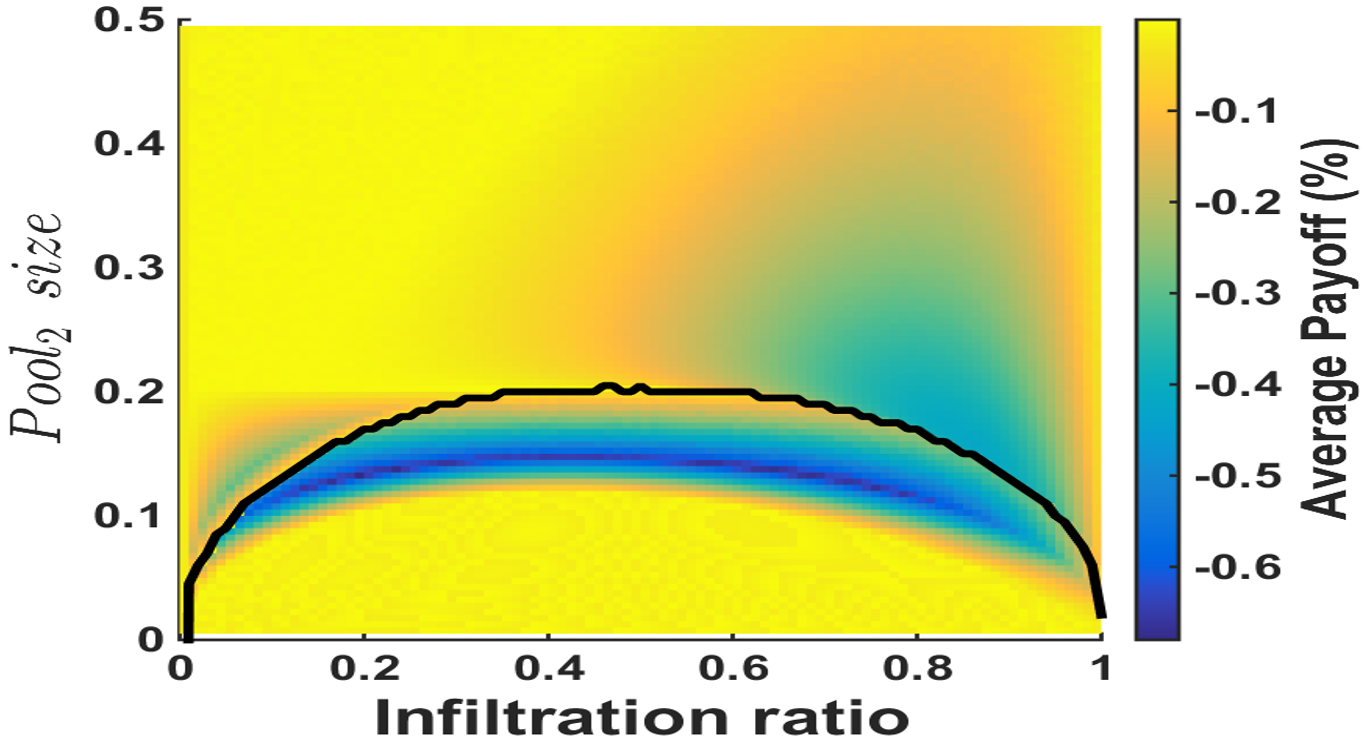}
\label{fig:faw_reward_a2}
}
\subfloat[Average relative extra reward (\%) of Pool$_2$ for the two stages 
when Pool$_1$ executes the FAW attack.]{
\includegraphics[width=0.25\textwidth]{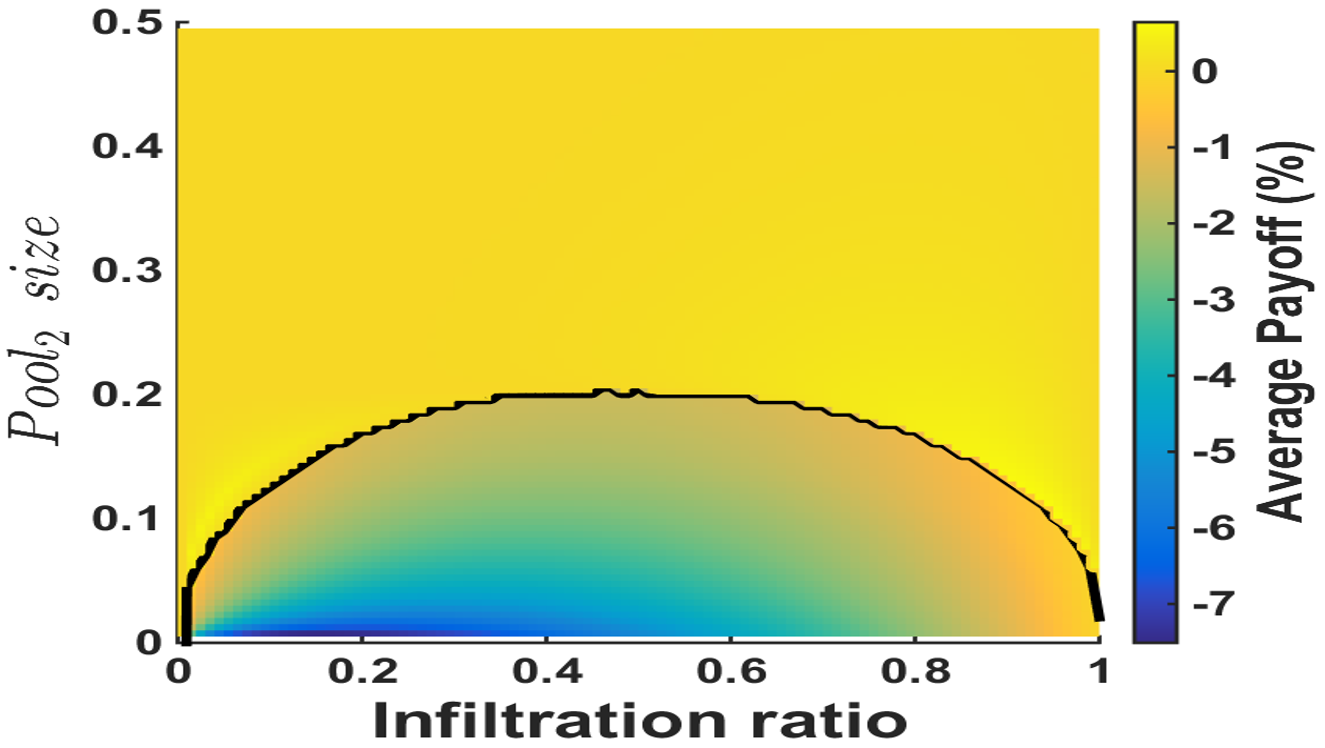}
\label{fig:faw_reward_p2}}}\vspace{-2mm}
\caption{Pool$_1$ with computational power of 20\% executes the FAW attack.}
\label{fig:faw2} \vspace{-5mm}
\end{figure*}

\begin{figure*}
\centering{
\subfloat[Pool$_2$'s infiltration ratio for retaliation using the FAW attack against Pool$_1$'s BWH attack.]{
\includegraphics[width=0.25\textwidth]{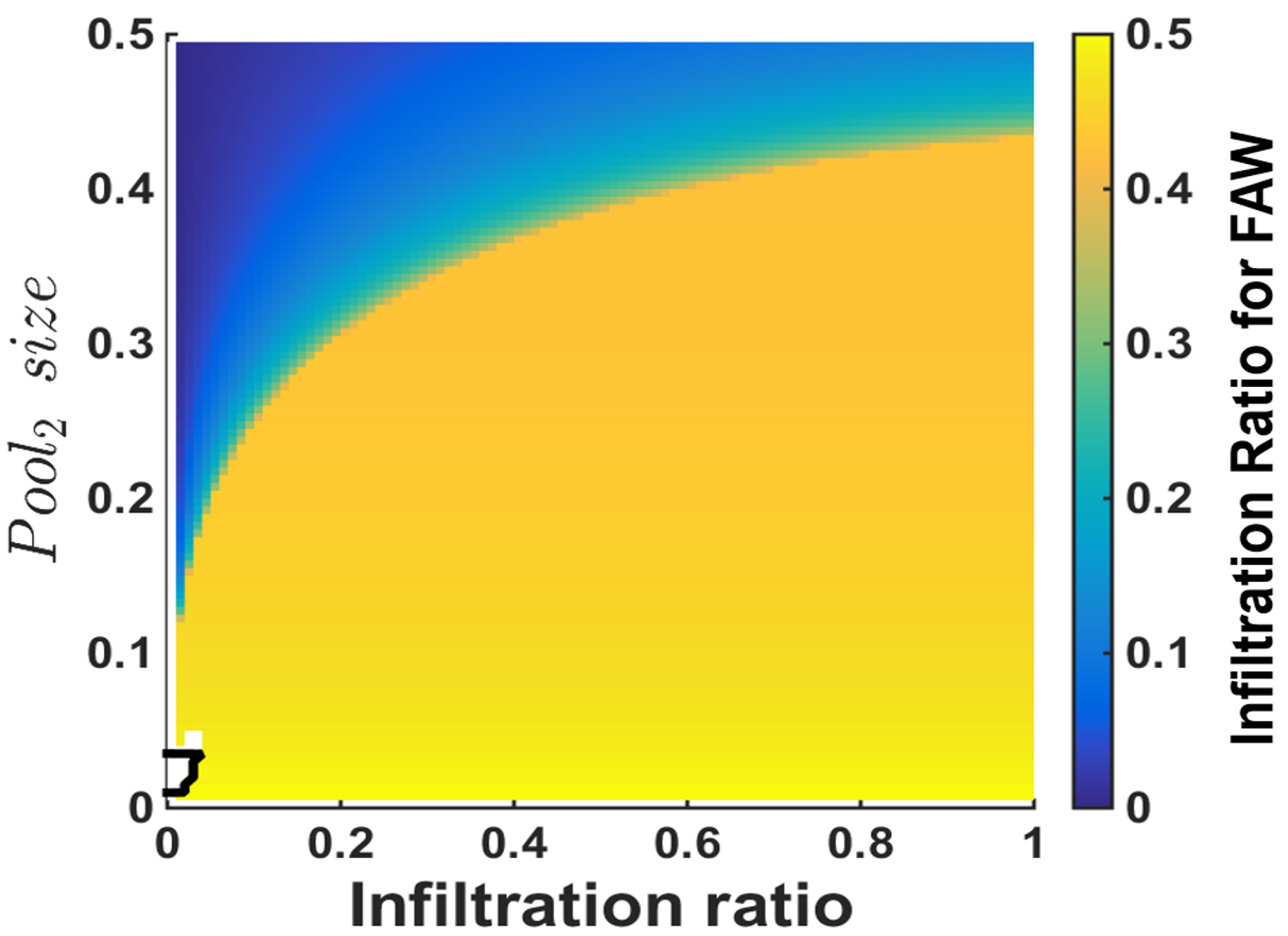}
\label{fig:bwh_faw_rt2}
}
\subfloat[Average relative extra reward (\%) of Pool$_1$ for the two stages 
when Pool$_1$ executes the BWH attack.]{
\includegraphics[width=0.25\textwidth]{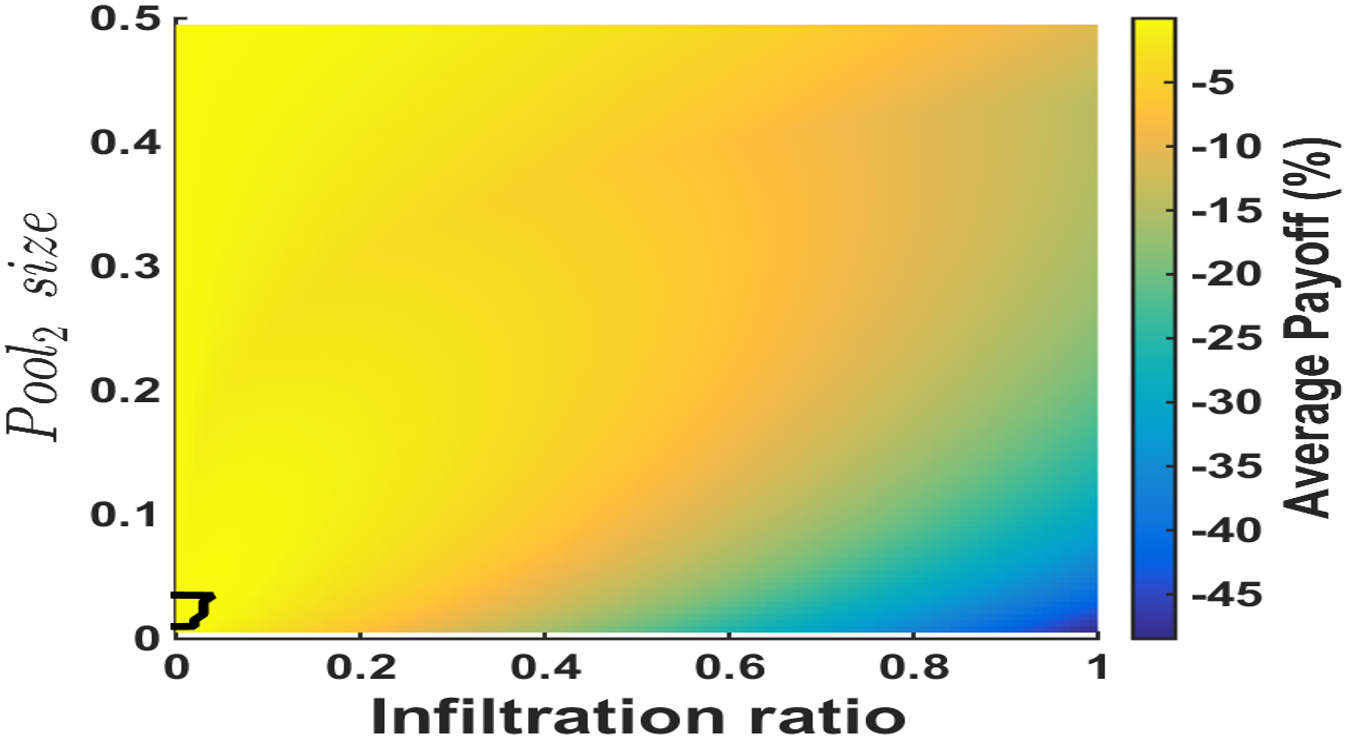}
\label{fig:bwh_reward_a2}
}
\subfloat[Average relative extra reward (\%) of Pool$_2$ for the two stages 
when Pool$_1$ executes the BWH attack.]{
\includegraphics[width=0.25\textwidth]{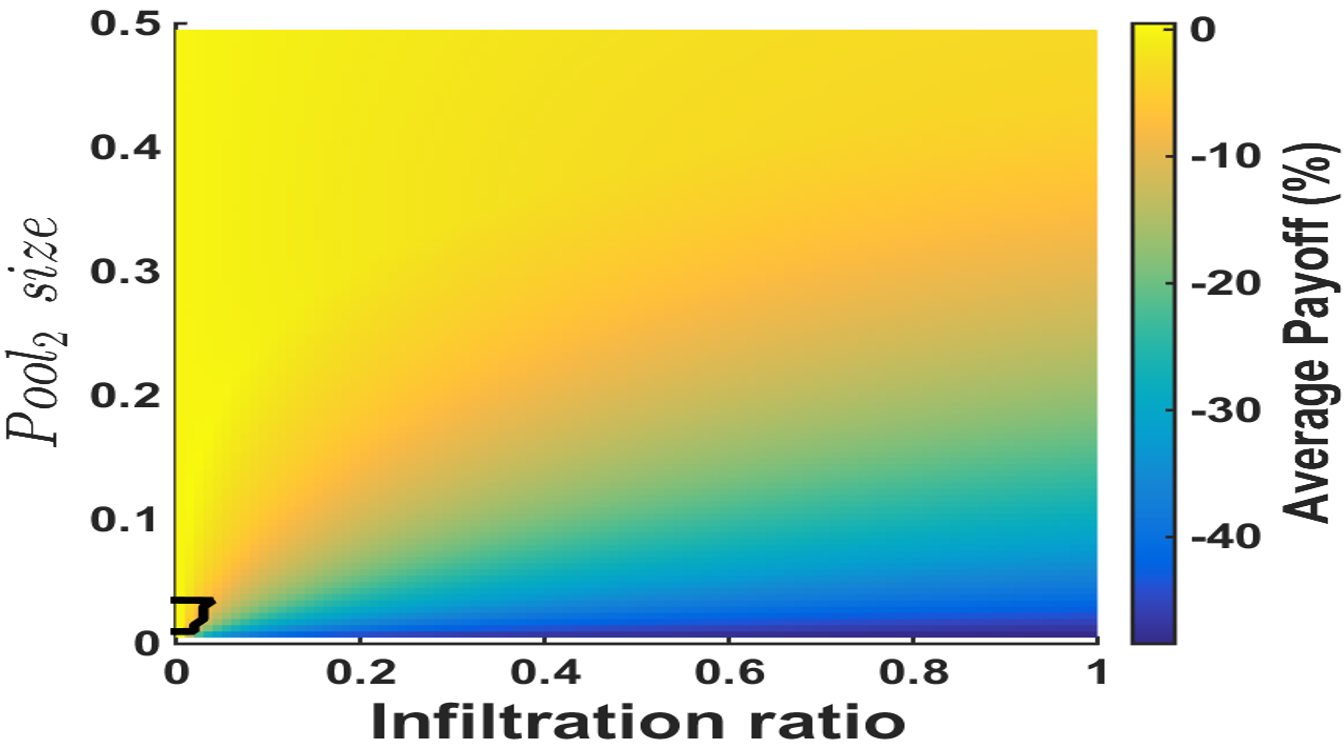}
\label{fig:bwh_reward_p2}}}
\vspace{-2mm}
\caption{Pool$_1$ with computational power of 20\% executes the BWH attack.}
\label{fig:bwh2}\vspace{-5mm}
\end{figure*}

Fig.~\ref{fig:faw} represents when Pool$_1$ optimally executes the FAW attack to maximize its payoff. 
Then, Pool$_2$ following ARS$_{1^-}$ retaliates %with FAW or BWH attacks 
at the next stage.
The $x$ and $y$-axes are Pool$_1$ and Pool$_2$'s sizes, respectively.
Moreover, we define infiltration ratios $r_i=(r_i^F, r_i^B)$, where $r_i^F$ and $r_i^B$ are proportions of infiltration power $f_i$ and $b_i$ for Pool$_i$'s computational power, respectively (i.e., $r_i^F=\frac{f_i}{\alpha_i}$, $r_i^B=\frac{b_i}{\alpha_i}$). 
Fig.~\ref{fig:faw_faw_rt} represents Pool$_2$'s infiltration ratio $r_2^F$ for retaliation using the FAW attack.
In the white region of Fig.~\ref{fig:faw_faw_rt}, Pool$_2$ cannot retaliate against Pool$_1$ with the FAW attack. Thus, Pool$_2$ should retaliate using the BWH attack.
Fig.~\ref{fig:faw_bwh_rt} represents Pool$_2$'s infiltration ratio $r_2^B$ for retaliation using the BWH attack. Here, we can see that all cases are covered with the colored regions in Fig.~\ref{fig:faw_faw_rt} and \ref{fig:faw_bwh_rt}.
Considering two stages where Pool$_1$ first executes the FAW attack and then Pool$_2$ retaliates,
Fig.~\ref{fig:faw_reward_a} and \ref{fig:faw_reward_p} represent average payoffs of Pool$_1$ and Pool$_2$ for two stages, respectively. 
That is, these figures show $\frac{U_i(a^0)+U_i(a^1)}{2}$ when we denote each of two stages by stage 0 and 1.
As shown in Fig.~\ref{fig:faw_reward_a}, Pool$_1$'s average payoff is always negative, meaning that 
ARS$_{1^-}$ makes FAW attacks unprofitable.
Moreover, Fig.~\ref{fig:faw_reward_p} shows that Pool$_2$ can completely recover a loss from Pool$_1$'s attack in the case where Pool$_2$ retaliates with the FAW attack.

Fig.~\ref{fig:bwh} represents when Pool$_1$ optimally executes the BWH attack to maximize its short-term payoff.
Similar to Fig.~\ref{fig:faw}, Fig.~\ref{fig:bwh_faw_rt} and \ref{fig:bwh_bwh_rt} represent Pool$_2$'s infiltration ratio $r_2^F$ and $r_2^B$, respectively. 
Fig.~\ref{fig:bwh_reward_a} and \ref{fig:bwh_reward_p} respectively represent the average payoffs of Pool$_1$ and Pool$_2$ for two stages. 
Fig.~\ref{fig:bwh_reward_a} shows that Pool$_1$ always suffers a loss from the retaliation of Pool$_2$ when Pool$_1$ executes the BWH attack.
Therefore, it shows that ARS makes BWH attacks unprofitable.

As a representative scenario, we simulate the repeated FAW-BWH game in terms of various Pool$_1$'s infiltration ratio used for attacks, assuming that Pool$_1$'s size is 0.2 (20\%).
Fig.~\ref{fig:faw2} and \ref{fig:bwh2} represent Pool$_1$'s execution of FAW and BWH attacks, respectively. 
The $x$ and $y$-axes are Pool$_1$'s infiltration ratio used for attack and Pool$_2$'s sizes, respectively.
Fig.~\ref{fig:faw_faw_rt2} and \ref{fig:faw_bwh_rt2} show 
the infiltration ratio $r_2^F$ and $r_2^B$ for retaliation against Pool$_1$, respectively. 
Because the extent of retaliation by ARS depends on the loss caused by the opponent's attack, Pool$_2$'s infiltration 
ratio for retaliation depends on Pool$_1$'s attack infiltration ratio.
Fig.~\ref{fig:faw_reward_a2} and \ref{fig:faw_reward_p2} represent the average payoffs of Pool$_1$ and Pool$_2$, respectively, for 
two stages in which Pool$_1$ executes the FAW attack and then Pool$_2$ retaliates.
Pool$_1$ always suffers a loss by deviating from ARS because all colors in Fig.~\ref{fig:faw_reward_a2} indicate negative values.
Meanwhile, there are some cases in which Pool$_2$ can earn extra profit in the process of retaliation, as shown in Fig.~\ref{fig:faw_reward_p2}.
Similar to Fig.~\ref{fig:faw2}, Fig.~\ref{fig:bwh2} shows Pool$_2$'s infiltration ratio $r_2^F$ for retaliation, and the attacker's and victim's average payoffs for two stages when Pool$_1$ executes the BWH attack. 
In most cases, Pool$_2$ chooses the FAW attack rather than the BWH attack for retaliation. Even though there exist some cases to execute the BWH attack in response to ARS$_{1^-}$, we omit Pool$_2$'s infiltration ratio $r_2^B$ for retaliation with the BWH attack because the region of such cases is very small (see small areas bounded by black bold lines at left-bottom corners in Fig.~\ref{fig:bwh2}). 
As a result, BWH attacks become unprofitable by ARS.

\begin{table}[ht]
\centering
\vspace{-1mm}
\caption{\small Considering the current power distribution~\cite{pools} and assuming that BTC.com is an attacker, this table lists infiltration ratios $r_i=(r_i^F,r_i^B)$ that four pools use for retaliation according to ARS$_{1^-}$.}
\vspace{-0.3cm}
\resizebox{0.9\columnwidth}{!}{  
\begin{tabular}{|c|c|c|c|c|}
\hline
\parbox[c][0.8cm][c]{1.1cm}{\centering \textbf{Name}} 
&\parbox[c][0.8cm][c]{1.8cm}{\centering \textbf{$(r_i^F,r_i^B)(\%)$ against FAW}} &\parbox[c][0.8cm][c]{1.0cm}{\centering \textbf{Total Payoff }}  & \parbox[c][0.8cm][c]{1.9cm}{\centering \textbf{$(r_i^F,r_i^B)(\%)$ against BWH}} &\parbox[c][0.8cm][c]{1.0cm}{\centering \textbf{Total Payoff}}
\\ \hline \hline
\parbox[c][0.35cm][c]{1.1cm}{\centering AntPool}
& \parbox[c][0.35cm][c]{1.8cm}{\centering $(0, 14.33\%)$} &
\parbox[c][0.35cm][c]{1.0cm}{\centering -1.89\%}
& \parbox[c][0.35cm][c]{1.9cm}{\centering $( 46.2\%,0)$} &\parbox[c][0.35cm][c]{1.0cm}{\centering -0.78\%}
\\ \hline
\parbox[c][0.35cm][c]{1.1cm}{\centering ViaBTC} 
& \parbox[c][0.35cm][c]{1.8cm}{\centering $(0,13.7\%)$} 
& \parbox[c][0.35cm][c]{1.0cm}{\centering -0.54\%}
& \parbox[c][0.35cm][c]{1.9cm}{\centering $(47.2\%,0)$} &\parbox[c][0.35cm][c]{1.0cm}{\centering -0.15\%}
\\ \hline
\parbox[c][0.35cm][c]{1.1cm}{\centering DPool}
& \parbox[c][0.35cm][c]{1.8cm}{\centering $(0, 17.71\%)$}
& \parbox[c][0.35cm][c]{1.0cm}{\centering -0.004\%}
& \parbox[c][0.35cm][c]{1.9cm}{\centering $(0,13.14\%)$}&\parbox[c][0.35cm][c]{1.0cm}{\centering -1.1\%}
\\ \hline
\parbox[c][0.35cm][c]{1.1cm}{\centering Bixin}
& \parbox[c][0.35cm][c]{1.8cm}{\centering $(0,21\%)$}
&\parbox[c][0.35cm][c]{1.1cm}{\centering -0.025\%}
& \parbox[c][0.35cm][c]{1.9cm}{\centering $(0,13\%)$}
&\parbox[c][0.35cm][c]{1.0cm}{\centering -0.63\%}
\\ \hline
\multicolumn{5}{c}{} \\
\end{tabular}} 
\label{tab:pools}
\vspace*{-1mm}
\end{table}

Also, we consider the current power distribution obtained from Blockchain.info~\cite{pools}.
We assume that BTC.com, which is the largest pool as of Jan 2019 and has a computational power of about 25\%, optimally executes FAW and BWH attacks against each of four pools (AntPool, ViaBTC, DPool, and Bixin), which have respective computational powers of 15\%, 10\%, 3.5\%, and 2\%. 
In this case, four pools would retaliate against BTC.com according to ARS.
Table~\ref{tab:pools} represents the infiltration ratio $r_i=(r_i^F,r_i^B)$, 
which Pool$_i$ uses for retaliation with FAW and BWH against BTC.com, respectively.
The second and fourth columns show how each pool should retaliate against BTC.com's FAW and BWH attacks, respectively.
The third and fifth columns represent the attacker's total payoff for each victim pool when the attacker executes FAW and BWH attacks, respectively. 
As shown in Table~\ref{tab:pools}, by retaliating according to ARS$_{1^-}$, the four pools make the attacks of BTC.com unprofitable.
\section{Identifying the opponent's attack}
\label{sec:detection}

To follow ARS, Pool$_i$ needs to know seven parameters in Table~\ref{tab:par}: $\alpha_i$, ${\tt stnd_i}$, $a_i^{t-1}$, $\alpha_{-i}$, ${\tt stnd_{-i}}$, $a_{-i}^{t-1}$, and $a_{-i,ars}^{t-1}$.
In this section, we describe how the pool can obtain these seven parameters, which make it possible for pools to adopt ARS.
Among these parameters, Pool$_i$ already knows $\alpha_i$, ${\tt stnd_i}$, and $a_i^{t-1}$, which are referred to as \textit{internal variables} in this paper. 
Also, Pool$_i$ can easily obtain $\alpha_{-i}$ because the computational power of pools can be approximately calculated from the mined block information~\cite{pools}. 
Among the remaining parameters, ${\tt stnd_{-i}}$, $a_{-i}^{t-1}$, and $a_{-i,ars}^{t-1}$, Pool$_i$ needs to know $a_{-i}^{t-1}$ and $a_{-i,ars}^{t-1}$ because ${\tt stnd_{-i}}$ is determined by $a_{-i}^{t-1}$ and $a_{-i,ars}^{t-1}$.
Moreover, the value $a_{-i,ars}^{t-1}$ is $\bm{\overline{0}}$ if $t-1$ is 0.
If $t-1$ is positive, the value can be obtained from pools' actions at stage $t-2$.
In other words, Pool$_i$ can determine $a_{-i,ars}^{t-1}$ by obtaining the opponent's action $a_{-i}^{t-2}$ at stage $t-2$.
As a result, Pool$_i$ only needs to know the opponent's previous action in order to determine ${\tt stnd_{-i}}$, $a_{-i}^{t-1}$, and $a_{-i,ars}^{t-1}$.

To guess the opponents' actions, Pool$_i$ can plant moles in other pools. Through the moles, Pool$_i$ can observe other pools' average reward densities and stochastically determine other pools' actions from their observed average reward densities. 
However, it may take a long time to find out other pools' actions with their average reward densities. 
Note that, if the time duration of a stage increases, the discount factor $\delta$ would be decreased because pools might focus on the increase of short-term advantages rather than long-term value. 
This implies that it is important to shorten the time duration of a stage. In the following section, we describe how to achieve this.

\begin{table}[ht] 
\centering
\caption{List of parameters.}
\vspace{-4mm}
\label{tab:par}
\begin{small}
\renewcommand{\tabcolsep}{1pt}
\begin{tabular}{|>{\centering\arraybackslash}m{2cm}|m{6cm}|}
\hline
\textbf{Notation} & \multicolumn{1}{c|}{\textbf{Definition}} \\ \hline\hline
$\alpha_i$ & Computational power of Pool$_i$ \\ \hline 
${\tt stnd_i}$ &  Standing of Pool$_i$ \\ \hline
$a_i^{t-1}$ & The action of Pool$_i$ at time $t-1$ \\ \hline
$\alpha_{-i}$ & Computational power of the opponent \\ \hline
${\tt stnd_{-i}}$ & Standing of the opponent \\ \hline
$a_{-i}^{t-1}$ & The action of the opponent at time $t-1$ \\ \hline
$a_{-i,ars}^{t-1}$ & Output of ARS of the opponent at time $t-1$ \\ \hline
\end{tabular}\end{small}
\end{table}

\smallskip
\noindent{\bf Detection of Attacks: }
First, Pool$_i$ must determine whether it is being attacked.
Indeed, Pool$_i$ can easily detect FAW attacks because the fork rate increases~\cite{kwon2017selfish}. 
Moreover, the manager of victim pool can detect BWH attacks by investigating the ratio between the number of submitted shares and the number of found blocks~\cite{eyal2015miner, luu2015power}.

Furthermore, we found that the period of attack detection can be reduced if the manager can focus only on the shares submitted by \textit{unlucky miners}\footnote{Miners who submitted a relatively small number of FPoWs} and \textit{short-term miners}\footnote{Miners who stayed in the pool for relatively short period of time}, rather than all shares.
For example, if a victim pool's size is 20\% and a BWH attacking pool infiltrates 0.5\% into the pool, the victim pool would find 20.1\% (0.2/0.995) of all blocks on average.
Note that the victim pool would find 20.5\% of all blocks on average in the case that the infiltrating power of 0.5\% belongs to benign miners. 
When the manager of the victim pool considers entire 2000 blocks (about two weeks in Bitcoin), there is an approximately 35.82\% probability 
that the ratio of the victim pool's found blocks to all blocks is 0.201, under the assumption that an attack had not occurred.
Therefore, it is difficult for the manager to determine whether an attack has occurred.
However, if the manager considers the shares submitted by only unlucky and short-term miners, he would be able to find these shares that account for about 0.5\% of the entire computational power. 
Assuming that the attack had not occurred, the probability that these shares do not contain FPoWs is about 0.0045\%. Therefore, the manager can successfully detect \textit{attacks} and estimate the \textit{infiltration power} used for attacks. 
If a pool does not detect attacks, other pools' actions are set as $\bm{\overline{0}}$. Otherwise, the pool needs to identify the attacking pool.

\smallskip
\noindent{\bf Identification of Attackers: } The FAW or BWH
attacker receives a portion of the reward earned by the victim pool. 
Therefore, the attacker's reward is related to the number of blocks generated from the victim pool. 
We denote by $P$ the period in which the attacker finds
one block in its pool. At the end of $P$, miners in the attacker's pool
receive a portion of the reward for the one block and the rewards gained
from the victim pool for period $P$. For ease of presentation, we
use $\alpha$ and $\beta$ to represent the attacker's size and the victim's size, respectively, instead of $\alpha_{-i}$ and $\alpha_i$. We also suppose the attacker infiltrates a fraction $\gamma$ of the attacker's computational power into the victim pool. Moreover, we denote by $Rd_p$ the earned reward
density of miners belonging to the attacking pool for the period $P$. If the number of generated blocks from the victim pool for the period $P$ is $N$, $Rd_p$ would be
$\frac{1}{\alpha}+\frac{N\gamma}{\beta+\gamma\alpha}$.  
We can easily check that $N$ has a geometric distribution with a parameter 
$\frac{(1-\gamma)\alpha}{\beta+\gamma\alpha(1-\alpha-\beta)+(1-\gamma)\alpha}$
for FAW attack, and
$\frac{(1-\gamma)\alpha}{\beta+(1-\gamma)\alpha}$ for
BWH attack (see Theorem~\ref{thm:geo} in Appendix~\ref{subsec:number} for details). 

If $\alpha$ does not change and the pool does not attack, the reward density $Rd_p$ for the period $P$ is fixed at $\frac{1}{\alpha}$. 
Meanwhile, if the pool attacks, $Rd_p$ would continuously change depending on $N$.
Therefore, the victim can identify the attacking pool immediately by observing their variances in $Rd_p$ after planting moles in pools.
However, even if attacks do not occur, the reward density as well as $\alpha$ and $\beta$ changes in practice.
If $\alpha$ usually has a large variance, 
it would indeed be difficult to identify the attacker by investigating the variance in the reward density.
Thus, to determine this, we should find out how much is the variance in pools' computational powers in real world when attacks do not occur.
To this end, we collected hash rates from two pools (ViaBTC~\cite{viabtc} and BTC.com~\cite{btc.com}) by monitoring their hash rates over one month (Jan. 21, 2019$\sim$Feb. 18, 2019).
Two pools publicly provide their average hash rate (PH/s) for one hour. 
Using the data, we first normalized their computational power and then calculated the reward density when assuming that these pools are benign, as a reciprocal number of computational power (e.g., if a pool's computational power is 0.2, its reward density $Rd_p$ would be 5).
The reward density at time $t$ indicates how much of the reward per computational power a pool miner can earn when the corresponding pool finds one block at time $t$. 
Note that the value of $\frac{Rd_p}{E[P]}$ for each miner is the 
same regardless of pool when all miners are honest, where $E[P]$ is the mean of $P$.

We simulated two scenarios with the FAW attack: (1) first scenario (attack pool: ViaBTC) and (2) second scenario (attack pool: BTC.com) in which the attack pool executes the FAW attack with an infiltration power of 0.005 (0.5\%) against the victim pool with a computational power 0.2 (20\%).
Fig.~\ref{fig:detection} shows the reward densities $Rd_p$ of two pools. 
Blue and orange lines represent reward densities observed when each pool executes the FAW attack and when they are benign, respectively.
In Fig.~\ref{fig:detection}, blue lines fluctuate noticeably more than the orange lines.
By running the simulation several times, we observed that the variation in the blue lines was usually about 99$\sim$152 and 24$\sim$32 times the variation in the orange lines for ViaBTC and BTC.com, respectively.

We also simulated two scenarios with the BHW attack: (1) first scenario (attack pool: ViaBTC) and (2) second scenario (attack pool: BTC.com) in which the attack pool executes the BWH attack with an infiltration power of 0.005 against the victim pool with a computational power 0.2.
Similarly, in two cases, the variance in reward densities of the attacker when the attack is executed increases about 85$\sim$177 and 22$\sim$36 times that for when the attack is not executed, respectively.
For space reasons, we omit figures representing the simulation results similar to Fig.~\ref{fig:detection}.
\textit{These results show that if a pool is a FAW or BWH attacker, the pool's reward density $Rd_p$ would usually increase (or decrease) when $N$ increases (or decreases).} 
Therefore, the victim can identify the attacker by observing the variance in reward densities in other pools and by comparing it with the number of blocks found in his pool.
Probably, the attacker would try several methods to reduce the variance in the attacker's reward density to hide the attack evidence. 
However, we note that the victim can still observe relatively high variance in the attacker's reward density even when such evasion methods are used (see Appendix~\ref{sec:evade}). 

\begin{figure}[t]
\hspace*{-2mm}
\centering{
\subfloat[ViaBTC's reward densities]{
\includegraphics[width=0.33\textwidth]{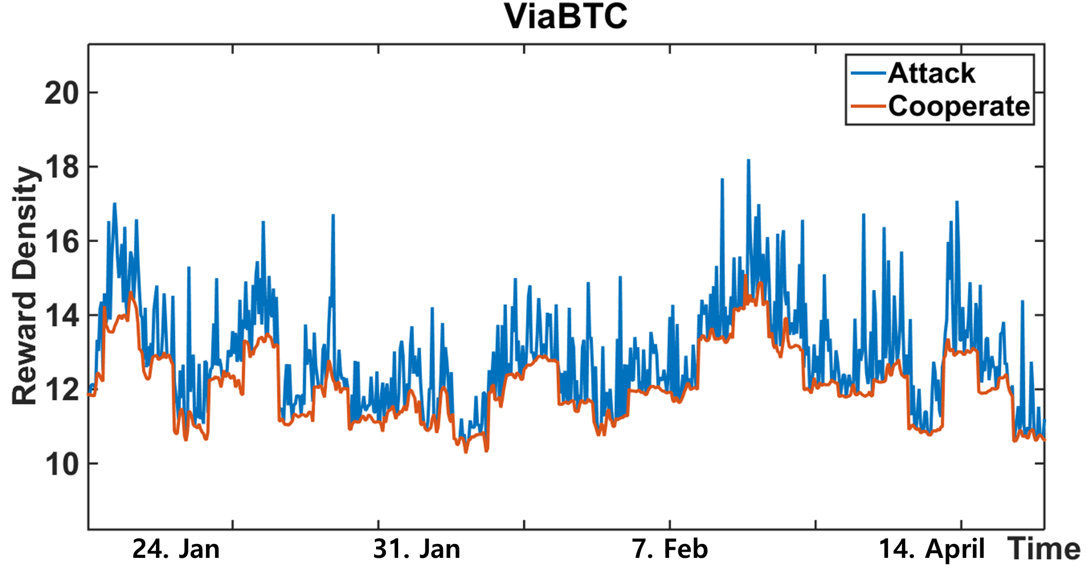}
\label{fig:viabtc}
}\\
\subfloat[BTC.com's reward densities]{
\includegraphics[width=0.33\textwidth]{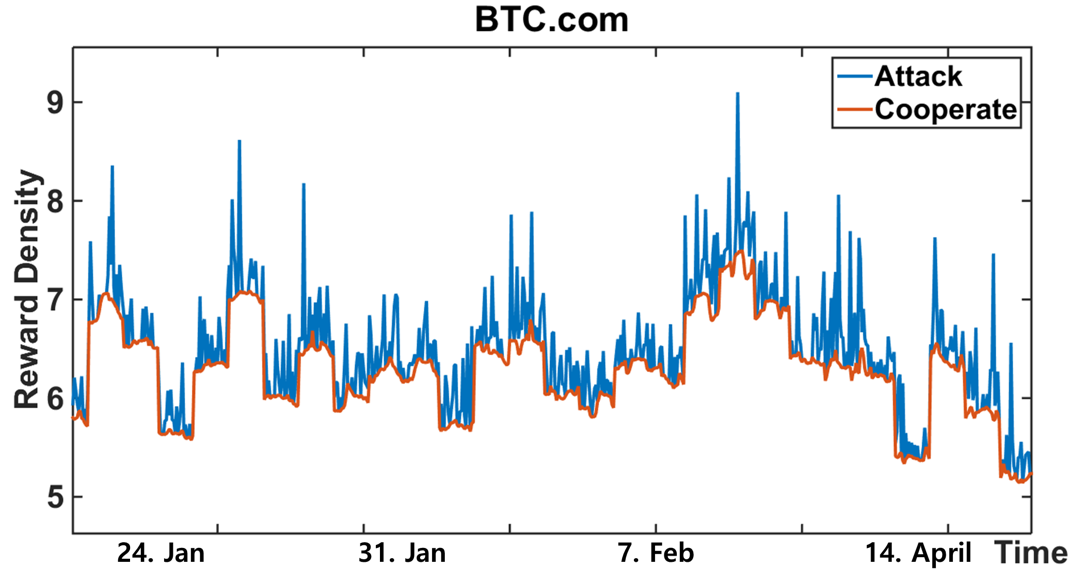}
\label{fig:btc_com}}}
 \vspace{-3mm}
\caption{Reward densities of two pools. In each figure, the blue and orange lines represent each pool's reward density when the pool executes FAW attacks and when the pool cooperates, respectively.}
\label{fig:detection} 
 \vspace{-1mm}
\end{figure}

\smallskip
\noindent\textbf{Summary:}
In summary, the victim first determines whether attacks occurred by investigating the fork rate and the ratio of found blocks to submitted shares.
This method also allows the victim to estimate the infiltration power in the victim pool.
The victim then identifies the attacker by investigating the variance in $Rd_p$ in other pools.
Using the above two methods, a pool can identify other pools' actions. Therefore, a pool can find out all seven parameters in Table~\ref{tab:par} to run ARS.
\section{Multiple Pools}
\label{sec:npools}

\begin{table*}[ht]
\centering
\caption{\small The simulation results of ARS for multiple pools.
The third and sixth columns show the optimal infiltration ratio for BTC.com's 
FAW and BWH attacks, respectively. 
The fourth and seventh columns show each pool's infiltration ratio
for retaliation against BTC.com's FAW and BWH attacks, respectively.
Finally, the fifth (or eighth) column represents BTC.com's payoff for two stages in 
which BTC.com first executes the FAW (or BWH) attack and then four pools 
retaliate against BTC.com. \vspace{-2mm}}
\resizebox{.95\textwidth}{!}{ 
\begin{tabular}{|c|c|c|c|c|c|c|c|}
  \hline
\parbox[c][0.5cm][c]{1.3cm}{\centering \textbf{Name}} & \parbox[c][0.5cm][c]{3.3cm}{\centering \textbf{Computational Power}} & \parbox[c][0.5cm][c]{0.8cm}{\centering \textbf{$r_{1j}^F$}}&
\parbox[c][0.5cm][c]{3.3cm}{\centering \textbf{$(r_{j1}^F, r_{j1}^B)$ against FAW}} & \parbox[c][0.5cm][c]{0.8cm}{\centering \textbf{$\mathcal{U}_1$}}& \parbox[c][0.5cm][c]{0.8cm}{\centering \textbf{$r_{1j}^B$}}& \parbox[c][0.5cm][c]{3.3cm}{\centering \textbf{$(r_{j1}^F, r_{j1}^B)$ against BWH}}& \parbox[c][0.5cm][c]{0.8cm}{\centering \textbf{$\mathcal{U}_1$}} \\ \hline \hline
\parbox[c][0.35cm][c]{1.3cm}{\centering AntPool} & \parbox[c][0.35cm][c]{3.3cm}{\centering 15\%} & \parbox[c][0.35cm][c]{0.8cm}{\centering $22.7\%$}& \parbox[c][0.35cm][c]{3.3cm}{\centering $(0, 13.67\%)$}& 
\multirow{4}{0.8cm}{\centering $-5.4 \%$}
& \parbox[c][0.35cm][c]{0.8cm}{\centering $9.5\%$} & \parbox[c][0.35cm][c]{3.3cm}{\centering $( 46.2\%, 0)$}& \multirow{4}{0.8cm}{\centering $-1.55 \%$}
 \\ %\hline
\parbox[c][0.35cm][c]{1.3cm}{\centering ViaBTC} & \parbox[c][0.35cm][c]{3.3cm}{\centering 10\%} & \parbox[c][0.35cm][c]{0.8cm}{\centering $15.1\%$} & \parbox[c][0.35cm][c]{3.3cm}{\centering $(0, 13.7\%)$}& 
& \parbox[c][0.35cm][c]{0.8cm}{\centering $6.4 \%$}& \parbox[c][0.35cm][c]{3.3cm}{\centering $( 47.2\%,0)$}& \\ %\hline
\parbox[c][0.35cm][c]{1.3cm}{\centering DPool}
& \parbox[c][0.35cm][c]{3.3cm}{\centering 3.5\%} & \parbox[c][0.35cm][c]{0.8cm}{\centering $5.3\%$} & \parbox[c][0.35cm][c]{3.3cm}{\centering $(0,13.14\%)$}
&& \parbox[c][0.35cm][c]{0.8cm}{\centering $2.2 \%$}& \parbox[c][0.35cm][c]{3.3cm}{\centering $(0,13.14\%)$}& 
\\ %\hline
\parbox[c][0.35cm][c]{1.3cm}{\centering Bixin}
& \parbox[c][0.35cm][c]{3.3cm}{\centering 2\%} & \parbox[c][0.35cm][c]{0.8cm}{\centering $3\%$} & \parbox[c][0.35cm][c]{3.3cm}{\centering $(0,13\%)$}& 
& \parbox[c][0.35cm][c]{0.8cm}{\centering $1.3\%$}& \parbox[c][0.35cm][c]{3.3cm}{\centering $(0, 13\%)$}&  \\ \hline
\end{tabular}}
\label{tab:npools}
\vspace{-0.2cm}
\end{table*}

\begin{figure}[!t]
 \removelatexerror
\scalebox{0.9}{
\begin{minipage}{1.1\linewidth}
\begin{algorithm}[H] 
  \caption{ARS for multiple pools.}
 \par\noindent\hrulefill \\
 {\bf ARS$_K$ for each pool $i$ against pool $j \neq i$}
  \vspace{-0.15cm}
 \par\noindent\hrulefill
\begin{compactenum}[$~$]
 \item[\bf Start when $t=0$:] $~$\\
 \begin{compactenum}[$~$]
   \item Start the stage game with no attack, (i.e., $a_{ij}^0 = \bm{\overline{0}}$), 
   and set a variable ${\tt stnd_{ij}=\good.}$
   \end{compactenum}
 \item[{\bf At each stage $t \geq 1$}:]  \label{step:select} $ $\\
   \begin{compactenum}[\hspace{4mm}\bf \em S1.]
\item {\em Set the standing of this stage.} \\
     If ($a_{ij}^{t-1} == a_{ij,ars}^{t-1}$),
       ${\tt stnd_{ij}=}\good,$ else ${\tt stnd_{ij}=}\bad.$
    \item {\em Estimate the infiltration power.} \\    
    If (${\tt stnd_{ij}} == \good$) and (${\tt stnd_{ji}}== \bad$) \\
      \hspace{0.6cm}$a^t_{ij}=\text{\bf Retaliate$_K$}\left(\alpha_i, a^{t-1}_{ij},\alpha_{j},a^{t-1}_{ji},a^{t-1}_{ji,ars}\right)$ \\
      else \hspace{0.08cm}$a_{ij}^t = \bm{\overline{0}}.$
    \item Output $a_{ij}^t.$
\end{compactenum}
\end{compactenum}
  \par\noindent\hrulefill	
  \vspace{0.1in}
  \label{al:npools}
\end{algorithm}
\end{minipage}}
\end{figure}

In this section, we present ARS for $n$ pools (Pool$_i$: $i=1\sim n).$ 
To this end, we now specify, for all notations of Pool$_i$'s standing and action, those against Pool$_j$ for each $j \neq i,$ such as $\staij,$ $a_{ij}^t,$ and $a_{ij,ars}^t.$ Then, Pool$_i$ maintains $n-1$ dimensional standing vectors $(\staij)_{j\neq i}$ and action vectors $(a_{ij}^t)_{j \neq i},$ $(a_{ij,ars}^t)_{j \neq i}.$ Note that $a_{ij}^t = (f_{ij}^t, b_{ij}^t),$ where $f_{ij}^t$ and $b_{ij}^t$ are the infiltration powers for FAW and BWH attacks, respectively. ARS$_K$ of Pool$_i$ against Pool$_j$ is described in Algorithm~\ref{al:npools} (similar to Algorithm~\ref{al:ars}).
When Pool$_i$ follows ARS$_K$, the pool retaliates against Pool$_j$ only if $\staij$ and ${\tt stnd_{ji}}$ are $\good$ and $\bad,$ respectively, where {\bf Retaliate} in Algorithm~\ref{al:npools} is similar to that in 
Algorithm~\ref{al:rt}.
The difference between {\bf Retaliate} for $n$ and two pools is that 
$n$-pool {\bf Retaliate} simply replaces $f_i$ and $b_{i}$ in Algorithm~\ref{al:ars} with $f_{ij}$ and $b_{ij}$, and outputs $a_{ij}^t$. 

For the identification of attackers, a mechanism similar to that described in Section~\ref{sec:detection} can be applied. 
Even if multiple attackers execute attacks against multiple victims in parallel, each victim can find out who the attacker is because
the corresponding attacker's reward density fluctuates 
depending on \textit{the number of blocks found by the corresponding victim.} 
In addition, when multiple pools target a victim, the victim can identify those attackers and estimate their infiltration rates because each attacker's reward density variance depends on its infiltration rate. 
More specifically, the larger infiltration rate is, the higher variance in the attacker's reward density is. 
To distinguish FAW from BWH, the victim can investigate a fork rate. 
Note that, because the FAW attack is to intentionally generate a fork with blocks generated by a victim and another miner, the fork rate from the FAW attack pool is relatively low compared with those of other pools.

It seems natural to expect that, similar to ARS for two pools, ARS for $n$ pools will make attacks unprofitable and thus induce cooperation, despite the mathematical challenges in formally proving it due to complex inter-coupling among $n$ pools. 
We numerically conduct this analysis through the simulation of the scenario in which BTC.com (Pool$_1$) possessing 25\% computational power optimally 
executes FAW or BWH attacks against four other pools (AntPool (Pool$_2$), ViaBTC (Pool$_3$), DPool (Pool$_4$), and Bixin (Pool$_5$)) at the same time while maximizing its short-term payoff, and four other pools follow ARS$_{1^-}$ (see Table~\ref{tab:npools}).
In Table~\ref{tab:npools}, we denote the infiltration ratios for FAW and BWH attacks of Pool$_i$ against Pool$_j$ by $r_{ij}=(r_{ij}^F,r_{ij}^B)$.
The third and sixth columns show the optimal infiltration ratios 
$r_{1j}^F$ and $r_{1j}^B$ ($j=2\sim5$) of BTC.com for FAW and BWH attacks, respectively; these ratios maximize the payoff of BTC.com. 
The fourth and seventh columns show each pool's infiltration ratio
for retaliation against BTC.com's FAW and BWH attacks, respectively.
Finally, the fifth (or eighth) column represents BTC.com's payoff for two stages in which BTC.com first executes the FAW (or BWH) attack and then four pools retaliate against BTC.com. For example, when BTC.com executes FAW attacks using infiltration ratios of 22.7\%, 15.1\%, 5.3\%, and 3\% against AntPool, ViaBTC, DPool, and Bixin, respectively, the four pools would retaliate against BTC.com at the next stage according to the fourth column. 
In this case, BTC.com suffers a loss of $5.4$ \% for two stages in the aggregate. As such, BTC.com's total payoff becomes negative when the pool executes FAW and BWH attacks. 
Considering these results, it becomes unprofitable to attack, and rational pools sustain cooperation without attacks when they follow ARS. 

\section{Discussion}
\label{Discussion}
\subsection{Closed Pools and Solo Miners}

Even though we focused on attacks executed only by open pools in this paper, solo miners or closed pools can also execute FAW and BWH attacks in practice.
If a solo miner or closed pool is an attacker, the victim cannot retaliate 
against the attacker because the victim cannot infiltrate its moles. 
This fact may lead to a rational solo miner or closed pool to execute FAW and/or BWH attacks.
However, fortunately, it is widely known that solo miners and closed pools have limited computational power.
To estimate the current size of solo miners or closed pools, we observed the hashrate distribution in websites given from BTC.com~\cite{pools2} at the time of writing (Jan. 2019), and we found the total of 22 pools. Among them, BitFury and 58COIN are only closed pools with about 3.1\% and 1.3\% computational powers, respectively.
If each of them executes FAW attacks, separately, against BTC.com, which has about 
25\% computational power as the largest mining pool,
BitFury and 58COIN earn extra reward densities of 0.74\%,
0.32\%, respectively.
Moreover, BTC.com suffers losses of 0.09\% and 0.016\% from each attack. 
However, as we can see, the impacts of those attacks seem rather marginal.

\subsection{Infiltration Power}

Pool$_i$'s infiltration power into other pools should be loyal to Pool$_i$~\cite{eyal2015miner, kwon2017selfish}. 
The loyal power can be the manager's own computational power or cloud mining~\cite{cloud}, or the computational power of miners with a private relation to the manager.
Moreover, pools' loyal power ratios are trade secrets.
In Section~\ref{sec:numerical}, there exist some cases in which the infiltration ratio for retaliation against the attacker is greater than 80\%.
However, these are extreme cases in which the victim pool's size is very small and the attacker's size is close to 50\%. 
In the current computational power distribution, the infiltration power for retaliation is less than 50\%, as shown in Table~\ref{tab:pools}. 
However, it is possible that a pool has a loyal power of less than 50\%.
For example, if AntPool has a loyal power of less than 46.2\%, the pool cannot prevent the BWH attack of BTC.com through retaliation with the FAW attack. 
Therefore, to retaliate against BTC.com, Antpool can execute the BWH attack instead of the FAW attack.
In this case, Antpool needs only about 14.33\% infiltration ratio to prevent 
the BWH attack through retaliation with the BWH attack of its own.

\subsection{Sabotage Attack}

We showed that rational pools can cooperate by making attacks unprofitable through ARS.
However, FAW and BWH attacks can still be executed by a large pool to disable a small pool (called \textit{Sabotage} attack) because a small pool's loss is greater than a large pool's loss when the large pool deviates from ARS. This situation is commonly known as the \emph{Chicken Game}. 
If the small pool stops the operation, the large pool would be the winner despite of the loss in their reward because the small pool's miners may migrate into the large pool.
However, if two pools among several pools launch attacks
against each other, other pools not involved in the attack do not suffer losses stemming from the attack. Therefore, even if the small pool ceases the mining operation, the small pool's miners would move into other pools rather than the large pool. As a result, the large pool cannot earn a direct profit through sabotage attacks.

\section{Related Work}
\label{Related}

Game theory has been used for analyzing attacks and protocols in Bitcoin.
Kroll et al.~\cite{kroll2013economics} analyzed the economics of 
Bitcoin mining under the assumption that all miners are rational.
They showed that there is a Nash equilibrium in which all miners comply with the Bitcoin protocol when considering a 51\% attack. 
Several studies~\cite{johnson2014game, laszka2015bitcoin} modeled a game in which two pools decide whether to trigger a DDoS 
attack against an opponent.
Lewenberg et al.~\cite{lewenberg2015bitcoin} considered 
miners' interactions among mining pools as a cooperative game.
They found that some miners would always switch among pools for their profit if the communication delay in the network is large.
Luu et al.~\cite{luu2015power} modeled a power splitting game to analyze how an attacker can optimally execute the BWH attack against multiple pools.
Moreover, Eyal~\cite{eyal2015miner} showed that a BWH game between two pools results in the prisoner's dilemma.
Carlsten et al.~\cite{carlsten2016instability} and Tsabary et al.~\cite{tsabary2018gap} analyzed a game among miners when miners earn only transaction fees as block rewards in the future.
Kwon et al.~\cite{kwon2017selfish} proposed the FAW attack and analyzed the FAW game between two pools in which can break the prisoner's dilemma. 
Yoo et al.~\cite{yoo2018promoting} studied an incentive design in proof-of-work blockchains, considering a cooperative and non-cooperative strategy of miners.
Kwon et al.~\cite{kwon2019bitcoin} analyzed a rational behavior of miners when two coins with a compatible proof-of-work mechanism exist.
To the best of our knowledge, this paper is the first attempt to consider and analyze both FAW and BWH attacks together in a repeated game. 
Moreover, we propose infinitely many strategies inducing cooperation among rational pools under the presence of both FAW and BWH attacks.
\section{Conclusion}
\label{sec:conclusion}
In this paper, by modeling a repeated game, we studied how pools can cooperate to avoid the health of systems being weakened from the attacks. 
Because the stage game, FAW-BWH game, highly differs from the prisoner's dilemma, it may be challenging to find a strategy inducing cooperation among pools. 
To solve this challenging problem, we proposed novel infinitely many strategies, called ARS, which are likely to be adopted by rational pools. 
In ARS, a pool first cooperates and then retaliates against the attacker in the case when attacks occur. 
ARS provably strikes a good balance between retaliation and selfishness. 
Moreover, there are several parameters required to use ARS in practice. 
Thus, we discuss the methods to determine the parameters, investigating the real-world data collected from mining pools.
As a result, ARS makes cooperation among pools stable, sustainable, and recoverable. 

%%% Local Variables:
%%% mode: latex
%%% TeX-master: "main"
%%% End:

\bibliographystyle{ACM-Reference-Format}
\bibliography{references} 

%%% -*-BibTeX-*-
%%% Do NOT edit. File created by BibTeX with style
%%% ACM-Reference-Format-Journals [18-Jan-2012].

\begin{thebibliography}{32}

%%% ====================================================================
%%% NOTE TO THE USER: you can override these defaults by providing
%%% customized versions of any of these macros before the \bibliography
%%% command.  Each of them MUST provide its own final punctuation,
%%% except for \shownote{}, \showDOI{}, and \showURL{}.  The latter two
%%% do not use final punctuation, in order to avoid confusing it with
%%% the Web address.
%%%
%%% To suppress output of a particular field, define its macro to expand
%%% to an empty string, or better, \unskip, like this:
%%%
%%% \newcommand{\showDOI}[1]{\unskip}   % LaTeX syntax
%%%
%%% \def \showDOI #1{\unskip}           % plain TeX syntax
%%%
%%% ====================================================================

\ifx \showCODEN    \undefined \def \showCODEN     #1{\unskip}     \fi
\ifx \showDOI      \undefined \def \showDOI       #1{#1}\fi
\ifx \showISBNx    \undefined \def \showISBNx     #1{\unskip}     \fi
\ifx \showISBNxiii \undefined \def \showISBNxiii  #1{\unskip}     \fi
\ifx \showISSN     \undefined \def \showISSN      #1{\unskip}     \fi
\ifx \showLCCN     \undefined \def \showLCCN      #1{\unskip}     \fi
\ifx \shownote     \undefined \def \shownote      #1{#1}          \fi
\ifx \showarticletitle \undefined \def \showarticletitle #1{#1}   \fi
\ifx \showURL      \undefined \def \showURL       {\relax}        \fi
% The following commands are used for tagged output and should be
% invisible to TeX
\providecommand\bibfield[2]{#2}
\providecommand\bibinfo[2]{#2}
\providecommand\natexlab[1]{#1}
\providecommand\showeprint[2][]{arXiv:#2}

\bibitem[\protect\citeauthoryear{Axelrod et~al\mbox{.}}{Axelrod
  et~al\mbox{.}}{1987}]%
        {axelrod1987evolution}
\bibfield{author}{\bibinfo{person}{Robert Axelrod} {et~al\mbox{.}}}
  \bibinfo{year}{1987}\natexlab{}.
\newblock \showarticletitle{The evolution of strategies in the iterated
  prisoner’s dilemma}.
\newblock \bibinfo{journal}{{\em The dynamics of norms\/}}
  (\bibinfo{year}{1987}), \bibinfo{pages}{1--16}.
\newblock


\bibitem[\protect\citeauthoryear{{Blockchain Info}}{{Blockchain Info}}{2018}]%
        {pools}
{Blockchain Info} \bibinfo{year}{2018}\natexlab{}.
\newblock \bibinfo{title}{{Hashrate Distribution}}.
\newblock \bibinfo{howpublished}{\url{https://blockchain.info/pools}}.
  (\bibinfo{year}{2018}).
\newblock
\newblock
\shownote{[Online; accessed 02-May-2018].}


\bibitem[\protect\citeauthoryear{Boyd}{Boyd}{1989}]%
        {boyd1989mistakes}
\bibfield{author}{\bibinfo{person}{Robert Boyd}.}
  \bibinfo{year}{1989}\natexlab{}.
\newblock \showarticletitle{Mistakes allow evolutionary stability in the
  repeated prisoner's dilemma game}.
\newblock \bibinfo{journal}{{\em Journal of theoretical Biology\/}}
  \bibinfo{volume}{136}, \bibinfo{number}{1} (\bibinfo{year}{1989}),
  \bibinfo{pages}{47--56}.
\newblock


\bibitem[\protect\citeauthoryear{{{BTC.com}}}{{{BTC.com}}}{2018}]%
        {btc.com}
{{BTC.com}} \bibinfo{year}{2018}\natexlab{}.
\newblock \bibinfo{title}{{{BTC.com}}}.
\newblock \bibinfo{howpublished}{\url{https://pool.btc.com/pool-stats}}.
  (\bibinfo{year}{2018}).
\newblock
\newblock
\shownote{[Online; accessed 30-April-2018].}


\bibitem[\protect\citeauthoryear{Carlsten, Kalodner, Weinberg, and
  Narayanan}{Carlsten et~al\mbox{.}}{2016}]%
        {carlsten2016instability}
\bibfield{author}{\bibinfo{person}{Miles Carlsten}, \bibinfo{person}{Harry
  Kalodner}, \bibinfo{person}{S~Matthew Weinberg}, {and}
  \bibinfo{person}{Arvind Narayanan}.} \bibinfo{year}{2016}\natexlab{}.
\newblock \showarticletitle{{O}n the {I}nstability of {B}itcoin without the
  {B}lock {R}eward}. In \bibinfo{booktitle}{{\em Conference on Computer and
  Communications Security}}. ACM.
\newblock


\bibitem[\protect\citeauthoryear{{{Cloud Mining}}}{{{Cloud Mining}}}{2018}]%
        {cloud}
{{Cloud Mining}} \bibinfo{year}{2018}\natexlab{}.
\newblock \bibinfo{title}{{{bestcloudmining}}}.
\newblock \bibinfo{howpublished}{\url{http://www.bestcloudmining.net/}}.
  (\bibinfo{year}{2018}).
\newblock
\newblock
\shownote{[Online; accessed 2-May-2018].}


\bibitem[\protect\citeauthoryear{Courtois and Bahack}{Courtois and
  Bahack}{2014}]%
        {courtois2014subversive}
\bibfield{author}{\bibinfo{person}{Nicolas~T Courtois} {and}
  \bibinfo{person}{Lear Bahack}.} \bibinfo{year}{2014}\natexlab{}.
\newblock \showarticletitle{On {S}ubversive {M}iner {S}trategies and {B}lock
  {W}ithholding {A}ttack in {B}itcoin {D}igital {C}urrency}.
\newblock \bibinfo{journal}{{\em arXiv preprint arXiv:1402.1718\/}}
  (\bibinfo{year}{2014}).
\newblock


\bibitem[\protect\citeauthoryear{Daian, Eyal, Juels, and Sirer}{Daian
  et~al\mbox{.}}{2017}]%
        {daian2017short}
\bibfield{author}{\bibinfo{person}{Philip Daian}, \bibinfo{person}{Ittay Eyal},
  \bibinfo{person}{Ari Juels}, {and} \bibinfo{person}{Emin~G{\"u}n Sirer}.}
  \bibinfo{year}{2017}\natexlab{}.
\newblock \showarticletitle{{(Short Paper) PieceWork: Generalized Outsourcing
  Control for Proofs of Work}}. In \bibinfo{booktitle}{{\em International
  Conference on Financial Cryptography and Data Security}}. Springer,
  \bibinfo{pages}{182--190}.
\newblock


\bibitem[\protect\citeauthoryear{Eyal}{Eyal}{2015}]%
        {eyal2015miner}
\bibfield{author}{\bibinfo{person}{Ittay Eyal}.}
  \bibinfo{year}{2015}\natexlab{}.
\newblock \showarticletitle{The {M}iner's {D}ilemma}. In
  \bibinfo{booktitle}{{\em Symposium on Security and Privacy}}. IEEE.
\newblock


\bibitem[\protect\citeauthoryear{Eyal and Sirer}{Eyal and Sirer}{2014a}]%
        {twophase}
\bibfield{author}{\bibinfo{person}{Ittay Eyal} {and}
  \bibinfo{person}{Emin~G{\"u}n Sirer}.} \bibinfo{year}{2014}\natexlab{a}.
\newblock \bibinfo{title}{How to {D}isincentivize {L}arge {B}itcoin {M}ining
  {P}ools}.
\newblock   (\bibinfo{year}{2014}).
\newblock
\newblock
\shownote{[Online; accessed 1-May-2017].}


\bibitem[\protect\citeauthoryear{Eyal and Sirer}{Eyal and Sirer}{2014b}]%
        {eyal2014majority}
\bibfield{author}{\bibinfo{person}{Ittay Eyal} {and}
  \bibinfo{person}{Emin~G{\"u}n Sirer}.} \bibinfo{year}{2014}\natexlab{b}.
\newblock \showarticletitle{Majority {I}s {N}ot {E}nough: {B}itcoin {M}ining
  {I}s {V}ulnerable}. In \bibinfo{booktitle}{{\em International Conference on
  Financial Cryptography and Data Security}}. Springer.
\newblock


\bibitem[\protect\citeauthoryear{Fudenberg and Maskin}{Fudenberg and
  Maskin}{2009}]%
        {fudenberg2009folk}
\bibfield{author}{\bibinfo{person}{Drew Fudenberg} {and} \bibinfo{person}{Eric
  Maskin}.} \bibinfo{year}{2009}\natexlab{}.
\newblock \showarticletitle{The folk theorem in repeated games with discounting
  or with incomplete information}.
\newblock In \bibinfo{booktitle}{{\em A Long-Run Collaboration On Long-Run
  Games}}. \bibinfo{publisher}{World Scientific}, \bibinfo{pages}{209--230}.
\newblock


\bibitem[\protect\citeauthoryear{Gervais, Karame, W{\"u}st, Glykantzis,
  Ritzdorf, and Capkun}{Gervais et~al\mbox{.}}{2016}]%
        {gervais2016security}
\bibfield{author}{\bibinfo{person}{Arthur Gervais}, \bibinfo{person}{Ghassan~O
  Karame}, \bibinfo{person}{Karl W{\"u}st}, \bibinfo{person}{Vasileios
  Glykantzis}, \bibinfo{person}{Hubert Ritzdorf}, {and} \bibinfo{person}{Srdjan
  Capkun}.} \bibinfo{year}{2016}\natexlab{}.
\newblock \showarticletitle{On the {S}ecurity and {P}erformance of {P}roof of
  {W}ork {B}lockchains}. In \bibinfo{booktitle}{{\em Conference on Computer and
  Communications Security}}. ACM.
\newblock


\bibitem[\protect\citeauthoryear{Johnson, Laszka, Grossklags, Vasek, and
  Moore}{Johnson et~al\mbox{.}}{2014}]%
        {johnson2014game}
\bibfield{author}{\bibinfo{person}{Benjamin Johnson}, \bibinfo{person}{Aron
  Laszka}, \bibinfo{person}{Jens Grossklags}, \bibinfo{person}{Marie Vasek},
  {and} \bibinfo{person}{Tyler Moore}.} \bibinfo{year}{2014}\natexlab{}.
\newblock \showarticletitle{{Game-theoretic analysis of DDoS attacks against
  Bitcoin mining pools}}. In \bibinfo{booktitle}{{\em International Conference
  on Financial Cryptography and Data Security}}. Springer,
  \bibinfo{pages}{72--86}.
\newblock


\bibitem[\protect\citeauthoryear{Karame, Androulaki, and Capkun}{Karame
  et~al\mbox{.}}{2012}]%
        {karame2012double}
\bibfield{author}{\bibinfo{person}{Ghassan~O Karame}, \bibinfo{person}{Elli
  Androulaki}, {and} \bibinfo{person}{Srdjan Capkun}.}
  \bibinfo{year}{2012}\natexlab{}.
\newblock \showarticletitle{{D}ouble-spending {F}ast {P}ayments in {B}itcoin}.
  In \bibinfo{booktitle}{{\em Conference on Computer and Communications
  Security}}. ACM.
\newblock


\bibitem[\protect\citeauthoryear{Kroll, Davey, and Felten}{Kroll
  et~al\mbox{.}}{2013}]%
        {kroll2013economics}
\bibfield{author}{\bibinfo{person}{Joshua~A Kroll}, \bibinfo{person}{Ian~C
  Davey}, {and} \bibinfo{person}{Edward~W Felten}.}
  \bibinfo{year}{2013}\natexlab{}.
\newblock \showarticletitle{The economics of Bitcoin mining, or Bitcoin in the
  presence of adversaries}. In \bibinfo{booktitle}{{\em Proceedings of WEIS}},
  Vol.~\bibinfo{volume}{2013}.
\newblock


\bibitem[\protect\citeauthoryear{Kwon, Kim, Son, Vasserman, and Kim}{Kwon
  et~al\mbox{.}}{2017}]%
        {kwon2017selfish}
\bibfield{author}{\bibinfo{person}{Yujin Kwon}, \bibinfo{person}{Dohyun Kim},
  \bibinfo{person}{Yunmok Son}, \bibinfo{person}{Eugene Vasserman}, {and}
  \bibinfo{person}{Yongdae Kim}.} \bibinfo{year}{2017}\natexlab{}.
\newblock \showarticletitle{Be Selfish and Avoid Dilemmas: Fork After
  Withholding (FAW) Attacks on Bitcoin}. In \bibinfo{booktitle}{{\em
  Proceedings of the 2017 ACM SIGSAC Conference on Computer and Communications
  Security}}. ACM, \bibinfo{pages}{195--209}.
\newblock


\bibitem[\protect\citeauthoryear{Kwon, Kim, Shin, and Kim}{Kwon
  et~al\mbox{.}}{2019}]%
        {kwon2019bitcoin}
\bibfield{author}{\bibinfo{person}{Yujin Kwon}, \bibinfo{person}{Hyoungshick
  Kim}, \bibinfo{person}{Jinwoo Shin}, {and} \bibinfo{person}{Yongdae Kim}.}
  \bibinfo{year}{2019}\natexlab{}.
\newblock \showarticletitle{{Bitcoin vs. Bitcoin Cash: Coexistence or Downfall
  of Bitcoin Cash?}}
\newblock \bibinfo{journal}{{\em arXiv preprint arXiv:1902.11064\/}}
  (\bibinfo{year}{2019}).
\newblock


\bibitem[\protect\citeauthoryear{Laszka, Johnson, and Grossklags}{Laszka
  et~al\mbox{.}}{2015}]%
        {laszka2015bitcoin}
\bibfield{author}{\bibinfo{person}{Aron Laszka}, \bibinfo{person}{Benjamin
  Johnson}, {and} \bibinfo{person}{Jens Grossklags}.}
  \bibinfo{year}{2015}\natexlab{}.
\newblock \showarticletitle{{When bitcoin mining pools run dry}}. In
  \bibinfo{booktitle}{{\em International Conference on Financial Cryptography
  and Data Security}}. Springer, \bibinfo{pages}{63--77}.
\newblock


\bibitem[\protect\citeauthoryear{Lewenberg, Bachrach, Sompolinsky, Zohar, and
  Rosenschein}{Lewenberg et~al\mbox{.}}{2015}]%
        {lewenberg2015bitcoin}
\bibfield{author}{\bibinfo{person}{Yoad Lewenberg}, \bibinfo{person}{Yoram
  Bachrach}, \bibinfo{person}{Yonatan Sompolinsky}, \bibinfo{person}{Aviv
  Zohar}, {and} \bibinfo{person}{Jeffrey~S Rosenschein}.}
  \bibinfo{year}{2015}\natexlab{}.
\newblock \showarticletitle{{Bitcoin mining pools: A cooperative game theoretic
  analysis}}. In \bibinfo{booktitle}{{\em Proceedings of the 2015 International
  Conference on Autonomous Agents and Multiagent Systems}}. International
  Foundation for Autonomous Agents and Multiagent Systems,
  \bibinfo{pages}{919--927}.
\newblock


\bibitem[\protect\citeauthoryear{Luu, Saha, Parameshwaran, Saxena, and
  Hobor}{Luu et~al\mbox{.}}{2015}]%
        {luu2015power}
\bibfield{author}{\bibinfo{person}{Loi Luu}, \bibinfo{person}{Ratul Saha},
  \bibinfo{person}{Inian Parameshwaran}, \bibinfo{person}{Prateek Saxena},
  {and} \bibinfo{person}{Aquinas Hobor}.} \bibinfo{year}{2015}\natexlab{}.
\newblock \showarticletitle{On {P}ower {S}plitting {G}ames in {D}istributed
  {C}omputation: {T}he {C}ase of {B}itcoin {P}ooled {M}ining}. In
  \bibinfo{booktitle}{{\em Computer Security Foundations Symposium (CSF)}}.
  IEEE.
\newblock


\bibitem[\protect\citeauthoryear{Merkle}{Merkle}{1980}]%
        {merkle1980protocols}
\bibfield{author}{\bibinfo{person}{Ralph~C Merkle}.}
  \bibinfo{year}{1980}\natexlab{}.
\newblock \showarticletitle{Protocols for {P}ublic {K}ey {C}ryptosystems.}. In
  \bibinfo{booktitle}{{\em Symposium on Security and privacy}}. IEEE.
\newblock


\bibitem[\protect\citeauthoryear{Nayak, Kumar, Miller, and Shi}{Nayak
  et~al\mbox{.}}{2016}]%
        {nayak2016stubborn}
\bibfield{author}{\bibinfo{person}{Kartik Nayak}, \bibinfo{person}{Srijan
  Kumar}, \bibinfo{person}{Andrew Miller}, {and} \bibinfo{person}{Elaine Shi}.}
  \bibinfo{year}{2016}\natexlab{}.
\newblock \showarticletitle{Stubborn {M}ining: {G}eneralizing {S}elfish
  {M}ining and {C}ombining with an {E}clipse {A}ttack}. In
  \bibinfo{booktitle}{{\em European Symposium on Security and Privacy}}. IEEE.
\newblock


\bibitem[\protect\citeauthoryear{Osborne}{Osborne}{2004}]%
        {osborne2004introduction}
\bibfield{author}{\bibinfo{person}{Martin~J Osborne}.}
  \bibinfo{year}{2004}\natexlab{}.
\newblock \bibinfo{booktitle}{{\em {An Introduction to Game Theory}}}.
\newblock \bibinfo{publisher}{Oxford university press New York}.
\newblock


\bibitem[\protect\citeauthoryear{Pete {R}izzo}{Pete {R}izzo}{2018}]%
        {doublespend}
Pete {R}izzo \bibinfo{year}{2018}\natexlab{}.
\newblock \bibinfo{title}{{D}ouble {S}pending {R}isk {R}emains {A}fter {J}uly
  4th {B}itcoin {F}ork}.
\newblock
  \bibinfo{howpublished}{https://www.coindesk.com/double-spending-risk-bitcoin-network-fork/}.
    (\bibinfo{year}{2018}).
\newblock
\newblock
\shownote{[Online; accessed 30-April-2018].}


\bibitem[\protect\citeauthoryear{{Pool Distribution}}{{Pool
  Distribution}}{2018}]%
        {pools2}
{Pool Distribution} \bibinfo{year}{2018}\natexlab{}.
\newblock \bibinfo{title}{{Pool Distribution}}.
\newblock \bibinfo{howpublished}{\url{https://btc.com/stats/pool}}.
  (\bibinfo{year}{2018}).
\newblock
\newblock
\shownote{[Online; accessed 2-May-2018].}


\bibitem[\protect\citeauthoryear{{Proof}}{{Proof}}{2018}]%
        {blockchain}
{Proof} \bibinfo{year}{2018}\natexlab{}.
\newblock \bibinfo{title}{{Proof of Work}}.
\newblock
  \bibinfo{howpublished}{\url{https://en.bitcoin.it/wiki/Proof_of_work}}.
  (\bibinfo{year}{2018}).
\newblock
\newblock
\shownote{[Online; accessed 30-April-2018].}


\bibitem[\protect\citeauthoryear{Rosenfeld}{Rosenfeld}{2011}]%
        {rosenfeld2011analysis}
\bibfield{author}{\bibinfo{person}{Meni Rosenfeld}.}
  \bibinfo{year}{2011}\natexlab{}.
\newblock \showarticletitle{Analysis of {B}itcoin {P}ooled {M}ining {R}eward
  {S}ystems}.
\newblock \bibinfo{journal}{{\em arXiv preprint arXiv:1112.4980\/}}
  (\bibinfo{year}{2011}).
\newblock


\bibitem[\protect\citeauthoryear{Sapirshtein, Sompolinsky, and
  Zohar}{Sapirshtein et~al\mbox{.}}{2015}]%
        {sapirshtein2015optimal}
\bibfield{author}{\bibinfo{person}{Ayelet Sapirshtein},
  \bibinfo{person}{Yonatan Sompolinsky}, {and} \bibinfo{person}{Aviv Zohar}.}
  \bibinfo{year}{2015}\natexlab{}.
\newblock \showarticletitle{Optimal {S}elfish {M}ining {S}trategies in
  {B}itcoin}.
\newblock \bibinfo{journal}{{\em arXiv preprint arXiv:1507.06183\/}}
  (\bibinfo{year}{2015}).
\newblock


\bibitem[\protect\citeauthoryear{Tsabary and Eyal}{Tsabary and Eyal}{2018}]%
        {tsabary2018gap}
\bibfield{author}{\bibinfo{person}{Itay Tsabary} {and} \bibinfo{person}{Ittay
  Eyal}.} \bibinfo{year}{2018}\natexlab{}.
\newblock \showarticletitle{The Gap Game}. In \bibinfo{booktitle}{{\em
  Proceedings of the 2018 ACM SIGSAC Conference on Computer and Communications
  Security}}. ACM, \bibinfo{pages}{713--728}.
\newblock


\bibitem[\protect\citeauthoryear{{ViaBTC}}{{ViaBTC}}{2018}]%
        {viabtc}
{ViaBTC} \bibinfo{year}{2018}\natexlab{}.
\newblock \bibinfo{title}{{{ViaBTC}}}.
\newblock \bibinfo{howpublished}{\url{https://pool.viabtc.com/}}.
  (\bibinfo{year}{2018}).
\newblock
\newblock
\shownote{[Online; accessed 3-May-2018].}


\bibitem[\protect\citeauthoryear{Yoo, Kim, Joy, and Gerla}{Yoo
  et~al\mbox{.}}{2018}]%
        {yoo2018promoting}
\bibfield{author}{\bibinfo{person}{Seunghyun Yoo}, \bibinfo{person}{Seungbae
  Kim}, \bibinfo{person}{Joshua Joy}, {and} \bibinfo{person}{Mario Gerla}.}
  \bibinfo{year}{2018}\natexlab{}.
\newblock \showarticletitle{Promoting Cooperative Strategies on Proof-of-Work
  Blockchain}. In \bibinfo{booktitle}{{\em 2018 International Joint Conference
  on Neural Networks (IJCNN)}}. IEEE, \bibinfo{pages}{1--8}.
\newblock


\end{thebibliography}

\appendix

\section{Appendix}

\subsection{Proof of Theorem~\ref{thm:nash_stage}}
\label{sec:proof5.1}

First, we prove \eqref{eq:stage_result} by showing that other actions
that at least one pool has positive BWH infiltration power cannot be a
Nash equilibrium. To this end, we first consider the action profile, say
$(0,b_i),(0,b_{-i}),$ where both pools execute BWH attacks. In this
case, it is easy to see that 
Pool$_i$ can increase its payoff by executing the
FAW attack rather than the BWH attack as follows:
$$U_i((0,b_i),(0,b_{-i}))<U_i((f_i,0),(0,b_{-i})), \quad \text{ if } f_i=b_i.$$
Second, consider the case when either of pools execute the BWH attack,
i.e., $((f_i,0),(0,b_{-i}))$ or $((0,b_i),(f_{-i},0)),$ which can be
similarly shown.

We now prove \eqref{eq:stage1} and \eqref{eq:stage2}, for which we
expand the FAW attack game in \cite{kwon2017selfish} by defining the
following mapping from the action profile $a^F$ in the FAW attack game to that
in the FAW-BWH attack game: 
\begin{equation*}
a^F = (f_i, f_{-i}) \longmapsto \text{ext}(a^F) = \Big((f_i,0), (f_{-i},0)\Big).
\end{equation*}
Lemma~\ref{lem:con} states a necessary and sufficient condition to be a
Nash equilibrium in the FAW-BWH game, whose proof is presented at
the end of this section. 
\begin{lemma}
\label{lem:con}
The action profile $a^F$ is a Nash equilibrium in the FAW game, if and only if
the action profile $\text{ext}(a^F)$ is a Nash equilibrium in the
FAW-BWH game. 
\end{lemma}
The above lemma is significantly convenient in that the Nash
equilibrium in the FAW-BWH game can be easily characterized by
the earlier result in the FAW game. 
Then, from \cite{kwon2017selfish}, the FAW game has a unique Nash
equilibrium $a^{F}_{\text{ne}}=(f_i^\star,f_{-i}^\star)$ which satisfies:
\begin{align*}
U^{FF}(f_i^\star,f_{-i}^\star)>0, \ U^{FF}(f_{-i}^\star,f_{i}^\star)<0&
\quad  \text{if} \quad \alpha_i > \alpha_{-i}, \cr
U^{FF}(f_i^\star,f_{-i}^\star) =  U^{FF}(f_{-i}^\star,f_{i}^\star) =0 &
\quad \text{if} \quad \alpha_i = \alpha_{-i}.
\end{align*}
Therefore, by Lemma~\ref{lem:con}, the action profile $\text{ext}(a^F_{\text{ne}})$ is
the Nash equilibrium, satisfying \eqref{eq:stage1} and \eqref{eq:stage2}. 
This completes the proof. 

\begin{proof}[Proof of Lemma~\ref{lem:con}]
  If an action profile $a^F$ is not a Nash equilibrium in the
  FAW game, the action profile $\text{ext}(a^F)$ is trivially not a Nash
  equilibrium in the FAW-BWH game.  Then we consider that an
  action profile $a^F=(f_i,f_{-i})$ is a Nash equilibrium in the FAW game.  Moreover, we assume that there exists
  $b'_i$ such that
$$U_i(\text{ext}(a^F))<U_i((0,b'_i),(f_{-i},0)).$$
In this case, the following
$$U_i(\text{ext}(a^F))<U_i((0,b'_i),(f_{-i},0))<U_i((f'_i,0),(f_{-i},0)) \text{ if } f'_i=b'_i$$
holds, and it is contradiction because $a^F$ is a Nash
equilibrium in the FAW game.  As a result, if an action
profile $a^F$ is a Nash equilibrium in the FAW game,
then $\text{ext}(a^F)$ is also a Nash equilibrium in the FAW-BWH game.
\end{proof}

\subsection{Proof of Theorem~\ref{thm:main}}
\label{sec:proofmain}
First, we prove that $(\text{ARS}_K, \text{ARS}_K)$ is a subgame perfect Nash equilibrium.
To prove this, we use the popular one-deviation
property which is a necessary and sufficient condition for subgame
perfect Nash equilibrium (SPNE). If {\em one-deviation property} (ODP) is
satisfied, no player can increase its payoff by changing its action at
the start of any subgame given the remainder of the player's own
strategy and the other players' strategies.

  \begin{theorem}[\cite{osborne2004introduction}] 
    \label{thm:odp}
For any
    infinitely repeated game with a discount factor $\delta <1,$ a
    strategy is a SPNE, if and only if it
    satisfies the one-deviation property.
\end{theorem}

Without loss of generality, to prove the ODP of (ARS$_K$, ARS$_K$), let Pool$_2$
deviates from ARS$_K$, i.e., attempts to change its action at the start of
any subgame, but all the actions of Pool$_1$ and the actions of Pool$_2$
after deviation follow ARS$_K$. Let time $0$ be the time when the deviation
occurs. Then, we have the following four cases of the subgame.  {\em
  (i)} $(\stai, \stami)$ is $({\tt G,G})$ at stage 0 (i.e., the start of
subgames), {\em (ii)} $(\stai, \stami)$ is $({\tt G,B})$ at stage 0,
{\em (iii)} $(\stai, \stami)$ is $({\tt B,G})$ at stage 0, and {\em (iv)}
$(\stai, \stami)$ is $({\tt B,B})$ at stage 0.

\smallskip
\noindent{{\bf Case (i)}}:   
In this case, the two pools continue to cooperate from stage 0, and their total payoffs $\mathcal{U}_1$ and $\mathcal{U}_2$ in the subgames are 0.
If Pool$_2$ executes an attack at the start of the subgame (i.e., $a_2^0 \not = \bm{\overline{0}}$), 
Pool$_1$ (which follows ARS$_K$) retaliates with infiltration power
$a_{1,ars}^1$ at stage 1, such that
\begin{equation}
U_2(\bm{\overline{0}},a_2^0)+U_2(a_{1,ars}^1,\bm{\overline{0}})<U_2(\bm{\overline{0}},\bm{\overline{0}}). \label{eq:case1_ars}
\end{equation}
Note that the value of $U_2(a_{1,ars}^1,\bm{\overline{0}})\leq 0.$
After stage 1, cooperation would be restored, thus Pool$_2$'s total
payoff $\mathcal{U}_2$ becomes $$U_2(\bm{\overline{0}},a_2^0)+\delta U_2(a_{1,ars}^1,\bm{\overline{0}}),$$ 
when the pool deviates from ARS$_K$ at the start of the subgame. 
This implies that, for any $a_2^0$, satisfying the ODP for Pool$_2$ requires the following condition: 
$$U_2(\bm{\overline{0}},a_2^0)+\delta
U_2(a_{1,ars}^1,\bm{\overline{0}}) \leq U_2(\bm{\overline{0}},\bm{\overline{0}}),$$
i.e., 
\begin{equation}
\delta \geq
    \frac{U_2(\bm{\overline{0}},\bm{\overline{0}})-U_2(\bm{\overline{0}},a_2^0)}{
    U_2(a_{1,ars}^1,\bm{\overline{0}})} \quad\text{  if  }\quad U_2(a_{1,ars}^1,\bm{\overline{0}})\not=0.\label{eq:case1}
\end{equation}
If $a^1_{1,ars}$ comes from $IP_{\text{faw}}$, not $IP_{\text{bwh}}$, RHS of \eqref{eq:case1} is less than $K$.
If $a^1_{1,ars}$ comes from $IP_{\text{bwh}}$, RHS of \eqref{eq:case1} is in a compact set 
$$\left\{\text{\bf Retaliate}_0\left(\alpha_1, \bm{\overline{0}},\alpha_{2},a_2^0, \bm{\overline{0}}\right)\big|\,0\leq f_2^0\leq \alpha_2 \text{ or } 0\leq b_2^0\leq \alpha_2\right\},$$ which is 
a piecewise continuous image of a compact space and includes only values less than 1.
Therefore, the maximum value of RHS of Eq.~\eqref{eq:case1} is less than $1$; i.e., 
\begin{equation}
    \max_{\substack{a^0_2\\U_2(a_{1,ars}^1,\bm{\overline{0}})\neq 0}}\frac{U_2(\bm{\overline{0}},\bm{\overline{0}})-U_2(\bm{\overline{0}},a_2^0)}{
    U_2(a_{1,ars}^1,\bm{\overline{0}})}<1.
    \label{eq:deltamax}
\end{equation}
Moreover, $U_2(a_{1,ars}^1,\bm{\overline{0}})$ is an increasing function of $K$ (i.e., $|U_2(a_{1,ars}^1,\bm{\overline{0}})|$ is a decreasing function of $K$), which implies that LHS of \eqref{eq:deltamax} is also an increasing function of $K$. 
This is because the FAW attack gives not only the retaliator (Pool$_1$) but also the deviator (Pool$_2$) more reward when compared with the BWH attack. 
In addition, for given $\alpha_1$ and $a_{2}^0$, when $\alpha_2$ increases, $U_2(\bm{\overline{0}},a_{2}^0)$ and $U_2(a_{1,ars}^1,\bm{\overline{0}})$ respectively increase and decrease (refer to \eqref{eq:faw_payoff} and \eqref{eq:bwh_payoff}). 
Again note that $U_2(a_{1,ars}^1,\bm{\overline{0}})\leq 0$ if $U_2(\bm{\overline{0}},a_{2}^0)\geq 0.$ 
Therefore, LHS of \eqref{eq:deltamax} is also an increasing function of $\alpha_2$ for given $\alpha_1$. 

\smallskip
\noindent{\bf Case (ii)}:  
In this case, Pool$_1$ retaliates against Pool$_2$, and Pool$_2$ cooperate at stage 0.
Let Pool$_1$ following ARS$_K$ retaliate with the action $a_{1,ars}^0$ at stage 0.
Then, when Pool$_2$ always follows ARS$_K$ in the subgames, 
Pool$_2$'s total payoff $\mathcal{U}_2 = U_2(a_{1,ars}^0,\bm{\overline{0}}).$ 
If Pool$_2$ deviates from ARS$_K$, executing an attack with an action $a_2^0
\neq \bm{\overline{0}},$ 
Pool$_1$ turns out to retaliate with an action $a_{1,ars}^1$ at stage 1,
such that \begin{equation*}
U_2(a_{1,ars}^0, a_2^0)+U_2(a_{1,ars}^1,\bm{\overline{0}})<U_2(a_{1,ars}^0,\bm{\overline{0}}), %\label{eq:case2_ars}
\end{equation*}
in which case Pool$_2$'s total payoff $\mathcal{U}_2$ is $$U_2(a_{1,ars}^0,a_2^0)+\delta U_2(a_{1,ars}^1,\bm{\overline{0}}).$$ 
To satisfy ODP for Pool$_2,$ the condition $\mathcal{U}_2
\leq U_2(a_{1,ars}^0,\bm{\overline{0}})$ should be satisfied, i.e., 
\begin{equation}
\delta \geq \frac{U_2(a_{1,ars}^0,\bm{\overline{0}})-U_2(a_{1,ars}^0,a_2^0)}{ U_2(a_{1,ars}^1,\bm{\overline{0}})}.\label{eq:case2}
\end{equation}
Similar to the Case \textit{(i)}, the following is satisfied:
\begin{equation}
\max_{\substack{a^0_2\\U_2(a_{1,ars}^1,\bm{\overline{0}})\neq 0}}{\frac{U_2(a_{1,ars}^0,\bm{\overline{0}})-U_2(a_{1,ars}^0,a_2^0)}{ U_2(a_{1,ars}^1,\bm{\overline{0}})}<1}\label{eq:deltamax2}
\end{equation}
Also, LHS of \eqref{eq:deltamax2} is an increasing function of $K$ and $\alpha_2$ for given $\alpha_1$. 

\smallskip
\noindent{\bf Case (iii)} and {\bf Case (iv)}:  We omit
these, because they can be treated similarly to Cases {\em (i)} and {\em (ii)}. 

As a result, we can set $F_K(\alpha_1,\alpha_2)$ as  $\max(X_1^\delta, X_2^\delta),$
where 
\begin{align*}
    X_1^\delta=\max_{\substack{a^0_2\\U_2(a_{1,ars}^1,\bm{\overline{0}})\not = 0}}&\left\{\frac{U_2(\bm{\overline{0}},\bm{\overline{0}})-U_2(\bm{\overline{0}},a_2^0)}{
    U_2(a_{1,ars}^1,\bm{\overline{0}})},
    \frac{U_2(a_{1,ars}^0,\bm{\overline{0}})-U_2(a_{1,ars}^0,a_2^0)}{ U_2(a_{1,ars}^1,\bm{\overline{0}})},\right.\\
    &\left.\frac{U_2(\bm{\overline{0}},a_{2,ars}^0)-U_2(\bm{\overline{0}},a_{2}^0)}{ U_2(a_{1,ars}^1,\bm{\overline{0}})}, \frac{U_2(\bm{\overline{0}},\bm{\overline{0}})-U_2(\bm{\overline{0}},a_2^0)}{U_2(a_{1,ars}^1,\bm{\overline{0}})}\right\}
\end{align*}
and 
\begin{align*}
    X_2^\delta=\max_{\substack{a^0_1\\U_1(\bm{\overline{0}},a_{2,ars}^1)\not = 0}}&\left\{\frac{U_1(\bm{\overline{0}},\bm{\overline{0}})-U_1(a_1^0,\bm{\overline{0}})}{
    U_1(\bm{\overline{0}},a_{2,ars}^1)},
    \frac{U_1(a_{1,ars}^0,\bm{\overline{0}})-U_1(a_1^0, a_{2,ars}^0)}{ U_1(\bm{\overline{0}},a_{2,ars}^1)},\right.\\
    &\left.\frac{U_1(a_{1,ars}^0,\bm{\overline{0}})-U_1(a_{1}^0,\bm{\overline{0}})}{ U_1(\bm{\overline{0}},a_{2,ars}^1)}, \frac{U_1(\bm{\overline{0}},\bm{\overline{0}})-U_1(a_1^0,\bm{\overline{0}})}{U_1(\bm{\overline{0}},a_{2,ars}^1)}\right\}.
\end{align*}
The values of $X_1^\delta$ and $X_2^\delta$ are always less than 1, and they are increasing functions of $K$. 
Moreover, $X_1^\delta$ is an increasing function of $\alpha_2$ for given $\alpha_1,$ and $X_1^\delta$ is a decreasing function of $\alpha_2$ for given $\alpha_1.$
Note that $X_1^\delta=X_2^\delta$ when $\alpha_1=\alpha_2.$
Therefore, $F_K(\alpha_1,\alpha_2)$ is an increasing function of $|\alpha_1-\alpha_2|$ for given $\alpha_1.$

Next, we prove that (ARS$_K$, ARS$_{K'}$) is a Nash equilibrium regardless of whether $K$ is equal to $K'$.
Indeed, this fact can be easily proved by using the above result.
When Pool$_1$ and Pool$_2$ follow ARS$_K$ and ARS$_{K'}$, respectively, they achieve cooperation, where their total payoff $\mathcal{U}_1$ and $\mathcal{U}_2$ are 0.
Because, in a Nash equilibrium, each strategy is the best response for given other strategies, we need to show 
\begin{align}
&\mathcal{U}_1(s_1, \text{ARS}_{K'})\leq \mathcal{U}_{1}(\text{ARS}_K, \text{ARS}_{K'})=0 \text{  if  } \,\,\forall s_1 \not\in \text{ARS},\label{eq:1_de}\\
&\mathcal{U}_2( \text{ARS}_{K}, s_2)\leq \mathcal{U}_{2}(\text{ARS}_K, \text{ARS}_{K'})=0 \text{  if  } \,\,\forall s_2 \not\in \text{ARS}.\label{eq:2_de}
\end{align}
Also, we note the following:
$$\mathcal{U}_{1}(\text{ARS}_K, \text{ARS}_{K'})=\mathcal{U}_{1}(\text{ARS}_{K'}, \text{ARS}_{K'})=0.$$
Therefore, for all $s_1\not\in$ ARS, \eqref{eq:1_de} would be satisfied because the strategy vector $(\text{ARS}_{K'}, \text{ARS}_{K'})$ is a subgame perfect Nash equilibrium.
In the same manner, \eqref{eq:2_de} holds.
This completes the proof.

\subsection{Non-Emptiness of infiltration set for BWH}
\label{sec:nonempty}
\begin{theorem}
The infiltration set $IP_{bwh}$ is non-empty. 
\label{thm:nonempty}
\end{theorem}

\begin{proof}

To prove this theorem, we show that there exists an infiltration power $b_i$ such that 
$$U_{-i} (\bm{\overline{0}},a_{-i}^{t-1})+U_{-i} ((0,b_i),\bm{\overline{0}})
< U_{-i} (\bm{\overline{0}},a_{-i,ars}^{t-1}),$$ 
\textit{and} there also exists an infiltration power $b_i$ such that 
$$U_{-i} (a_i^{t-1},a_{-i}^{t-1})+U_{-i} ((0,b_i),\bm{\overline{0}})< U_{-i} (a_i^{t-1},\bm{\overline{0}}),$$
because only one of $a_i^{t-1}$ and $a_{-i,ars}^{t-1}$ can have 
a value different to $\bm{\overline{0}}$.
To show this, we need to show that there is an infiltration power $b_i$ satisfying the following:
\begin{align}
U_{-i} ((0,b_i),\bm{\overline{0}})< 
\min_{a_i^{t-1},a_{-i}^{t-1},a_{-i,ars}^{t-1}} \Big\{&
U_{-i} (a_i^{t-1},\bm{\overline{0}})-U_{-i} (a_i^{t-1},a_{-i}^{t-1}),
\cr & U_{-i} (\bm{\overline{0}},a_{-i,ars}^{t-1})-U_{-i} (\bm{\overline{0}},a_{-i}^{t-1})
\Big\},\label{eq:con_f}
\end{align} 
where the right-hand side of \eqref{eq:con_f} can have the minimum value when $a_{-i}^{t-1}$, and $a_{-i,ars}^{t-1}$ have the following forms, 
$(f_{-i},0)$ and $(0, b_{-i})$, respectively.
Moreover, for an arbitrary value $p$ in $[0,\alpha_i]$, 
$U_{-i} (a_i^{t-1},\bm{\overline{0}})-U_{-i} (a_i^{t-1},a_{-i}^{t-1})$ in the case that
$a_i^{t-1}$ is $(p,0)$ is always less than or equal to that in the case that $a_i^{t-1}$ is $(0,p)$.

First, we consider the range of $a_i^{t-1}$ and $a_{-i,ars}^{t-1}$.
They are outputs of ARS executed by Pool$_i$ and Pool$_{-i}$, respectively.
Note that Pool$_1$, which calls {\bf Retaliate}, has \textit{good} standing. 
Therefore, $a_i$ should be $\bm{\overline{0}}$ or an output of \textbf{Retaliate}.
Also, $a_{-i,ars}$ is $\bm{\overline{0}}$ or an output of {\bf Retaliate} called by Pool$_{-i}$, by its definition. 
Next, we assume that $IP_{\text{bwh}}$ in {\bf Retaliate} of Pool$_i$ and Pool$_{-i}$ includes $M_{i}^B$ and $M_{-i}^B$, respectively.
Then the element $b_i$ of $a_i^{t-1}$ and element $b_{-i}$ of $a_{-i,ars}$ would be less than or equal to $M_i^{B}$ and $M_{-i}^B$, respectively.
Henceforth, $IP_{\text{bwh}}$ for Pool$_i$ and Pool$_{-i}$ is expressed as $IP_{\text{bwh},i}$ and $IP_{\text{bwh},-i}$, respectively, unless confusion arises.
Similarly, $IP_{\text{faw}}$ would be also expressed as described above.
Moreover, if $IP_{\text{faw},i}$ is a non-empty set, the set includes as follows.
\begin{equation} 
\frac{\sqrt{(1-\alpha_1)\alpha_2^2+(\alpha_1-\alpha_1^2)\alpha_2}-\alpha_2}{1-\alpha_1-\alpha_2},
\label{eq:max}
\end{equation}
which is an infiltration power $f_i$ that minimizes $U_2((f_i,0),\bm{\overline{0}})$.
Therefore, it is sufficient to consider the range of $a_i^{t-1}$ as 
$$\{(f_i,0)\,|\, 0\leq f_i \leq \max\{M_i^{B}, Eq.~\eqref{eq:max}\}\}$$
in order to show that there always exists 
an infiltration power $b_i$, which satisfies \eqref{eq:con_f}, because of the fact that
$U_{-i} (a_i^{t-1},\bm{\overline{0}})-U_{-i} (a_i^{t-1},a_{-i}^{t-1})$ in the case that
$a_i$ is $(p,0)$ is always less than or equal to that in the case that $a_i$ is $(0,p)$, for an arbitrary value $p$ in $[0,\alpha_i]$.

We define $\mathcal{F}$ as 
\begin{align*}
\min_{a_i^{t-1},a_{-i}^{t-1},a_{-i,ars}^{t-1}}\Big \{&U_{-i} (a^{t-1}_i,\bm{\overline{0}})-U_{-i} (a^{t-1}_i,a^{t-1}_{-i}),\cr 
&U_{-i} (\bm{\overline{0}},a_{-i,ars}^{t-1})-U_{-i} (\bm{\overline{0}},a_{-i}^{t-1})\Big\}-U_{-i} ((0,M_i^{B}),\bm{\overline{0}}).
\end{align*} 
Under the assumption that $IP_{\text{bwh},i}$ includes $M_i^B$, 
we investigated $\mathcal{F}$ by varying Pool$_i$'s size and Pool$_{-i}$'s size.
If $\mathcal{F}$ is always positive, $IP_{\text{bwh},i}$ is a non-empty set, including $M_i^B$.
Fig.~\ref{fig:p_existence} represents the value of $\mathcal{F}$ 
under the assumption that $IP_{\text{bwh},i}$ includes $M_i^B$.
However, the figure shows that there are some cases that $\mathcal{F}$ is negative. 

\begin{figure}[ht]
\centering{
\subfloat[When assuming that $IP_{\text{bwh},i}$ includes $M_i^B$, 
this figure represents the value of $\mathcal{F}$.]{
\includegraphics[width=0.24\textwidth]{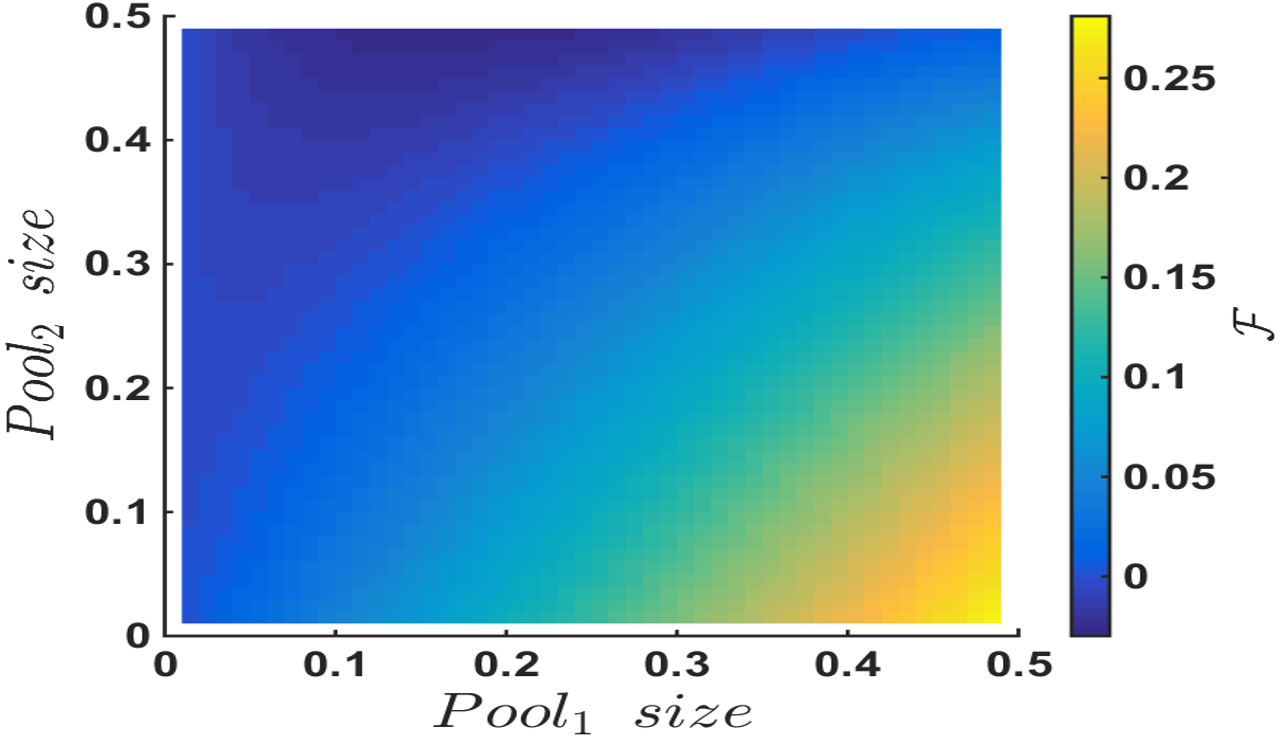}
\label{fig:p_existence}}
\subfloat[This figure represents the value of $\mathcal{F'}(k)$ for a value $k$, 
which makes $\mathcal{F'}(k)$ positive.]{
\includegraphics[width=0.24\textwidth]{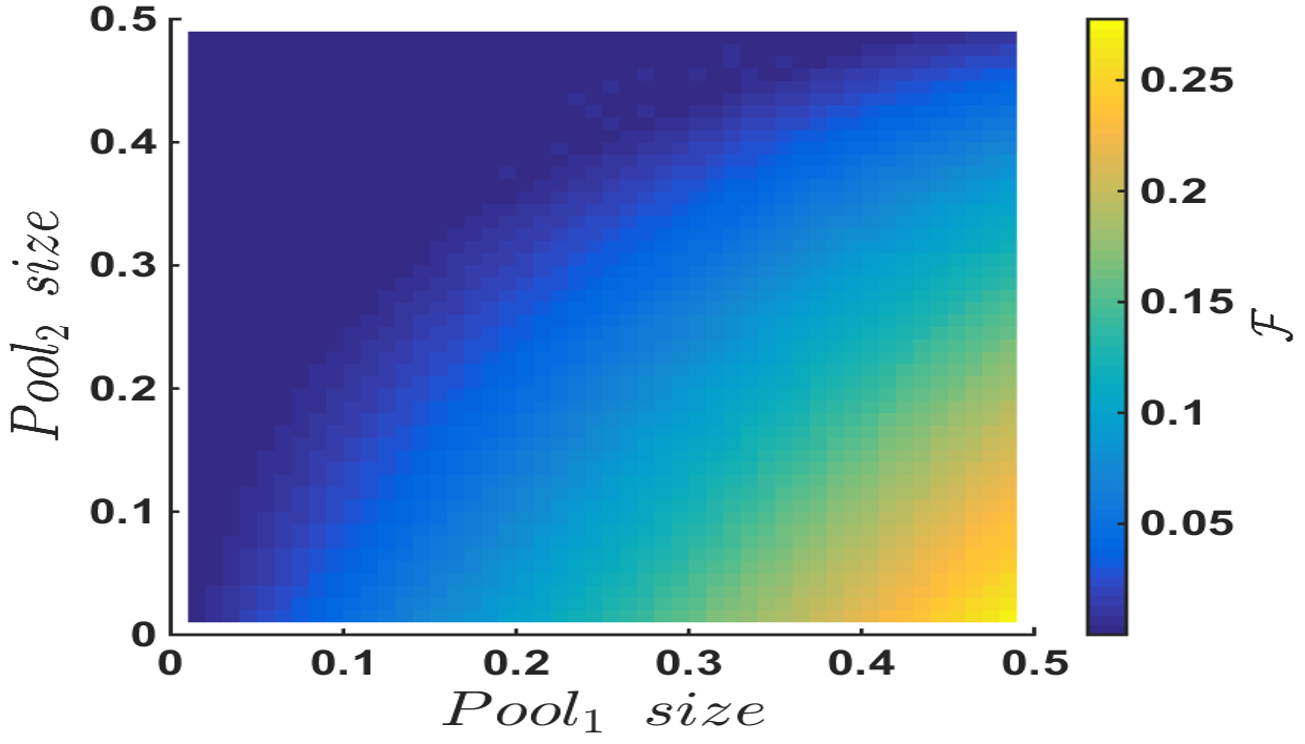}
\label{fig:existence}}}
\caption{The values of $\mathcal{F}$ and $\mathcal{F'}(k)$}
\label{fig:exist}
\end{figure}

In the case that $\mathcal{F}$ is negative, all elements of the set $IP_{\text{bwh},i}$ 
are greater than $M_i^B$, or the set $IP_{\text{bwh},i}$ is an empty set.
Therefore, we consider a function $\mathcal{F'}(k)$, which is defined as
\begin{align*}
\min_{a^{t-1}_i,a^{t-1}_{-i},a^{t-1}_{-i,ars}} \Big\{&U_{-i} (a_i^{t-1},\bm{\overline{0}})-U_{-i} (a^{t-1}_i,a^{t-1}_{-i}), \cr
&U_{-i} (\bm{\overline{0}},a^{t-1}_{-i,ars})-U_{-i} (\bm{\overline{0}},a^{t-1}_{-i})\Big\}-U_{-i} ((0,k),\bm{0}),\notag
%$}
\end{align*} 
and an element $f_i$ of $a_{i}^{t-1}$ ranges from 0 to $k$.
When $k$ is $M_i^B$, the function $\mathcal{F'}(k)$ is equal to $\mathcal{F}$.
In the case that $\mathcal{F}$ is negative, $k$ should be greater than $M_i^B$
to ensure that $\mathcal{F'}(k)$ is positive.
We find the value of $k$ such that $\mathcal{F'}(k)$ is positive, considering only $k$ equal to or greater than $M_i^B$.
For example, even though $k$ is less than $M_i^B$ in the case that  
$\mathcal{F}$ is positive, we regard the value of $k$ as $M_i^B$.
Moreover, in this case, $b_{-i}$ in $a_{-i,ars}^{t-1}$ ranges from $0$ to $k'$, which corresponds to $k$ when 
Pool$_{-i}$ calls \textbf{Retaliate} for retaliation against Pool$_i$.

Fig.~\ref{fig:existence} represents $\mathcal{F'}(k)$ for $k$, which makes $\mathcal{F'}(k)$ positive.
We can see that Fig.~\ref{fig:existence} represents positive values of $\mathcal{F'}(k)$ in all cases.
It also means that there always exists $k$, which makes $\mathcal{F'}(k)$ positive.
As a result, $IP_{\text{bwh},i}$ is a non-empty set. 
\end{proof}

\subsection{Number of Blocks Found by Attackers}
\label{subsec:number}
\begin{theorem}
  Let $N_{\text{faw}}$ and $N_{\text{bwh}}$ be the number of blocks
  found by attackers for FAW and BWH attacks, respectively. Then, the
  following hold:
  \begin{align*}
    N_{\text{faw}} &\sim Geo\,\,\left(\frac{(1-\gamma)\alpha}{\beta+\gamma\alpha(1-\alpha-\beta)+(1-\gamma)\alpha}\right)\cr
    N_{\text{bwh}} &\sim Geo\,\,\left(\frac{(1-\gamma)\alpha}{\beta+(1-\gamma)\alpha}\right).
  \end{align*}
\label{thm:geo}
\end{theorem}

\begin{proof}
  We consider the FAW attack, where it is straightforward to see that  
  $P$ has an exponential distribution with rate parameter
  $\frac{(1-\gamma)\alpha}{1-\gamma\alpha}$, and for a \textit{given}
  value of $P$, $N$ has a Poisson distribution with a parameter
  $\frac{P(\beta+\gamma\alpha(1-\alpha-\beta))}{1-\gamma\alpha}$.  Then,
  the following holds:
\begin{equation}
\resizebox{\hsize}{!}{
$\begin{aligned}
&\mathsf{Pr}(N) = \int_{0}^{\infty}\mathsf{Pr}(N,P)dP=\int_{0}^{\infty}\mathsf{Pr}(N|P)\mathsf{Pr}
(P)dP\notag
\\=&\int_{0}^{\infty}\frac{(P(\beta+\gamma\alpha(1-\alpha-\beta)))^Ne^{-\frac{P(\beta+\gamma\alpha(1-\alpha-\beta))}{1-\gamma\alpha}}}{N!(1-\gamma\alpha)^N}\cdot \frac{(1-\gamma)\alpha e^{-\frac{P(1-\gamma)\alpha}{1-\gamma\alpha}}}{1-\gamma\alpha}\notag\\
=&\left(\frac{\beta+\gamma\alpha(1-\alpha-\beta)}{1-\gamma\alpha}\right)^N \cdot\frac{(1-\gamma)\alpha}{1-\gamma\alpha}\int_{0}^{\infty}\frac{P^N}{N!}e^{\frac{-P(\beta+\gamma\alpha(1-\alpha-\beta)+(1-\gamma)\alpha)}{1-\gamma\alpha}}\notag\\
=&\left(\frac{\beta+\gamma\alpha(1-\alpha-\beta)}{1-\gamma\alpha}\right)^N \cdot\frac{(1-\gamma)\alpha}{1-\gamma\alpha}\cdot\left(\frac{1-\gamma\alpha}{\beta+\gamma\alpha(1-\alpha-\beta)+(1-\gamma)\alpha}\right)^{N+1}\notag\\
=&\frac{(\beta+\gamma\alpha(1-\alpha-\beta))^N (1-\gamma)\alpha}{(\beta+\gamma\alpha(1-\alpha-\beta)+(1-\gamma)\alpha)^{N+1}},\notag
\end{aligned}$}
\end{equation}
which implies that $N$ is a geometric random variable with parameter
$\frac{(1-\gamma)\alpha}{\beta+\gamma\alpha(1-\alpha-\beta)+(1-\gamma)\alpha}$
for FAW attack. The proof of BWH attack is similar to the above.   
\end{proof}

\subsection{Evasion of the Identification Methods Reducing the Variance}\label{sec:evade}

\noindent\textbf{Case 1: What if the attacker (the manager of the attacking pool) distributes the rewards gained from the victim pool to miners at a random time?} 
The attacker can pay the rewards gained from the victim pool to miners
at a random time. This behavior can reduce the variance in $Rd_p$.
However, if the victim (the manager of the victim pool) frequently changes the mole's account, the victim still observes a relatively large variance among rewards in the accounts.
Therefore, although identification is slackened by the attacker, the victim can still identify the attacker by observing the variance in reward densities.

\noindent\textbf{Case 2: What if the attacker distributes the rewards gained from the victim pool for one period $P$ in proportion to the number of shares submitted over several periods?}
To reduce the variance in $Rd_p$, the attacker may pay the rewards 
gained from the victim pool for one period $P$ to miners in proportion to
the number of shares submitted over several periods. 
In other words, even if a miner works in the attacker's pool for only one period, the miner can earn part of the rewards gained from the victim pool for several periods.
Indeed, the greater the number of considered periods for paying rewards, the smaller the variance in $Rd_p$ will be.
However, the number of periods cannot be large.  
This is because there are side effects: First, miners in the attacker's pool must wait for a long time to receive the total reward for their work. 
Second, some miners can suffer losses from the continuous changes in $\alpha$ and $\beta$.
When $\alpha$ or $\beta$ changes, the mean value of the
$\frac{N\gamma}{\beta+\gamma\alpha}$ term of reward density 
$Rd_p$, which is earned from the victim, as well as the infiltration ratio $\gamma$~\footnote{The optimal infiltration ratio depends on $\alpha$ and $\beta$.} would change. 
However, if the attacker distributes the rewards earned from the victim 
over large number of periods, many miners may receive part of the rewards gained from the victim for some periods even though they did not work during the corresponding periods.
Because the part of reward density $Rd_p$ earned from the victim changes over time, the more periods the attacker distributes the rewards gained from the victim, the more unfair the reward system becomes.
For these reasons, the number of periods cannot be large. When the number of periods is small, the victim can still perceive a relatively large variance in $Rd_p$ in the attacker's pool.

\noindent\textbf{Case 3: What if the attacker distributes only partial rewards gained from the victim pool?}
Another method for the attacker to reduce the variance of $Rd_p$ is to share not entire rewards gained from the victim.
Indeed, in order to distribute extra rewards to miners in the attacker's pool, the attacker must share most of the rewards gained from the victim with the miners.
For example, we consider that a pool, which possesses a computational power of 0.2, executes the FAW attack with an infiltration power of 0.005 against other pool with the same computational power (0.2). 
Then, the attacker can earn an extra reward of 0.48\%. 
Meanwhile, if the attacker does not share rewards earned from the victim with her miners, the miners suffer a loss of 2.01\%.
Therefore, the attacker should divide at least approximately 80.7\% ($\frac{2.01}{2.01+0.48}$) of the rewards earned from the victim pool with the miners, in order to prevent losses of her miners.
As a result, a mole in her pool can still observe a relatively large variance in reward density $Rd_p$.

%%% Local Variables:
%%% mode: latex
%%% TeX-master: "main"
%%% End:

\end{document}